\RequirePackage[-l2tabu,orthodox]{nag}
\documentclass
[11pt,letterpaper]
{article} 

\usepackage[notes=true,later=false,camera=false]{dtrt}
\usepackage[utf8]{inputenc}
\usepackage{ stmaryrd }
\usepackage{xspace,enumerate}
\usepackage[T1]{fontenc}
\usepackage[full]{textcomp}
\usepackage[american]{babel}
\usepackage{mathtools}

\renewcommand{\hat}[1]{\widehat{#1}}
\usepackage{amsthm}
\usepackage{thmtools}

\usepackage{tikz}
\usetikzlibrary{arrows.meta,calc,positioning,fit,backgrounds}
\usepackage{pdflscape}
\usepackage{xcolor}
\usepackage{array}
\usepackage{hyperref}

\usepackage[capitalise,nameinlink]{cleveref}

\usepackage{empheq}

\definecolor{citeblue}{HTML}{0055cc}
\hypersetup{
colorlinks=true,
urlcolor=citeblue,
linkcolor=NavyBlue,
citecolor=PineGreen,
linktocpage=true,
}
\renewcommand*\backref[1]{\ifx#1\relax \else (pg. #1) \fi}

\renewcommand{\tilde}{\widetilde}

\usepackage{paralist}
\usepackage{turnstile}
\usepackage[framemethod=TikZ]{mdframed}
\mdfsetup{frametitlealignment=\center}
\usepackage{tikz}
\usepackage{caption}
\DeclareCaptionType{Algorithm}
\usepackage{newfloat}

\newtheorem{theorem}{Theorem}[section]
\newtheorem{lemma}[theorem]{Lemma}
\newtheorem*{lemma*}{Lemma}

\newtheorem{proposition}[theorem]{Proposition}
\newtheorem{fact}[theorem]{Fact}

\theoremstyle{definition}
\newtheorem{definition}[theorem]{Definition}
\newtheorem*{definition*}{Definition}
\newtheorem{remark}[theorem]{Remark}

\usepackage[linesnumbered,ruled]{algorithm2e}
\crefname{lemma}{Lemma}{Lemmas}
\crefname{fact}{Fact}{Facts}
\crefname{theorem}{Theorem}{Theorems}
\crefname{mtheorem}{Theorem}{Theorems}
\crefname{itheorem}{Theorem}{Theorems}
\crefname{corollary}{Corollary}{Corollaries}
\crefname{claim}{Claim}{Claims}
\crefname{example}{Example}{Examples}
\crefname{algorithm}{Algorithm}{Algorithms}
\crefname{problem}{Problem}{Problems}
\crefname{definition}{Definition}{Definitions}
\crefname{equation}{Eq.}{Eq.}
\crefname{strategy}{Strategy}{Strategies}
\crefname{observation}{Observation}{Observations}

\Crefname{algocf}{Algorithm}{Algorithms}

\usepackage[
letterpaper,
top=1.2in,
bottom=1.2in,
left=1in,
right=1in]{geometry}
\usepackage{newpxtext} 
\usepackage{textcomp} 
\usepackage[scr=rsfso]{mathalfa}
\usepackage{bm} 

\usepackage{microtype}

\usepackage{footnotebackref}

\renewcommand{\bar}{\overline}

\newcommand{\cN}{\mathcal{N}}

\newcommand{\op}{\operatorname{op}}

\newcommand{\AND}{\operatorname{AND}}

\newcommand{\Tr}{\operatorname{Tr}}

\newcommand{\QSG}{\operatorname{QSG}}

\allowdisplaybreaks
\newcommand{\FormatAuthor}[3]{
\begin{tabular}{c}
#1 \\ {\small\texttt{#2}} \\ {\small #3}
\end{tabular}
}

\newcommand{\sfK}{\mathsf K}

\newcommand{\sfc}{\mathsf c}

\newcommand{\sfa}{\mathsf a}
\newcommand{\sfx}{\mathsf x}
\newcommand{\sfb}{\mathsf b}
\newcommand{\Num}{\mathsf N}
\newcommand{\vac}{f_{\bm{0}}}

\newcommand{\R}{{\mathbb R}}

\newcommand{\eps}{\varepsilon}

\newcommand{\cF}{{\mathcal F}}
\newcommand{\cE}{{\mathcal E}}

\newcommand{\1}{\mathbf{1}}

\newcommand{\tr}{\mathrm{tr}}
\newcommand{\Z}{\mathbb Z}

\newcommand{\C}{\mathbb C}

\newcommand{\cH}{\mathcal H}

\newcommand{\poly}{\mathrm{poly}}

\newcommand{\Id}{\operatorname{Id}}

\newcommand{\seq}{\subseteq}
\newcommand{\spn}{\operatorname{span}}
\newcommand{\wt}{\operatorname{wt}}

\DeclareMathOperator*{\E}{\mathbb{E}}

\newcommand{\cR}{\mathcal R}
\newcommand{\cK}{\mathcal K}

\newcommand{\cM}{\mathcal M}

\newcommand{\cL}{\mathcal L}

\renewcommand{\emptyset}{\varnothing}
\renewcommand{\geq}{\geqslant}
\renewcommand{\ge}{\geqslant}
\renewcommand{\leq}{\leqslant}
\renewcommand{\le}{\leqslant}
\renewcommand{\preceq}{\preccurlyeq}
\renewcommand{\succeq}{\succcurlyeq}
\renewcommand{\epsilon}{\varepsilon}

\newcommand{\ignore}[1]{}

\newcommand{\Unif}{\operatorname{Unif}}

\newcommand{\Ree}{\operatorname{Re}}

\newcommand{\cP}{\mathcal{P}}
\newcommand{\comp}{\operatorname{comp}}

\renewcommand{\iota}{\mathbf{i}}

\newcommand{\syk}{\operatorname{SYK}}
\begin{document}
\author{
 \FormatAuthor{Arpon Basu}{arpon.basu@princeton.edu}{Princeton University} 
 \FormatAuthor{Pravesh K. Kothari}{kothari@cs.princeton.edu}{Princeton University} 
  \FormatAuthor{Siddhant Midha}{siddhantm@princeton.edu}{Princeton University}
  }

\title{Sharp Bounds on Ground State Energy of the SYK Model}

\maketitle

\abstract{
We study the Sachdev-Ye-Kitaev (SYK) Hamiltonian $H_{\operatorname{SYK}}$ on $n$ Majorana modes with $k$-body interactions, and prove that $\mathbb{E}\|H_{\operatorname{SYK}}\|_{\operatorname{op}} = (1 - o(1))\cdot\sqrt{2n}/k$ for super-constant $k\leq o(\sqrt{n})$, where the expectation is over the disorder variables in the Hamiltonian. This confirms the predictions due to \cite{SYK_Garcia3} and answers a question posed in \cite{feng2019spectrum}. Our results extend to the sparse SYK Hamiltonian. As a corollary, we obtain that the dissipative quantum algorithm of~\cite{basso2024optimizing} provably computes the ground state energy of the SYK Hamiltonian up to an $O(1)$-multiplicative factor for all $k < \sqrt{n}/4$.

Our key technical idea is identifying an explicit, deterministic linear operator $\mathsf{x}$ such that a fixed quadratic form of $\mathsf{x}^{2\ell}$ exactly equals the expected trace moments of the SYK Hamiltonian for every $n$ and $k$. This linear operator can be naturally viewed as a \emph{twisted} model of bosons on the space of hyperedges of a hypergraph. The problem thus reduces to identifying the spectral edge of $\mathsf{x}$, which we show is dominated by the spectrum of a natural ${n \choose k}$-dimensional matrix from the \emph{Johnson} scheme and is straightforward to compute using known results. To show that our bound is sharp, we construct a witness state with a large quadratic form on $\mathsf{x}$ and transform it into a certificate of a lower bound on the largest quadratic form on $H_{\operatorname{SYK}}$.

Beyond the SYK Hamiltonian, our techniques also apply to the $k-$local quantum spin glass (QSG), yielding $\mathbb{E}\|H_{k-\operatorname{QSG}}\|_{\operatorname{op}}\leq O(\sqrt{n/k})$ for all $1\leq k\leq\sqrt{n}/2$, thus improving the standard $O(\sqrt{n})$ bound by a factor of $\sqrt{k}$.

}
\pagenumbering{arabic}
\section{Introduction}
We study the spectral edge of the \emph{Sachdev-Ye-Kitaev} (SYK) Hamiltonian \cite{SachdevY93,Kitaev15,MaldacenaS16,Sachdev24} in this paper. Let $n\geq k\geq 2$ be even integers, and let $\{\gamma_i\}_{1\leq i\leq n}$ be \emph{Majorana fermion operators}, satisfying the anti-commutation relations $\{\gamma_i, \gamma_j\} = 2\delta_{ij}\cdot\Id$. Then consider the SYK Hamiltonian given as 
\begin{equation}\label{eq:SYK-intro-def}
    H_{\operatorname{SYK}} := \binom{n}{k}^{-1/2}\cdot\sum_{S\in\binom{[n]}{k}}g_S\Gamma_S,
\end{equation}
where $\Gamma_S:= \iota^{k(k - 1)/2}\gamma_{i_1}\cdots\gamma_{i_k}$, $S = \{i_1 < \cdots < i_k\}\seq[n]$ is a subset of $[n]$ of size $k$, $\binom{[n]}{k}:= \{S\seq[n]: |S| = k\}$ is the collection of all $k$-sets in $[n]$, and $\bm g:= \{g_S\}_{S\in\binom{[n]}{k}}$ is a collection of i.i.d $\cN(0, 1)$ Gaussians, also referred to as the \emph{disorder variables}. The normalization $\binom{n}{k}^{-1/2}$ is chosen to ensure that $\E_{\bm g}[H_{\operatorname{SYK}}^2] = \Id$.\footnote{This is in contrast to the standard physics convention, as the energy is not extensive.}

The SYK model~\cite{SachdevY93,Kitaev15} has been extensively studied in high-energy~\cite{Kitaev15,MaldacenaS16,polchinski2016spectrum,rosenhaus2019introduction} and condensed-matter~\cite{SachdevY93,SYK_review} physics, with related models studied as far back as~\cite{syk_old1,syk_old2}. The model exhibits ``maximally chaotic'' behavior~\cite{Kitaev15,maldacena2016bound,MaldacenaS16}, yet remains analytically tractable in certain regimes~\cite{SYK_Berkooz2,SYK_Berkooz3}. The structure and complexity of (near) ground states of the SYK Hamiltonian arise naturally when studying the model as a candidate for establishing quantum speedup in simulation. This is because while there are efficient \emph{quantum} algorithms~\cite{hastings2022optimizing,basso2024optimizing}, there are indications that there may not be any efficient classical algorithms to compute useful descriptions of ground states of the model. The evidence for this latter claim is based on results showing that easy-to-compute candidates such as Gaussian states~\cite{HaldarGaussian,dingOptimizingSparseSYK2026} have energies significantly separated from the energy of states output by the above quantum algorithms. 

\paragraph{The Ground State Energy of the SYK Model:} Identifying the ground state energy (and relatedly, the spectral edge) of the SYK model itself has been an outstanding open question. Standard matrix concentration inequalities~\footnote{Gaussian Lipschitz concentration immediately implies a sharp concentration of $\|H_{\syk}\|$ around its mean. See \cref{thm:main-thm-subgaussian}.} imply that $\E_{\bm g}\|H_{\operatorname{SYK}}\|_{\operatorname{op}}\leq O(\sqrt{n})$ \cite{feng2019spectrum,hastings2022optimizing}. On the other hand, in~\cite{anschuetzStronglyInteractingFermions2025}, the authors prove a lower bound of $\E_{\bm g}\|H_{\operatorname{SYK}}\|_{\operatorname{op}} \geq \Omega(\sqrt{n}/k)$. Interestingly, a beautiful recent work gives an efficient quantum algorithm to compute a state with an energy matching the above lower bound up to a constant factor~\cite{basso2024optimizing}.\footnote{ \label{footnote:HO-eff-quant} For $k = 4$ \cite{hastings2022optimizing} gives a simpler quantum algorithm to synthesize a state $\rho:= \rho(H_{\operatorname{4-SYK}})$ that satisfies $\E_{\bm g}\Tr(\rho(H_{\operatorname{4-SYK}})H_{\operatorname{4-SYK}})\geq\Omega(\sqrt n)$.}

The natural upper bound and the lower bounds above have a multiplicative gap of $\Omega(k)$ between them. Researchers have proposed both numerical computations and heuristic arguments to argue that $\E_{\bm g}\|H_{\operatorname{SYK}}\|_{\operatorname{op}}$ should in fact be $\asymp \sqrt{n}/k$ \cite{SYK_Garcia3,hastings2022optimizing} matching the lower bound above. A sequence of works by Feng, Tian, and Wei~\cite{feng2019spectrum,feng2020spectrum,feng2021spectrum} made progress towards this question by understanding the \emph{shape}~(as opposed to the spectral edge)~of the spectrum of the SYK model. In particular, they prove that for $k\leq o(\sqrt{n})$, the SYK model has a Gaussian spectrum~\cite{feng2019spectrum}. Furthermore, for super-constant $k\leq o(\sqrt{n})$ \cite{SYK_Garcia3}, \cite[Section~2.9]{hastings2022optimizing} present fine-grained heuristic calculations that suggest that $\E_{\bm g}\|H_{\operatorname{SYK}}\|_{\operatorname{op}} \sim (1 - o_n(1))\cdot\sqrt{2n}/k$.

\paragraph{Our Work} In this work, we rigorously confirm the above predictions and identify the spectral edge of the SYK model. Our bounds are sharp down to the right constant for large $k$.
\begin{theorem}[Spectral Edge of SYK (See \cref{thm:syk-spectral-norm,thm:syk-spectral-norm-lower} for full version)]\label{thm:main-thm-complete}
    Let $n\geq k\geq 2$ be even integers such that $k^2/n < 1/16$. Consider the SYK Hamiltonian $H_{\operatorname{SYK}}$ as in \eqref{eq:SYK-intro-def}. Then we have 
    \[{\frac{\sqrt{2n}}{k}}\cdot\left(1 - O(1)\cdot\max\left\{k^{-1/2}, \frac{k^4}{n^2}\right\}\right) \leq \E_{\bm g}\|H_{\operatorname{SYK}}\|_{\operatorname{op}}\leq\frac{\sqrt{2n}}{k} + O(1).\]
    In particular, if $\omega_n(1)\leq k\leq o(\sqrt{n})$, then $\E_{\bm{g}} \lambda_{\max}(H_{\syk}),  \E_{\bm g}\|H_{\operatorname{SYK}}\|_{\operatorname{op}} = (1 - o_n(1))\cdot\sqrt{2n}/k$.
\end{theorem}

\begin{remark}
    \begin{enumerate}[(1)]
    \item Our results, by natural modifications to our arguments, actually imply the same bounds as above on the ground state energy $\lambda_{\max}(H_{\syk})$. Observe that the upper bound follows immediately since $\lambda_{\max}(H)\leq \|H\|_{\op}.$ Our lower bound argument with minor modifications actually yields that a quadratic form of $H_{\syk}$ is large and thus applies to $\lambda_{\max}(H_{\syk})$ and not just $\|H_{\syk}\|_{\op}$.\footnote{In particular, \cref{lem:transfer-virtual-to-real,lem:lipschitz-operator,lem:lower-bound-syk-general} hold directly upon replacing the operator norm $\|H_{\syk}\|_{\op}$ by $\lambda_{\max}(H_{\syk})$.} 
    \item While our result gives an estimate that is sharp as $k$ grows, we note that for small constant values of $k$, the multiplicative factor of $(\sqrt{2}/k)$ for $\sqrt{n}$ is \emph{not} tight. For $k = 2$, the SYK Hamiltonian is exactly solvable, and in particular \cite{MaldacenaS16,feng2019spectrum} showed that $\|H_{\operatorname{2-SYK}}\|_{\operatorname{op}} = (\frac{4\sqrt{2}}{3\pi} + o(1))\sqrt{n}\approx 0.6\sqrt{n}$ almost surely, while $\sqrt{2n}/2\approx 0.707\sqrt{n}$.
    \item For $k\geq 4$, \cite{feng2019spectrum} show that $\E_{\bm g}\lambda_{\max}(H_{\operatorname{SYK}})\leq\sqrt{\ln(2)\cdot n}\approx 0.832\sqrt{n}$. On the other hand, for $k\geq 4$ \cref{thm:main-thm-complete} implies that $\E_{\bm g}\lambda_{\max}(H_{\operatorname{SYK}})\leq\E_{\bm g}\|H_{\operatorname{SYK}}\|_{\operatorname{op}}\leq\frac{\sqrt{2n}}{4} + O(1)\leq 0.354\sqrt{n} + O(1)$, and thus our result outperforms that of \cite{feng2019spectrum} for all $k\geq 4$.
    \item For $k = 4$, \cite{hastings2022optimizing} give an \emph{algorithmic}\footnote{More precisely, they show the following result: Let $\R^{\binom{[n]}{4}}\ni\bm g\mapsto \mathbf{SoS}_8(\bm g)\in\R$ be the degree-$8$ Sum-of-Squares relaxation to estimate $\|H_{\operatorname{4-SYK}}\|_{\operatorname{op}}$. $\mathbf{SoS}_8(\cdot)$ can be computed in deterministic $\poly(n)$ time given $\bm g$, and furthermore, $\mathbf{SoS}_8(\bm g)\geq\|H_{\operatorname{4-SYK}}\|_{\operatorname{op}}$ \emph{pointwise} for all $\bm g\in\R^{\binom{[n]}{4}}$. Finally, $\E_{\bm g}\mathbf{SoS}_8(\bm g)\leq\sqrt{1 + \sqrt{6}}\cdot\sqrt n$.} proof that $\lambda_{\max}\left(H_{\operatorname{4-SYK}}\right)\leq\sqrt{1 + \sqrt{6}}\cdot\sqrt{n}\approx 1.857\sqrt{n}$ with probability $\geq 1 - o_n(1)$. 
    \item Very recently,~\cite{he2026freeenergyspectraledge} established the exact spectral edge for all \emph{fixed} even integers $k\geq 4$, showing that almost surely $\lim_{n\to\infty}\lambda_{\max}(H_{\operatorname{k-SYK}})/\sqrt{n} = \kappa_{k,\operatorname{SD}}$, where $\kappa_{k,\operatorname{SD}}$ is a constant satisfying $\sqrt{2/(k(k + 1))}\leq\kappa_{k,\operatorname{SD}}\leq\sqrt{2}/k$ for all even $k\geq 4$. Consequently, combining~\cite{he2026freeenergyspectraledge} along with \cref{thm:main-thm-random} determines the exact asymptotic location of $\lambda_{\max}(H_{\operatorname{k-SYK}})$ for all even $4\leq k\leq\sqrt{n}/4$. It is worth noting though that the techniques in \cite{he2026freeenergyspectraledge} do not handle the regime where $k\geq\omega_n(1)$ is super-constant, while our techniques work for all $2\leq k < \sqrt{n}/4$, and are very different from the techniques in~\cite{he2026freeenergyspectraledge}. Furthermore, when $\omega_n(1)\leq k\leq o(\sqrt{n})$ \cref{thm:main-thm-complete} obtains the correct prefactor. Furthermore, our results also apply to the sparsified SYK Hamiltonian, and much more generally to random Hamiltonians (such as quantum spin glasses), which we now describe.
    \end{enumerate}
\end{remark}

\paragraph{Spectral Edge of Sparse SYK Models} We show that our results extend to sparse variants of the SYK model. Formally, let $\cH \sim \binom{[n]}{k}$  be a random hypergraph with $m= |\cH| \ll {n \choose k}$ hyperedges. The sparse SYK Hamiltonian we consider is defined by: 
\[H^\cH_{\operatorname{SYK}}:= \frac{1}{\sqrt{|\cH|}}\sum_{S\in\cH}g_S\Gamma_S.\]

Understanding such sparse variants has been of interest recently ~\cite{xu2020sparse,jafferis2022traversable,HerasymenkoSHT23,dingOptimizingSparseSYK2026} as they may be easier to implement in practice. In this context, we care about whether relevant properties of the dense SYK model transfer to this ``sparsification.'' We show that the spectral edge identified in our main result above naturally extends to the sparse variants after appropriate normalization.

In the following (and throughout this paper), for any $m\geq 2$, a random $k$-uniform hypergraph $\cH$ with $m$ hyperedges is obtained by taking $m$ independent and uniform samples from $S\in\binom{[n]}{k}$. ~\footnote{If we write the sampled set of hyperedges $\cH = (S_1,S_2,\dots,S_m)$ with $S_i \sim\Unif\binom{[n]}{k}$ i.i.d., then the Gaussian disorder variables corresponding to $S_i$ and $S_j$ for $j\neq i$ are independent even though $S_i$ and $S_j$ might correspond to the same hyperedge.}

\begin{theorem}[Operator Norm of the Sparse Random SYK Hamiltonian (See \cref{thm:syk-spectral-norm-random,thm:syk-spectral-norm-lower-random} for full version)]\label{thm:main-thm-random}
    Let $n\geq k\geq 2$ be even integers, and suppose $k^2/n < 1/16$.  Let $\cH$ be a random $k$-uniform hypergraph with $m\geq 2^{Ck}n\log n$ edges, where $C > 0$ is some sufficiently large absolute constant. Let 
    \[H^\cH_{\operatorname{SYK}}:= \frac{1}{\sqrt{|\cH|}}\sum_{S\in\cH}g_S\Gamma_S\]
    be the SYK Hamiltonian corresponding to $\cH$. Then with probability $\geq 1 - o_n(1)$ over the randomness of $\cH$ we have 
    \[\frac{\sqrt{2n}}{k}\cdot
\left(
1-O(1)\cdot\max\left\{
k^{-1/4},
\frac{k^4}{n^2}
\right\}
\right)\leq\E_{\bm g}\|H^\cH_{\operatorname{SYK}}\|_{\operatorname{op}}\leq\frac{\sqrt{2n}}{k}\cdot\left(1 + O(1)\cdot\max\left\{\frac{k}{\sqrt{n}}, e^{-\Omega(k)}\right\}\right).\]
    In particular if $\omega(\log(n))\leq k\leq o(\sqrt{n})$ then $\E\|H^\cH_{\operatorname{SYK}}\|_{\operatorname{op}} = (1 - o_n(1))\cdot\sqrt{2n}/k$.
\end{theorem}
\begin{remark}
    Note that bounds of the above kind are not true for \emph{arbitrary} $k$-uniform hypergraphs: Indeed, consider $\cH:= \left\{[r]\cup S:S\in\binom{[n]\setminus[r]}{k - r}\right\}\seq\binom{[n]}{k}$, where we assume $\omega(\log(n))\leq k\leq o(\sqrt{n})$, and choose $r:= k - \lceil 3k/\log_2(n)\rceil$ to obtain $|\cH|\geq 2^{\Omega(k)}n\log(n)$. But at the same time we can also show that $\E_{\bm g}\|H^\cH_{\operatorname{SYK}}\|_{\operatorname{op}}\geq\Omega((\sqrt{n}/k)\cdot\log(n))$. Note that this happens simply because the ``effective arity'' of $\cH$ is only $k - r$, which is how the spectral norm scales too. 

    In fact, examining our proof more closely yields the following deterministic condition that a $k$-uniform hypergraph $\cH$ needs to satisfy so that we obtain tight bounds on $\E_{\bm g}\|H^\cH_{\operatorname{SYK}}\|_{\operatorname{op}}$ as in \cref{thm:main-thm-random}: Let $\cE:= |\cH|^{-1}\cdot((-1)^{|S\cap T|})_{S, T\in\cH}$ be the \emph{sign matrix} encoding the commutation structure of $\{\Gamma_S\}_{S\in\cH}$. Also write $u:= \1/\sqrt{|\cH|}\in\C^\cH, \Pi_{u^\perp}:= \Id - uu^*, \delta:= \|\Pi_{u^\perp}\cE u\|_2$, and suppose $\delta\leq o(1)/n$. Further suppose $1-\langle u,\cE u\rangle = (2+o(1))k^2/n,-O(k/n)\cdot\Id\preceq\Pi_{u^\perp}\cE\Pi_{u^\perp}\preceq O(k^2/n^2)\cdot\Id$. Then we have $\E_{\bm g}\|H^\cH_{\operatorname{SYK}}\|_{\operatorname{op}} = (1 - o_n(1))\sqrt{2n}/k$ for $\omega(\log(n))\leq k\leq o(\sqrt{n})$.
\end{remark}

After the first version of this preprint was posted online, David Gamarnik pointed out to us that \cite[Theorems~6,~25]{AnschuetzGK25}, together with \cref{thm:main-thm-complete}, can be used to show that (in the setting of \cref{thm:main-thm-random}) with probability $\geq 1-o_n(1)$ over the randomness of $\cH$, we have $\E_{\bm g}\|H^\cH_{\operatorname{SYK}}\|_{\operatorname{op}}=(1+o_n(1))\cdot\sqrt{2n}/k$ whenever $1\ll k\ll\sqrt n$ and $m\geq \poly(k)\cdot n$. This slightly improves parameters in \cref{thm:main-thm-random} in some regimes. The proof in \cite{AnschuetzGK25} works for general Hamiltonian ensembles (beyond just SYK) and uses a Lindeberg interpolation argument to transfer spectral norm bounds from the dense to sparse case.

\paragraph{Subgaussian Concentration for $\|H^\cH_{\syk}\|_{\op}$} Using standard Lipschitz concentration tools in probability theory (see \cref{fact:gaussian-lipschitz-conc}), we can establish subgaussian concentration for $\|H^\cH_{\syk}\|_{\op}$:

\begin{theorem}[Subgaussian Concentration for $\|H^\cH_{\syk}\|_{\op}$ (see \cref{thm:syk-subgaussian,thm:syk-subgaussian-sparse} for full version)]\label{thm:main-thm-subgaussian}
    Let $n\geq k\geq 2$ be even integers, and suppose $k^2/n < 1/16$. Fix any real $t\geq 0$. Then for $\cH = \binom{[n]}{k}$ (the ``dense case'')
    \[\Pr_{\bm g}\left(\left|\|H^\cH_{\syk}\|_{\operatorname{op}} - \E_{\bm g}\|H^\cH_{\syk}\|_{\operatorname{op}}\right|\geq t\right)\leq 2\cdot\exp\left(-\Omega(nt^2)\right).\]
    If $\cH$ is a random $k$-uniform hypergraph with $m\geq 2^{Ck}n\log n$ edges, where $C > 0$ is some sufficiently large absolute constant, then with probability $\geq 1 - o_n(1)$ over the randomness of $\cH$ we have
    \[\Pr_{\bm g}\left(\left|\|H^\cH_{\syk}\|_{\operatorname{op}} - \E_{\bm g}\|H^\cH_{\syk}\|_{\operatorname{op}}\right|\geq t\right)\leq 2\cdot\exp\left(-\Omega(kt^2)\right).\]
\end{theorem}
\begin{remark}
    In the random hypergraph case, by increasing the number of hyperedges in $\cH$ from $2^{\Omega(k)}n\log(n)$ to $2^{\Omega(k)}n^3\log(n)$ we can get the tail decay to be $2\cdot\exp\left(-\Omega(nt^2)\right)$ as in the dense case.
\end{remark}

\paragraph{Optimal Quantum Algorithms to sample SYK Ground States} As a corollary of \cref{thm:main-thm-complete,thm:main-thm-random}, we prove that the dissipative quantum algorithm of~\cite{basso2024optimizing} provably achieves an $\Omega(1)$ approximation to the ground state energy of the sparse SYK model whenever $\omega_n(1)\leq k\leq o(\sqrt{n})$. We note that for the case of $k=4$, a recent work~\cite{dingOptimizingSparseSYK2026} gave an efficient quantum algorithm to synthesize a state $\rho(H^\cH_{\operatorname{4-SYK}})$ which achieves $\E_{\bm g}\Tr(\rho(H^\cH_{\operatorname{4-SYK}})H^\cH_{\operatorname{4-SYK}})\geq\Omega(\sqrt{n})$ w.h.p. over random $4$-uniform hypergraphs $\cH$ with $m\geq\Omega(n^3\log(n))$ hyperedges. Their algorithm was a suitable modification of a quantum algorithm due to \cite{hastings2022optimizing} who showed a similar result for the dense case $\cH = \binom{[n]}{4}$. See also \cref{footnote:HO-eff-quant}.

\begin{theorem}[Optimal Quantum Algorithms for the SYK Model]\label{thm:optimal-qalg-syk}
    Let $n\geq k\geq 2$ be even integers such that $k^2/n < 1/16$. Consider the $H^\cH_{\operatorname{SYK}}$ Hamiltonian, where we take $\cH = \binom{[n]}{k}$ (which yields the SYK Hamiltonian \eqref{eq:SYK-intro-def}). Then there exists a dissipative quantum algorithm with runtime $\leq n^{O(k)}$, which on inputting $H^\cH_{\operatorname{SYK}}$ synthesizes the state $\rho(H^\cH_{\operatorname{SYK}})$ such that 
    \begin{equation}\label{eq:optimal-ground-state}
        \E_{\bm g}\left[\Tr(\rho(H^\cH_{\operatorname{SYK}})\cdot H^\cH_{\operatorname{SYK}})\right]\geq\Omega(1)\cdot\E_{\bm g}\left\| H^\cH_{\operatorname{SYK}}\right\|_{\operatorname{op}}.
    \end{equation}
    Note that $\E_{\bm g}\left[\Tr(\rho(H^\cH_{\operatorname{SYK}})\cdot H^\cH_{\operatorname{SYK}})\right]\leq\E_{\bm g}\left\| H^\cH_{\operatorname{SYK}}\right\|_{\operatorname{op}}$, and thus the result is optimal up to constant factors. 
\end{theorem}

\cref{thm:optimal-qalg-syk} follows simply by combining~\cite[Theorem~3.1]{basso2024optimizing} with \cref{thm:main-thm-complete}.

\paragraph{Improved Bounds on the Spectral Edge of Quantum Spin Glasses} Our technique to obtain the \emph{upper bound} on $\E_{\bm g}\|H_{\operatorname{SYK}}\|_{\operatorname{op}}$ extends readily to a much more general setting. To illustrate that, we obtain improved bounds on the spectral edge of quantum spin glasses. To that end, for any integer $n\geq 1$ write 
\[\{\Id_2, X, Y, Z\}^{\otimes n} := \left\{P = \bigotimes_{i = 1}^n P_i:P_1, \ldots, P_n\in\{\Id, X, Y, Z\}\right\}\subset\C^{2^n\times 2^n},\]
where $\Id_2, X, Y, Z$ are the single-qubit Pauli matrices, 
\[
\Id_2 =\begin{pmatrix}1&0\\0&1\end{pmatrix},\quad X=\begin{pmatrix}0&1\\1&0\end{pmatrix},\quad Y=\begin{pmatrix}0&-\iota\\\iota&0\end{pmatrix},\quad Z=\begin{pmatrix}1&0\\0&-1\end{pmatrix}
\]
For any $P\in\{\Id, X, Y, Z\}^{\otimes n}$ denote by $\wt(P)$ the number of indices $i\in[n]$ for which $P_i\neq\Id$, and let $\cP_{n, k}\seq\{\Id, X, Y, Z\}^{\otimes n}$ be the set of matrices $P$ for which $\wt(P) = k$. Note that $|\cP_{n, k}| = 3^k\binom{n}{k}$, and define the $k$-local \emph{Quantum Spin Glass} Hamiltonian
\[H_{k-\QSG}:= \frac{1}{\sqrt{m}}\sum_{P\in\cP_{n, k}}g_PP,\]
where $m:= 3^k\binom{n}{k}$, and $(g_P)_{P\in\cP_{n, k}}$ are i.i.d $\cN(0, 1)$ Gaussians disorder variables, similar to the SYK model.

For these spin glass Hamiltonians, recent work~\cite{basso2024optimizing} shows that $\Omega(\sqrt{n}/k)\leq\E_{\bm g}\|H_{k-\QSG}\|_{\op}\leq O(\sqrt{n})$ for all $1 < k \leq O(\sqrt{n})$, with a conjecture (as in the case of SYK) that in fact $\E_{\bm g}\|H_{k-\QSG}\|_{\op} = \Theta(\sqrt{n}/k)$. While we are unable to close the $O(1/k)$ gap, we tighten the upper bound to $O(\sqrt{n/k})$:
\begin{theorem}[Improved Bounds for $\E\|H_{k-\QSG}\|_{\op}$ (See \cref{thm:pauli-improved-spectral-norm} for full version)]\label{thm:main-pauli-improved-spectral-norm}
    For any $n\geq k\geq 1$ with $k\leq\sqrt{n}/2$ we have
    \[\Omega\left(\frac{\sqrt{n}}{k}\right)\leq\E_{\bm g}\|H_{k-\QSG}\|_{\op}\leq O\left(\sqrt{\frac{n}{k}}\right).\]
\end{theorem}
As with $H_{\syk}$, \cite{basso2024optimizing} give an efficient quantum algorithm to synthesize a state $\rho(H_{k-\QSG})$ for which $\E_{\bm g}\Tr(\rho(H_{k-\QSG})H_{k-\QSG})\geq\Omega(\sqrt{n}/k)$, and thus \cref{thm:pauli-improved-spectral-norm} implies that $\rho(H_{k-\QSG})$ achieves a $O(\sqrt{k})$-approximation to the spectral norm of $H_{k-\QSG}$.

We note that, in general, the problem of identifying the spectral edge of disordered (SYK or QSG) Hamiltonians can be viewed as a natural quantum analogue of determining the optimum of a classical random satisfaction problem, or the injective norm of a random tensor. This is a topic with significant recent developments and applications~\cite{Feige02,WeinEAM25,GuruswamiKM22,HsiehKM23}, serving as an inspiration for some of the technical ideas in this work.

\section{Technical Overview}\label{sec:overview}
We now give an exposition of the main ideas that go into the proof of our main results. Let us first focus on the upper bound for $\E_{\bm g}\|H^\cH_{\operatorname{SYK}}\|_{\operatorname{op}}$, where we take $\cH$ to be either $\binom{[n]}{k}$ or a random $k$-uniform hypergraph with $2^{\Omega(k)}n\log(n)$ edges. 

\paragraph{The Trace Moment Method.} A standard way to prove sharp upper bounds on the operator norms of random matrices is the trace moment method based on the following simple observation: for any integer $\ell\geq 1$ we have 
\[\E_{\bm g}\left[\|H^\cH_{\operatorname{SYK}}\|_{\operatorname{op}}\right]^{2\ell}\overset{\text{Jensen}}{\leq}\E_{\bm g}\left[\|H^\cH_{\operatorname{SYK}}\|_{\operatorname{op}}^{2\ell}\right]\leq\E_{\bm g}\left[\Tr((H^\cH_{\operatorname{SYK}})^{2\ell})\right] = N\cdot\E_{\bm g}\left[\tr((H^\cH_{\operatorname{SYK}})^{2\ell})\right]\] 
\[\implies\E_{\bm g}\|H^\cH_{\operatorname{SYK}}\|_{\operatorname{op}}\leq N^{1/(2\ell)}\cdot\E_{\bm g}\left[\tr((H^\cH_{\operatorname{SYK}})^{2\ell})\right]^{1/(2\ell)},\]
where $H_{\syk}$ is treated as a $N$-dimensional operator,\footnote{It is well known that $H_{\syk}$ can be represented as a finite dimensional operator in $N = 2^{n/2}$-dimensional Hilbert space~\cite{BravyiK02}.} and $\Tr, \tr$ refer to unnormalized and normalized trace operators respectively. The advantage of this estimate is that $\E_{\bm g}\left[\tr(H^{2\ell})\right]$ admits a combinatorial interpretation. Indeed, writing $m = |\cH|$ we obtain that 
\[\E_{\bm g}\left[\tr((H^\cH_{\operatorname{SYK}})^{2\ell})\right] = m^{-\ell}\cdot\sum_{S_1, \ldots, S_{2\ell}\in\cH}\E_{\bm g}\tr\left(g_{S_1}\Gamma_{S_1}\cdot g_{S_2}\Gamma_{S_2}\cdots g_{S_{2\ell}}\Gamma_{S_{2\ell}}\right)\]
\[ = m^{-\ell}\cdot\sum_{S_1, \ldots, S_{2\ell}\in\cH}\E_{\bm g}\left[g_{S_1}g_{S_2}\cdots g_{S_{2\ell}}\right]\cdot\tr\left(\Gamma_{S_1}\Gamma_{S_2}\cdots\Gamma_{S_{2\ell}}\right).\]
\paragraph{Evaluating Gaussian moments:} $\E_{\bm g}\left[g_{S_1}g_{S_2}\cdots g_{S_{2\ell}}\right]$ can be evaluated using the classical Isserlis-Wick theorem (see \cref{fact:isserlis-wick}) in probability theory. To invoke \cref{fact:isserlis-wick} we recall some standard terminology: Let $\pi$ be a perfect matching on $[2\ell]$, i.e. $\pi = \{e_1, \ldots, e_\ell\}$ where $\bigsqcup_{i = 1}^\ell e_i = [2\ell]$, $|e_i| = 2$, and $e_i\cap e_j = \emptyset$ for all $1\leq i < j\leq\ell$. The edges $e_1, \ldots, e_\ell$ are also popularly known as \emph{chords}, and for a chord $e = \{i < j\}$, we say $i$ is its \emph{left endpoint}, while $j$ is its \emph{right endpoint}. If we draw the numbers $1, 2, \ldots, 2\ell$ on a line, and connect $i, j$ with a chord if $\{i, j\}\in\pi$, then we get the \emph{chord diagram} $G_\pi$ corresponding to $\pi$. More formally, $G_\pi$ is a graph whose vertices are the chords of $\pi$, and two chords $e, e'\in\pi$ are connected if they \emph{cross} each other, i.e. if $e = \{i_1 < i_2\}$ and $e' = \{i_3 < i_4\}$ satisfy either $i_1 < i_3 < i_2 < i_4$ or $i_3 < i_1 < i_4 < i_2$. 

Now \cref{fact:isserlis-wick} gives us that 
\[\E_{\bm g}\left[g_{S_1}g_{S_2}\cdots g_{S_{2\ell}}\right] = \sum_\pi\prod_{e = \{i, j\}\in\pi}\E_{\bm g}[g_{S_i}g_{S_j}].\]
Since $\{g_S\}_{S\in\binom{[n]}{k}}$ are independent $\cN(0, 1)$ Gaussians, we have $\E_{\bm g}[g_{S_i}g_{S_j}] = \1(S_i = S_j)$. Thus 
\[\E_{\bm g}\left[\tr((H^\cH_{\operatorname{SYK}})^{2\ell})\right] = m^{-\ell}\cdot\sum_\pi\sum_{S_1, \ldots, S_{2\ell}\in\cH}\1_\pi(S_1, \ldots, S_{2\ell})\cdot\tr\left(\Gamma_{S_1}\Gamma_{S_2}\cdots\Gamma_{S_{2\ell}}\right),\]
where $\1_\pi(S_1, \ldots, S_{2\ell}):= \prod_{e = \{i, j\}\in\pi}\1(S_i = S_j)$. Notice that the only tuples $(S_1, \ldots, S_{2\ell})$ that contribute a non-zero term are ones where every $S\in\cH$ occurs an even number of times.

\paragraph{Using Anti-Commutation Relations:} Now, it follows from the anti-commutation properties of Majoranas that $\Gamma_S^2 = \Id$, and $\Gamma_S\Gamma_T = \eps_{S, T}\Gamma_T\Gamma_S$, where $\eps_{S, T}:= (-1)^{|S|\cdot|T| - |S\cap T|} = (-1)^{|S\cap T|}$. Consequently, we can take the expression $\Gamma_{S_1}\Gamma_{S_2}\cdots\Gamma_{S_{2\ell}}$ and keep swapping the matrices $\Gamma_\ast$s until they meet their pair under $\pi$. Consequently, $\Gamma_{S_1}\cdots\Gamma_{S_{2\ell}} = \eps_\pi\Id$, where $\eps_\pi\in\{\pm 1\}$ is a sign which depends on $\pi$. In fact, suppose $(S_1, \ldots, S_{2\ell})$ is such that $\1_\pi(S_1, \ldots, S_{2\ell}) = 1$. Then for every $e = \{i, j\}\in\pi$, we have $S_i = S_j$, and thus define $S_e:= S_i = S_j$. Also write $\eps_{e, e'}:= (-1)^{|S_e\cap S_{e'}|}$. Now note that $\eps_\pi = \prod_{\{e, e'\}\in E(G_\pi)}\eps_{e, e'}$, and thus we have 
\[\E_{\bm g}\left[\tr((H^\cH_{\operatorname{SYK}})^{2\ell})\right] = m^{-\ell}\cdot\sum_\pi\sum_{S_1, \ldots, S_{2\ell}\in\cH}\1_\pi(S_1, \ldots, S_{2\ell})\cdot\prod_{\{e, e'\}\in E(G_\pi)}\eps_{e, e'}\]
\[ = m^{-\ell}\cdot\sum_\pi\sum_{\{S_e\}_{e\in\pi}}\prod_{\{e, e'\}\in E(G_\pi)}\eps_{e, e'},\]
where we only sample $\ell$ elements $S_e$ from $\cH$, and repeat them as $\pi$ dictates. But we can also absorb the $m^{-\ell}$ term into the summation to make it an expectation, and we obtain
\begin{equation}
\label{eq:trace-moment-expansion}
\E_{\bm g}\left[\tr((H^\cH_{\operatorname{SYK}})^{2\ell})\right] = \sum_\pi\E_{\{S_e\}_{e\in\pi}\sim\cH}\prod_{\{e, e'\}\in E(G_\pi)}\eps_{e, e'}.
\end{equation}

\subsection{Our Key Idea: A Deterministic Moment-Matching Operator?}
\label{sec:baby-johnson} To illustrate our key idea, let us go back to the trace moment expansion from \eqref{eq:trace-moment-expansion}. Each term there is indexed by a pairing $\pi$, and its value is determined by the crossing graph $G_\pi$.  

Let us focus on the dense case $\cH=\binom{[n]}{k}$ and on a pairing whose crossing graph is a cycle to see the main insight: we can hope to relate the \emph{expected} trace moments of $H_{\syk}$ to the trace moments of a ${n \choose k}$ deterministic dimensional matrix from the \emph{Johnson} scheme, which, in particular, allows us to read off its trace moments easily. 

Write $m=|\cH|$, and define the deterministic matrix
\[
    \cE\in\R^{\binom{[n]}{k}\times\binom{[n]}{k}},
    \qquad
    \cE(S,T):=\frac{1}{m}(-1)^{|S\cap T|}.
\]
Suppose that $G_\pi$ is a cycle on $r$ chords, listed in cyclic order as $e_1,\ldots,e_r$, and set $e_{r+1}:=e_1$. The contribution of $\pi$ to the trace-moment expansion is then
\begin{align*}
    \E_{\{S_e\}_{e\in\pi}\sim\cH}
    \prod_{\{e,e'\}\in E(G_\pi)}\eps_{e,e'}
    &=
    \frac{1}{m^r}
    \sum_{S_1,\ldots,S_r\in\binom{[n]}{k}}
    \prod_{i=1}^r(-1)^{|S_i\cap S_{i+1}|}\\
    &=
    \sum_{S_1,\ldots,S_r\in\binom{[n]}{k}}
    \cE(S_1,S_2)\cdots \cE(S_r,S_1)
    =
    \Tr(\cE^r),
\end{align*}
where $S_{r+1}:=S_1$. In the expansion of the $2\ell$-th moment above, $\pi$ has $\ell$ chords, so this identity gives $\Tr(\cE^\ell)$. Notice that the equality also includes choices for which two or more chords receive the same hyperedge, since the hyperedges are sampled independently with replacement.

If $\lambda_1,\ldots,\lambda_m$ are the eigenvalues of $\cE$, counted with multiplicity, then
\[
    \Tr(\cE^r)=\sum_{j=1}^m\lambda_j^r.
\]
Thus, the contribution from all $\pi$ where $G_{\pi}$ is a cycle can be bounded directly in terms of the eigenvalues of $\cE$. These eigenvalues are explicit: the entry $\cE(S,T)$ depends only on $|S\cap T|$, and
\[
    \cE=\frac{1}{m}\sum_{t=0}^k(-2)^tM_t,
    \qquad
    M_t(S,T):=\binom{|S\cap T|}{t}.
\]

Thus $\cE$ belongs to the \emph{Johnson} association scheme~\cite{Godsil10}, whose eigenspaces and eigenvalues are well understood. We have therefore expressed a part of the random SYK trace-moment expansion as a trace power of a deterministic, $\binom{n}{k}$-dimensional operator that can be analyzed explicitly.

In the next section, we show a generalization of this simple observation by constructing an explicit deterministic operator whose spectrum controls the sum of \emph{all} (i.e., not just the case when $G_{\pi}$ is a cycle) the terms in \eqref{eq:trace-moment-expansion}. Like above, this deterministic operator will also be closely related to a matrix from the Johnson scheme that will help us read off the trace moment bound exactly. It turns out that our construction has a natural interpretation as a twisted model of well studied operators from the bosonic algebra. We describe this construction next.

\subsection{The ``$q$-deformed'' heuristic} Before we explain our general method to estimate this sum, it is instructive to consider the following heuristic (see \cite{hastings2022optimizing} for a nice exposition): Suppose we replace 
\[\E_{\{S_e\}_{e\in\pi}\sim\cH}\prod_{\{e, e'\}\in E(G_\pi)}\eps_{e, e'}\] 
with 
\[\prod_{\{e, e'\}\in E(G_\pi)}\E_{S_e, S_{e'}\sim\cH}\eps_{e, e'}.\]
Now note that 
\[q:= \E_{S_e, S_{e'}\sim\cH}\eps_{e, e'} = \E_{S_e, S_{e'}\sim\cH}(-1)^{|S_e\cap S_{e'}|} = 1 - 2\alpha + o(\alpha),\] 
where $\alpha:= k^2/n$ is assumed to be $o_n(1)$.\footnote{In our actual proof we perform a more detailed computation and can take $\alpha$ to be up to a small constant, say $1/16$.} Thus substituting this in, we obtain 
\begin{align}\label{eq:trace-method-heuristic}
    \E_{\bm g}\left[\tr((H^\cH_{\operatorname{SYK}})^{2\ell})\right] = \sum_\pi\E_{\{S_e\}_{e\in\pi}\sim\cH}\prod_{\{e, e'\}\in E(G_\pi)}\eps_{e, e'}\approx \sum_\pi\prod_{\{e, e'\}\in E(G_\pi)}\E_{S_e, S_{e'}\sim\cH}\eps_{e, e'} = \sum_\pi q^{e(G_\pi)},
\end{align}
where $e(G_\pi)$ is the number of edges in $G_\pi$. We first discuss a method developed to perform this combinatorial sum, and thereafter discuss connections to the physics literature. The combinatorial problem of evaluating such chord diagrams has been studied as far back as the work of Touchard~\cite{touchard1952probleme} and Riordan~\cite{riordan1975distribution}. We also note that~\cite{gamarnik2026free} recently proposed a ``Chord-Cavity-Generator'' technique to evaluate chord diagrams rigorously. 

\paragraph{Evaluating $\sum_\pi q^{e(G_\pi)}$.} Setting aside the issue that the heuristic in \cref{eq:trace-method-heuristic} has not been justified yet, there is still the question of evaluating the expression $\sum_\pi q^{e(G_\pi)}$ for $q = 1 - 2\alpha + o(\alpha)\approx 1$. We follow the exposition of \cite{SYK_Berkooz3}. Write $[n]_q := \frac{1-q^n}{1-q}$, and define $\sfc, \sfc^*, T:= \sfc + \sfc^*$ as
\[
\sfc \coloneqq
\left(
\begin{array}{cccccc}
0 & \sqrt{[1]_q} & 0 & 0 & 0 & \cdots\\
0 & 0 & \sqrt{[2]_q} & 0 & 0 & \cdots\\
0 & 0 & 0 & \sqrt{[3]_q} & 0 & \cdots\\
0 & 0 & 0 & 0 & \sqrt{[4]_q} & \cdots\\
\vdots & \vdots & \vdots & \vdots & \ddots & \ddots
\end{array}
\right),
\qquad
\sfc^* \coloneqq
\left(
\begin{array}{cccccc}
0 & 0 & 0 & 0 & 0 & \cdots\\
\sqrt{[1]_q} & 0 & 0 & 0 & 0 & \cdots\\
0 & \sqrt{[2]_q} & 0 & 0 & 0 & \cdots\\
0 & 0 & \sqrt{[3]_q} & 0 & 0 & \cdots\\
0 & 0 & 0 & \sqrt{[4]_q} & 0 & \ddots\\
\vdots & \vdots & \vdots & \vdots & \ddots & \ddots
\end{array}
\right).
\]

\[
T \coloneqq \sfc+\sfc^*
=
\left(
\begin{array}{cccccc}
0 & \sqrt{[1]_q} & 0 & 0 & 0 & \cdots\\
\sqrt{[1]_q} & 0 & \sqrt{[2]_q} & 0 & 0 & \cdots\\
0 & \sqrt{[2]_q} & 0 & \sqrt{[3]_q} & 0 & \cdots\\
0 & 0 & \sqrt{[3]_q} & 0 & \sqrt{[4]_q} & \cdots\\
0 & 0 & 0 & \sqrt{[4]_q} & 0 & \ddots\\
\vdots & \vdots & \vdots & \vdots & \ddots & \ddots
\end{array}
\right).
\]
Here, we imagine the rows and columns of $\sfc$ and $\sfc^*$ as indexed by $\Z_{\geq 0}$. Formally, $\sfc, \sfc^*$ are linear operators on $\ell_2(\Z_{\geq 0}) = \{(x_0,x_1, \ldots, x_n, \ldots): \sum_{n = 0}^{\infty} |x_n|^2 <\infty\}$. We write $f_{\bm 0} = \begin{pmatrix}
    1 & 0 & 0 & \cdots
\end{pmatrix}^\top\in\ell_2(\Z_{\geq 0})$. Then a simple combinatorial argument\footnote{Note that since the only non-zero entries in $T$ are in the first off-diagonal entries, expanding $\langle f_{\bm 0}, T^{2\ell}f_{\bm 0}\rangle$ yields a sum parametrized by tuples $\tau:= (n_0, \ldots, n_{2\ell})\in\Z_{\geq 0}^{2\ell + 1}$ such that $n_0 = n_{2\ell} = 0$ and $|n_i - n_{i + 1}| = 1$ for all $0\leq i < 2\ell$. Now for any perfect matching $\pi$ on $[2\ell]$, define a tuple $\tau:= (n_0, \ldots, n_{2\ell})$ where $n_{i} - n_{i - 1}:= 1$ if $i$ is the left endpoint of a chord in $\pi$, and $-1$ otherwise. It can be shown that the term in the sum corresponding to $\tau$ equals $\sum_{\pi}q^{e(G_\pi)}$ over all $\pi$ which produce $\tau$ as their tuple, and we are done.} \cite{SYK_Berkooz3} shows that $\langle f_{\bm 0}, T^{2\ell}f_{\bm 0}\rangle = \sum_\pi q^{e(G_\pi)}$. 

Thus, heuristically, understanding the spectral norm of $T$ suffices to provide an upper bound on $\sum_\pi q^{e(G_\pi)}$, and thus also on $\E_{\bm g}\left[\tr((H^\cH_{\operatorname{SYK}})^{2\ell})\right]$. Indeed, $\E_{\bm g}\left[\tr((H^\cH_{\operatorname{SYK}})^{2\ell})\right]^{1/(2\ell)}\overset{\text{\cref{eq:trace-method-heuristic}}}{\leq}\langle f_{\bm 0}, T^{2\ell}f_{\bm 0}\rangle^{1/(2\ell)}\leq\|T\|_{\operatorname{op}} \overset{(\ast)}{=} 2/\sqrt{1 - q}\approx\sqrt{2/\alpha} = \sqrt{2n}/k$, where $(\ast)$ was also shown in \cite{SYK_Berkooz3}. Note that $\sqrt{2n}/k$ was precisely the upper bound we were aiming for!

 \subsection{A Deterministic ``Twisted-Bosonic'' Moment-Matching Model}

The key idea of our proof is to introduce a natural, tractable,  \emph{deterministic} operator $\sfx$ such that a natural functional of $\sfx$ has the \emph{same} moments as the expected trace moments of $H$.  We will define operators $\sfa$ and $\sfx:= \sfa + \sfa^*$ such that $\E_{\bm g}\left[\tr((H^\cH_{\operatorname{SYK}})^{2\ell})\right] = \langle f_{\bm 0}, \sfx^{2\ell}f_{\bm 0}\rangle$ for $f_0$ being the vacuum state. This equality holds for all $\ell$ simultaneously. We can then analyze $\sfx$ to understand $\E_{\bm g}\left[\tr((H^\cH_{\operatorname{SYK}})^{2\ell})\right]$. 

\paragraph{Review of Bosonic Algebra:} The algebra for a single boson is defined on the Fock space $\ell_2(\Z_{\geq 0})$ equipped with the orthonormal basis $\{f_r\}_{r\in \Z_{\geq 0}}$, with the annihilation and creation operators defined as, 
\[
\sfb f_r = \sqrt{r}f_{r-1} \qquad \sfb^* f_r = \sqrt{r+1}f_{r+1}
\]
Crucially, a single bosonic mode satisfies the commutation relation (known as \emph{bosonic statistics}) $[\sfb, \sfb^*] =\Id$.

In contrast, the operator $\sfc,\sfc^*$ used in the evaluation of the heuristic satisfy the $q$-deformed commutation relation $\sfc\sfc^* - q \sfc^*\sfc = \Id$.

A collection of bosons, indexed by the set of hyperedges, $\{\sfb_S\}_{S\in\cH}$ satisfy the mutual commutation relations (see~\cref{def:bosonop} for a formal construction and~\cref{lem:bosonopcomm} for proof of commutation relations),
\[
[\sfb_S,\sfb_T]=[\sfb_S^*,\sfb_T^*]=0,
\qquad
[\sfb_S,\sfb_T^*]=\1[S=T]\Id.
\]
In particular, for $S\neq T$, the respective bosonic operators commute.

\paragraph{Constructing $\sfa$:} 
We begin by defining a collection of operators $\{\sfa_S\}_{S\in\cH}$ which satisfy the same (anti)-commutation relations as $\{\Gamma_S\}_{S\in\cH}$, i.e. $\sfa_S\sfa_T = \eps_{S, T}\sfa_T\sfa_S$ and $\sfa_S\sfa_T^* - \eps_{S, T}\sfa_T^*\sfa_S=\1[S=T]\Id$ (see \cref{def:twisted-bosonic} for formal construction), where $\eps_{S, T}:= (-1)^{|S\cap T|}$. Note that $\eps_{S,S}=1$, thus each $\sfa_S$ is individually bosonic, satisfying $[\sfa_S,\sfa_S^*] = \Id$. Their mutual statistics are determined by the signs $\eps_{S, T}$, hence the \emph{twist}. We then define $\sfa:= |\cH|^{-1/2}\cdot\sum_{S\in\cH}\sfa_S$ and further $\sfx:= \sfa + \sfa^*$ (see~\cref{tab:twisted-boson-algebra} for a summary of the twisted boson construction), and the key technical step in our proof (see \cref{prop:varphi-trace-re}) is the fact that 
\begin{equation}\label{eq:twisted-boson-majorana-trace}
    \langle f_{\bm 0}, \sfx^{2\ell} f_{\bm 0}\rangle = \E_{\bm g}\left[\tr\left(\left(H^\cH_{\operatorname{SYK}}\right)^{2\ell}\right)\right].
\end{equation}
While we shall not explain the proof of \cref{eq:twisted-boson-majorana-trace}, the fact itself should not be that surprising, since the operators $\{\sfa_S\}_{S\in\cH}$ have been constructed to encode the same exchange signs as the SYK constituent operators $\{\Gamma_S\}_{S\in\cH}$.

Now we proceed to justify why $\sfa$ is the correct ``finitary'' analog of $\sfc$. Towards that end, let $\cE:= |\cH|^{-1}\cdot(\eps_{S, T})_{S, T\in\cH}\in\C^{\cH\times\cH}$ be the \emph{sign matrix}, and write 
\[u:= \frac{\1}{\sqrt{m}}\in\C^\cH,\qquad q:= \langle u, \cE u\rangle,\qquad \Pi_{u^\perp}:= \Id - uu^*,\qquad \delta:= \|\Pi_{u^\perp}\cE u\|_2,\] 
\[\rho^+:= \lambda_{\max}(\Pi_{u^\perp}\cE\Pi_{u^\perp}),\qquad \rho^-:= \max\{0, -\lambda_{\min}(\Pi_{u^\perp}\cE\Pi_{u^\perp})\}.\]
Looking a bit ahead, we will essentially show that $\cE\approx quu^*$, and $\delta, \rho^+, \rho^-$ are three different ways of quantifying how much $\cE$ actually deviates from $quu^*$ (see~\cref{tab:sign-matrix-parameters} for a summary of these parameters). In fact, for $k\leq O(\sqrt{n})$ we'll show that $\rho^+, \rho^-, \delta\leq n^{-\Omega(1)}$, and thus $\cE$ does equal $quu^*$ as $n\to\infty$.

\paragraph{Bounding the Spectral Norm of $\sfx$.} By \cref{eq:twisted-boson-majorana-trace} it suffices to bound $\langle f_{\bm 0}, \sfx^{2\ell} f_{\bm 0}\rangle$. A natural way to do so is to mimic the proof in \cite{SYK_Berkooz3}. Towards that, in \cref{lem:varphi-spectral-dom,eq:aalowerbound} we show that 
\begin{equation}\label{eq:approx-deformed-commutation}
    \rho^-\sfa^*\sfa - (\rho^- + \delta)\cdot\sum_{S\in\cH}\sfa^*_S\sfa_S\preceq\sfa\sfa^* - q\sfa^*\sfa - \Id\preceq(\rho^+ + \delta)\sum_{S\in\cH}\sfa^*_S\sfa_S - \rho^+\sfa^*\sfa.
\end{equation}
Now, as mentioned above, we can show that $\rho^+, \rho^-, \delta\leq n^{-\Omega(1)}$ for $k\leq O(\sqrt{n})$. The way we do this is to observe that for $\cH = \binom{[n]}{k}$, the matrix $\cE$ is an example of a matrix from the \emph{Johnson scheme}, i.e. $\cE$ is a matrix indexed by $\binom{[n]}{k}$ such that $\cE(S, T)$ depends only on $|S\cap T|$. The spectra of matrices in the Johnson scheme is very well-studied (see \cite{delsarte1973algebraic,Godsil10,godsil2015erdos,Filmus16}), and thus it follows from standard computations that $\rho^+, \rho^-, \delta\leq n^{-\Omega(1)}$.\footnote{When $\cH$ is a random hypergraph, $\cE$ is no longer from the Johnson scheme, but standard matrix concentration results allow us to prove the necessary bounds on $\rho^+, \rho^-, \delta$.}

Consequently, note that as $n\to\infty$, \eqref{eq:approx-deformed-commutation} ``converges'' to $\sfa\sfa^* - q\sfa^*\sfa - \Id = 0$, which exactly matches the commutation relation $\sfc\sfc^* - q\sfc^*\sfc = \Id$ from our heuristic that we wanted to emulate in the first place! Equivalently, \eqref{eq:approx-deformed-commutation} shows that up to an error quantified by $\rho^+, \rho^-, \delta$, the heuristic \cref{eq:trace-method-heuristic} holds, and thus we obtain $\E_{\bm g}\|H^\cH_{\operatorname{SYK}}\|_{\operatorname{op}}\leq\sqrt{2n}/k + \operatorname{error}(\rho^+, \rho^-, \delta)$, thus yielding the upper bounds in \cref{thm:main-thm-complete,thm:main-thm-random}.

\paragraph{Fock Space Methods for General Matrix Concentration} Very recently, \cite{BandeiraKNPW26} generalized our Fock space methods to general random matrix ensembles, showing that for many interesting ensembles (including $H_{\syk}$ itself) one can improve upon the usual ``$\sqrt{\log N}$'' loss in matrix concentration inequalities by relating the spectral edge of the random matrix with the spectral edge of a suitable system of $q$-deformed operators, as we do in our paper. 


We recommend \cite{Handel26,Bandeira26,BandeiraKNPW26} for a detailed introduction into recent advances in matrix concentration inequalities.

\paragraph{Relation to Kikuchi matrices.}
A related random-to-Johnson comparison arises in the study of Kikuchi graphs. The Kikuchi graph of the complete uniform hypergraph belongs to the Johnson scheme, while the Kikuchi graph of a sufficiently dense random hypergraph spectrally approximates this complete-hypergraph operator~\cite{Kothari25}. Kikuchi and related symmetric-difference matrices were introduced in the study of tensor PCA and random CSP refutation and have since played an important role in algorithms for random and semirandom CSPs~\cite{WeinEAM25,GuruswamiKM22,HsiehKM23,GuruswamiHKM23}. Although our sign matrix is a different operator, the underlying principle is similar: a random-hypergraph operator can be understood by comparison with a highly symmetric operator from the Johnson scheme.

Now that we've proven the upper bounds in \cref{thm:main-thm-complete,thm:main-thm-random}, we focus on proving lower bounds for $\E_{\bm g}\|H^\cH_{\operatorname{SYK}}\|_{\operatorname{op}}$. Once again, we will take inspiration from the operators $\sfc, T$ and \cref{eq:trace-method-heuristic}.

\subsection{\texorpdfstring{Lower Bounding $\E_{\bm g}\|H^\cH_{\operatorname{SYK}}\|_{\operatorname{op}}$}{Lower Bounding the SYK Operator Norm}} Let us go back to the operators $\sfc, T$, and see what they tell us about $\E_{\bm g}\|H^\cH_{\operatorname{SYK}}\|_{\operatorname{op}}$. Choose $t_0$ so that $q^{t_0}\leq o_n(1)$. Also choose $M\geq \omega_n(1)$. Note that for indices $t\geq t_0$, $[t]_q = (1 - q^t)/(1 - q)\approx 1/(1 - q)$. Consequently, restricting $T$ to the indices $t_0\leq t < t_0 + M$ yields (roughly) the following matrix:
\[\frac{1}{\sqrt{1 - q}}\begin{pmatrix}
0 & 1 & 0 & \cdots & 0 \\
1 & 0 & 1 & \ddots & \vdots \\
0 & 1 & 0 & \ddots & 0 \\
\vdots & \ddots & \ddots & \ddots & 1 \\
0 & \cdots & 0 & 1 & 0
\end{pmatrix}.\]
But this is simply the adjacency matrix of the path graph (a.k.a. the ``hopping'' model in physics), which admits the eigenvector $\begin{pmatrix}
    \sin\left(\frac{(j + 1)\pi}{M + 1}\right)
\end{pmatrix}_{0\leq j < M}^\top$ and the eigenvalue $\frac{1}{\sqrt{1 - q}}\cdot 2\cos(\frac{\pi}{M + 1})\approx \frac{2}{\sqrt{1 - q}}$, which is very close to the operator norm of $T$!

Once again, we are able to reproduce this construction for $T$ to get a construction for $\sfx$. More precisely, for a specific choice of parameters $t_0, M$ (outlined in \cref{eq:main-params} and the proof of \cref{thm:syk-spectral-norm-lower-random}), we show (in \cref{lem:x-large-eig}) that if we define $\psi_t:= (\sfa^*)^tf_{\bm 0}, v_t:= \psi_t/\|\psi_t\|_2$, and finally the \emph{Krylov space}\footnote{For any operator $A$ and vector $v$, a space of the form $\spn\{A^tv:t_0\leq t\leq t_1\}$ is known as a \emph{Krylov space}.} $\cR:=\spn\{v_{t_0 + t}:0\leq t < M\}$, then 
\[\Pi_\cR\sfx\Pi_\cR = \begin{pmatrix}
0 & \beta_{t_0} & 0 & \cdots & 0 \\
\beta_{t_0} & 0 & \beta_{t_0+1} & \ddots & \vdots \\
0 & \beta_{t_0+1} & 0 & \ddots & 0 \\
\vdots & \ddots & \ddots & \ddots & \beta_{t_0+M-2} \\
0 & \cdots & 0 & \beta_{t_0+M-2} & 0
\end{pmatrix},
\]
where $\min_{0\leq t < M}\beta_{t_0 + t}\geq \frac{1 - o(1)}{\sqrt{1 - q}}$, and consequently we obtain a state (unit vector) $\psi\in\cR$ such that $\langle\psi, \sfx\psi\rangle\geq \frac{2 - o(1)}{\sqrt{1 - q}}\approx \sqrt{2n}/k$.

Note that we are not yet done: $\psi$ is only a good test function for our operator $\sfx$, but we still need to produce a certificate which achieves an energy of $\geq(1 - o(1))\cdot\sqrt{2n}/k$ against $H^\cH_{\operatorname{SYK}}$. We use the theory of multivariate Hermite polynomials (see \cref{def:hermite-polynomial-multivar}) and Gaussian concentration (see \cref{lem:lower-bound-syk-general}) to map the state $\psi$ to a certificate (see \cref{lem:transfer-virtual-to-real}) which achieves an energy of $\approx\sqrt{2n}/k$ against $H^\cH_{\operatorname{SYK}}$, thus yielding the lower bounds in \cref{thm:main-thm-complete,thm:main-thm-random}, as desired.

\paragraph{Relation to the doubly-scaled SYK transfer matrix picture.} The heuristic described in \cref{eq:trace-method-heuristic} has been studied extensively in the doubly-scaled SYK literature~\cite{SYK_Berkooz1,SYK_Berkooz2,SYK_Berkooz3,SYK_Garcia1,SYK_Garcia2,SYK_Garcia3,SYK_Jia,lin2022bulk,lin2023symmetry}. Here, the doubly-scaled limit refers to, 
\[
n,k\to\infty,
\qquad
\frac{k^2}{n}\to\alpha\in(0,\infty),
\]
with the parameter $q$ identified as $q=\E_{S,T\sim\Unif(\cH)}\eps_{S,T}\approx 1 - 2\alpha$. The operators $\sfc$ and $T$ are then interpreted as follows: $\sfc$ is the annihilation operator on the ``chord Hilbert space'', satisfying ``deformed''\footnote{Note that $q = 1$ yields the usual commutation relation for bosons, while $q \neq 1$ yields the deformation.} 
 commutation relation $\sfc\sfc^* - q \sfc^*\sfc = \Id$ with the corresponding creation operator $\sfc^*$. Equivalently, the chord Hilbert-space can be defined as the span of states resulting from excitations created by $\sfc^*$ on the vacuum state $f_{\bm 0}$. The transfer matrix is then exactly $T:= \sfc + \sfc^*$, so that the RHS $\sum_\pi q^{e(G_\pi)}$ evaluates to the vacuum state functional $\varphi(T^{2\ell}):= \langle f_{\bm 0},T^{2\ell}f_{\bm 0}\rangle$. Thus, by keeping track of the collective excitations on the chord Hilbert space, the transfer matrix picture allows moment computations in the doubly-scaled limit.

In contrast, our construction retains the full hyperedge-resolved twisted boson algebra before taking any large$-n$ limit. Each hyperedge (interaction term in the Hamiltonian) $S\in \cH$ is equipped with a twisted bosonic mode $\sfa_S$ satisfying the mutual statistics, 
\[
\sfa_S\sfa_T^*
-
\eps_{S,T}\sfa_T^*\sfa_S
=
\1[S=T]\cdot\Id.
\]
The $q-$deformed oscillator $\sfc$ naturally emerges out of our treatment by taking the doubly scaled limit of the collective operator $\sfa=|\cH|^{-1/2} \sum_{S\in\cH} \sfa_S$. Thus, our approach can be viewed as a finite-$(n,k)$
refinement of the transfer matrix picture in the following two aspects. First, the moment equality 
\[
    \langle f_{\bm 0}, \sfx^{2\ell} f_{\bm 0}\rangle = \E_{\bm g}\left[\tr\left(\left(H^\cH_{\operatorname{SYK}}\right)^{2\ell}\right)\right].
\]
holds for any $n,k$ and $\ell$. Correlations among the individual crossing signs are retained at finite size, with the deviations from the collective mode accounted for rigorously in terms of the leakage parameters $\rho^{\pm}$ and $\delta$ of the sign matrix $\cE$.
Second, retaining the individual modes $\{\sfa_S\}_{S\in\cH}$ allows for the construction of an explicit Wiener--It\^o isometric embedding into the original fermionic
operator algebra.\footnote{More precisely, the embedding is from the twisted Fock space into the space
of Gaussian matrix-valued polynomials over the original fermionic algebra.} This is crucial for the lower bound, as the top eigenvector for the twisted operator $\sfx$ is then transferred back into a disorder-dependent witness for the SYK Hamiltonian.

\paragraph{Prior work on ``twisted'' bosons.} Creation and annihilation operators with twisted commutation relations belong to the general family of mixed-$q$ Gaussian literature~\cite{bozejko1991example,frisch1970parastochastics,bozejko1997q,junge2015ultraproduct} further widely used in the doubly-scaled SYK literature~\cite{SYK_Berkooz1,SYK_Berkooz2,SYK_Berkooz3,SYK_Garcia1,SYK_Garcia2,SYK_Garcia3,SYK_Jia}. In particular,~\cite{junge2015ultraproduct,microscopicq_bozejko1994,microscopicq_speicher} provide results on operators encoding a generalized commutation matrix $q_{ij}$. A dynamical version of the SYK model has been related to $q$-Brownian motion in~\cite{pluma2022dynamical}, and convergence of the joint distribution of SYK Hamiltonians to a mixed $q$-Gaussian system in the large-$n$ limit has been shown in~\cite{liu2026limit}. Moreover,~\cite{DoresicMM94,MeljanacMP94,MeljanacM96} investigated, in an abstract setting, operators which satisfy the twisted commutation relations our operators satisfy.


\begin{figure}[p]
\centering
\scriptsize

\begin{tikzpicture}[x=1cm,y=1cm]


\tikzset{
  lemmanode/.style={
    draw,
    rounded corners=2.5pt,
    thick,
    align=center,
    fill=cyan!8,
    text width=4.25cm,
    minimum height=0.95cm,
    inner sep=3pt
  },
  resultnode/.style={
    draw,
    rounded corners=2.5pt,
    thick,
    align=center,
    fill=blue!10,
    text width=3.75cm,
    minimum height=0.95cm,
    inner sep=3pt
  },
  techlabel/.style={
    font=\scriptsize\itshape,
    align=center,
    fill=white,
    inner sep=1pt
  },
  dep/.style={
    -{Latex[length=2.2mm]},
    thick
  },
  branchbox/.style={
    draw=blue!50!black,
    rounded corners=4pt,
    dashed,
    thick
  },
  branchlabel/.style={
    font=\bfseries\footnotesize,
    fill=white,
    inner sep=2pt
  }
}


\draw[branchbox] (-8.3,0.7) rectangle (-0.4,-11.6);
\node[branchlabel,anchor=west] at (-6.75,0.7)
  {Upper bound};

\draw[branchbox] (0.4,-0.8) rectangle (7.0,-11.6);
\node[branchlabel,anchor=west] at (0.65,-0.7)
  {Lower bound};


\node[lemmanode] (l34) at (-3.7,0.0)
  {\cref{lem:twisted-commutations}\\Twisted commutation relations.};

\node[lemmanode] (l36) at (-3.7,-1.8)
  {\cref{lem:twisted-boson-mapping}\\ Mapping onto twisted bosons.};

\node[lemmanode] (p37) at (-3.7,-3.6)
  {\cref{prop:varphi-trace-re}\\Annealed moments of SYK equal the twisted boson vacuum functional};

\node[lemmanode] (p38) at (-3.7,-5.4)
  {\cref{prop:cE-dom} \\ Bounds on the sign matrix spectral ``leakage''.};

\node[lemmanode] (l39) at (-3.7,-7.2)
  {\cref{lem:varphi-spectral-dom}\\Spectral bound on twisted vacuum functional.};

\node[resultnode] (t42) at (-5.35,-9.2)
  {\cref{thm:syk-spectral-norm}\\Dense SYK spectral upper bound};

\node[resultnode] (t44) at (-2.05,-9.2)
  {\cref{thm:syk-spectral-norm-random}\\Sparse SYK spectral upper bound};


\draw[dep] (l34.south) -- node[techlabel,right=2pt,midway] {Isserlis-Wick} (l36.north);

\draw[dep] (l36.south) -- node[techlabel,left=2pt,midway] {} (p37.north);

\draw[dep] (p37.south) -- node[techlabel,right=2pt,midway] {} (p38.north);

\draw[dep] (p38.south) -- node[techlabel,left=2pt,midway] {Spectral bound.} (l39.north);

\draw[dep]
  (l39.south west) to[out=250,in=90]
  node[techlabel,left=1pt,pos=0.58] {Johnson scheme}
  (t42.north);

\draw[dep]
  (l39.south east) to[out=290,in=90]
  node[techlabel,right=1pt,pos=0.58] {concentration}
  (t44.north);


\node[lemmanode] (l310) at (3.7,-1.5)
  {\cref{lem:betam}\\Lower bound on hopping amplitude.};

\node[lemmanode] (l312) at (3.7,-3.3)
  {\cref{lem:x-large-eig}\\Spectral edge of hopping model.};

\node[lemmanode,text width=4.45cm] (p31314) at (3.7,-5.1)
  {\cref{prop:isometry} + \cref{lem:intertwining}\\Hermite isometry $+$ intertwiner.};

\node[lemmanode] (l315) at (3.7,-6.9)
  {\cref{lem:transfer-virtual-to-real}\\ Operator norm witness.};

\node[lemmanode,text width=4.45cm] (l31718) at (3.7,-8.9)
  {\cref{lem:lipschitz-operator} + \cref{lem:lower-bound-syk-general}\\Lipschitz + unweighting.};

\node[resultnode] (t43) at (2.05,-10.8)
  {\cref{thm:syk-spectral-norm-lower}\\Dense SYK spectral lower bound};

\node[resultnode] (t45) at (5.35,-10.8)
  {\cref{thm:syk-spectral-norm-lower-random}\\Sparse SYK spectral lower bound};


\draw[dep] (l310.south) -- node[techlabel,right=2pt,midway] {Krylov space} (l312.north);

\draw[dep] (l312.south) -- node[techlabel,left=2pt,midway] {Hermite polynomials} (p31314.north);

\draw[dep] (p31314.south) -- node[techlabel,right=2pt,midway] {Multiplication Operator} (l315.north);

\draw[dep] (l315.south) -- node[techlabel,left=2pt,midway]
  {Hypercontractivity} (l31718.north);

\draw[dep]
  (l31718.south west) to[out=250,in=90]
  node[techlabel,left=1pt,pos=0.58] {}
  (t43.north);

\draw[dep]
  (l31718.south east) to[out=290,in=90]
  node[techlabel,right=1pt,pos=0.58] {}
  (t45.north);


\draw[dep]
  (p38.east) -- ++(0.85,0)
  |- node[techlabel,pos=0.30,above] {...}
  (l310.west);
\node[resultnode] (tqsg) at (-2.5,-10.8)
  {\cref{thm:pauli-improved-spectral-norm}\\
   Quantum spin glass spectral upper bound};

\draw[dep]
  (l39.west) -- (-8.0,-7.2) -- (-8.0,-10.8)
  -- node[techlabel,above=2pt,midway]{Pauli Sign Matrix Spectrum}
     node[techlabel,below=2pt,midway]{\cref{lem:pauli-sign-spec}}
  (tqsg.west);
\end{tikzpicture}

\caption{A map of the main results.}
\label{fig:leitfaden}
\end{figure}
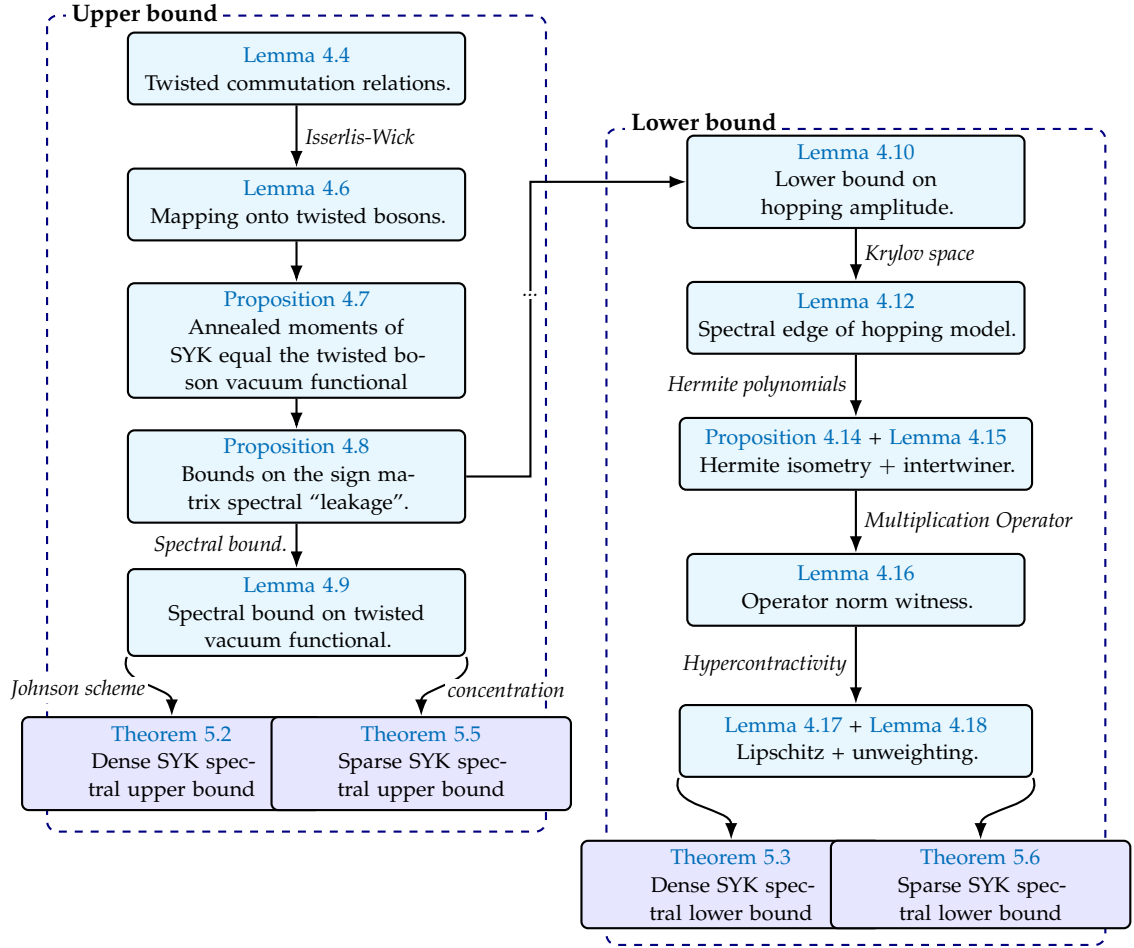

\newpage 

\section{Preliminaries}

\paragraph{Majorana Operators} Let $n\geq 2$ be an even integer. Consider the \emph{Majorana operators} $\gamma_1, \ldots, \gamma_n$ satisfying the following properties:
\begin{enumerate}[(1)]
    \item \textbf{Hermitian} $\gamma_i = \gamma_i^*$ for all $i\in[n]$, i.e.\ $\gamma_i$ is Hermitian, 
    \item \textbf{Involution} $\gamma_i^2 = \Id$ for all $i\in[n]$, 
    \item \textbf{Anti-Commutation} $\gamma_i\gamma_j = -\gamma_j\gamma_i$ for all $i\neq j\in[n]$. 
\end{enumerate}
The algebra generated by $\gamma_1, \ldots, \gamma_n$ admits an irreducible unitary representation into $\C^{N\times N}, N:= 2^{n/2}$ \cite{BravyiK02}.\footnote{While there is no unique unitary irreducible representation, all irreducible representations are unitarily equivalent, and thus share the same spectrum. Since we will only be concerned with spectral properties in this paper, we fix an arbitrary such unitary irreducible representation, such as, for example, the Jordan-Wigner transform.}

For any non-empty $S = \{s_1 < \cdots < s_r\}\seq[n]$, define the Majorana monomial as,
\[\Gamma_S:= \iota^{r(r - 1)/2}\gamma_{s_1}\cdots\gamma_{s_r},\]
where $\iota:= \sqrt{-1}$. Let $\Gamma_\emptyset := \Id$.

Using the properties of Majorana operators, we have that $\Gamma_S$ is a unitary Hermitian operator with $\Gamma_S^2 = \Id$. Moreover, for any two non-empty sets $S, T$ we have
\begin{equation}\label{eq:comm-rel}
    \Gamma_S\Gamma_T = \eps_{S, T}\cdot\Gamma_T\Gamma_S.
\end{equation}
where for any $S, T\seq[n]$, we define $\eps_{S, T}:= (-1)^{|S|\cdot|T| - |S\cap T|}$. Note that $\eps_{S, T} = \eps_{T, S}$ for all $S, T\seq[n]$, and $\eps_{S, S} = 1$ for all $S\seq[n]$. 

\paragraph{The SYK Hamiltonian} Let $n\geq k\geq 2$ be even integers, and let $\cH$ be a $k$-uniform hypergraph. Let $\{g_S\}_{S\in\cH}$ be i.i.d $\cN(0, 1)$ Gaussians. Then we define the \emph{SYK Hamiltonian} to be 
\begin{equation}\label{eq:syk-def}
    H^\cH_{\operatorname{SYK}}:= \frac{1}{\sqrt{|\cH|}}\sum_{S\in\cH}g_S\Gamma_S.
\end{equation}
If $\cH$ is unspecified, $\cH$ is to be inferred from context. The normalization ensures that $\E (H^\cH_{\operatorname{SYK}})^2 =\Id$. 

\paragraph{Miscellaneous} Denote by $\tr(\cdot)$ the normalized trace operator. For instance, $\tr(\Id) = 1$. Also let $\Tr(\cdot)$ be usual (unnormalized) trace operator. For instance, $\Tr(\Id_N) = N$ for $\Id_N\in\C^{N\times N}$. For matrices $X, Y$, define the commutator $[X, Y]:= XY - YX$ and the anti-commutator $\{X, Y\}:= XY + YX$. For a matrix $X$ define the Hilbert-Schmidt norm as $\|X\|_{\operatorname{HS}}:= \sqrt{\tr(X^*X)}$. 

For a vector $v\in \C^N$ and a Hermitian matrix $T\in \C^{N\times N}$ we call $\langle v, Tv\rangle $ the (unnormalized) expectation value of $T$ on $v$. The normalized expectation value defined for $v\neq 0$, is $\langle v, Tv\rangle/\langle v,v\rangle$. If $T$ is a Hamiltonian, we call the normalized expectation value the \emph{energy} under that Hamiltonian.

For any $\eta\in\R$, write $\eta_+:= \max\{\eta, 0\}$. 

For any integers $n, t$ with $n\geq 0$ define $\binom{n}{t}:= 0$ if $t < 0$ or $t > n$. For $n = 0$, we define $\binom{0}{0}:= 1$, and $\binom{0}{t}:= 0$ for all integers $t\geq 1$. Also, for any $n\geq j\geq 1$, write $(n)_j:= n(n - 1)\cdots (n - j + 1)$. Also define $(n)_0:= 1$. Thus $\binom{n}{k} = (n)_k/k!$. 

\paragraph{The Johnson Scheme} Fix integers $n\geq k\geq t\geq 0$, and consider the matrix $M_t\in\C^{\binom{[n]}{k}\times\binom{[n]}{k}}$ given as $M_t(S, T):= \binom{|S\cap T|}{t}$. The span of the matrices $\{\sum_{t = 0}^k\alpha_tM_t:\alpha_0, \ldots, \alpha_k\in\C\}$ is the Bose-Mesner algebra of the \emph{Johnson scheme}. For any matrix $M = \sum_{t = 0}^k\alpha_tM_t$, $M(S, T)$ depends only on $|S\cap T|$. 

The Johnson scheme has been intensely studied in various contexts in mathematics and computer science, and as such a detailed understanding of the spectrum of matrices in the Johnson scheme is available; see Godsil and Meagher~\cite[Chapter 6]{godsil2015erdos} and also Delsarte~\cite{delsarte1973algebraic} for the foundational association-scheme treatment and Filmus~\cite{Filmus16} for an explicit orthogonal basis realizing the Johnson eigenspaces. We recall the following standard fact about the spectrum of $M_t$ for our purposes:
\begin{fact}[Johnson eigenspaces, \cite{delsarte1973algebraic,Wilson90}]\label{fact:johnson-eigenspaces}
    There exists an orthogonal decomposition $\C^{\binom{[n]}{k}} = \bigoplus_{j = 0}^k V_j$ such that $V_j$ is an eigenspace for $M_t$ for all $0\leq j, t\leq k\leq n/2$. Moreover,
    \begin{enumerate}[(1)]
        \item $\dim(V_j) = \binom{n}{j} - \binom{n}{j - 1}$ for all $j\geq 0$, where $\binom{n}{-1}:= 0$. Furthermore, $V_0 = \spn\{\1\}$.
        \item $M_t\vert_{V_j} = \mu_{t, j}\cdot\Id_{V_j}$ where 
        \[\mu_{t, j}:= \begin{cases}
        \binom{k - j}{t - j}\cdot\binom{n - t - j}{k - t} & j\leq t,\\
            0 & j > t
        \end{cases}.\]
        \item $V_j$ is $\mathfrak{S}_n$-invariant for all $0\leq j\leq k$, i.e. for any $\psi\in V_j$ and any $\sigma\in\mathfrak{S}_n$ we have $\sigma\cdot\psi\in V_j$. Here $\mathfrak{S}_n$ is the symmetric group on $n$ points, and for any $\psi\in\C^{\binom{[n]}{k}}$, $\sigma\in \mathfrak{S}_n$ naturally acts on $\psi$ as $(\sigma\cdot\psi)(S):= \psi(\sigma^{-1}(S))$, where $\sigma^{-1}(S):= \{\sigma^{-1}(i): i\in S\}$. 
    \end{enumerate}
\end{fact}

\paragraph{The Matrix Bernstein Bound} We recall the following variant of the classical matrix Bernstein bound which we will need:
\begin{fact}[Matrix Bernstein, Theorem 7.3.1 in \cite{Tro15}]\label{fact:matrix-Bernstein}
    Let $Y_1, \ldots, Y_m$ be independent mean $0$ Hermitian matrices, and write $Y:= \sum_{k = 1}^mY_i$. Suppose $\max_{k\in[m]}\|Y_k\|_{\operatorname{op}}\leq R$ almost surely, and suppose $\sum_{i = 1}^m\E[Y_i^2]\preceq \Sigma$. Write $\sigma^2:= \|\Sigma\|_{\operatorname{op}}, d:= 2\Tr(\Sigma)/\|\Sigma\|_{\operatorname{op}}$. Then
    \[\Pr\left(\left\|\sum_{i = 1}^mY_i\right\|_{\operatorname{op}}\geq t\right)\leq 2d\cdot\exp\left(-\frac{t^2}{2\sigma^2 + 2Rt/3}\right).\]
\end{fact}
We include a somewhat specialized use-case of the above fact here: Let $\Omega$ be a set of size $D$. Let $\Omega' = \{\omega_1, \ldots, \omega_m\}$ be a random sample of $\Omega$ of size $m$, where every element $\omega_i$ is chosen independently, and uniformly at random, from $\Omega$. Note that $\Omega'$ can have repeated elements in general. Define the linear map $P:\C^{\Omega'}\to\C^{\Omega}$, where we have $Pe_{\omega_i}:= e_{\omega_i}\in\C^\Omega$ for all $i\in[m]$, where $\{e_{\omega_i}\}_{i\in[m]}$ (resp. $\{e_\omega\}_{\omega\in\Omega}$) is the standard basis for $\C^{\Omega'}$ (resp. $\C^{\Omega}$).

We remind the reader that $\tr(G) = \Tr(G)/D$.
\begin{lemma}\label{lem:random-proj-compression}
    Let $G\in\C^{\Omega\times\Omega}$ be a non-zero Hermitian PSD matrix with $\|G\|_{\operatorname{op}} = \mu$. Further suppose that $G$ has a constant diagonal, i.e. $\langle e_\omega, Ge_\omega\rangle = \tr(G)$ for every $\omega\in\Omega$. Then for $A_G:= (D/m)P^*GP\in\C^{\Omega'\times\Omega'}$ we have with probability $\geq 1 - O(\exp(-100\log(2\Tr(G)/\mu)))$ that
    \[\|A_G\|_{\operatorname{op}}\leq O(1)\cdot\max\left\{\mu, \frac{\Tr(G)}{m}\cdot\log\left(\frac{2\Tr(G)}{\mu}\right)\right\}\]
\end{lemma}
\begin{proof}
    Note that $P^*GP$ is non-zero since $G$ has non-zero diagonal entries. Consequently, 
    \[\|P^*GP\|_{\operatorname{op}} = \|G^{1/2}PP^*G^{1/2}\|_{\operatorname{op}},\]
    where we use the fact that the non-zero spectrum of $BB^*$ and $B^*B$ coincides for any matrix $B$. Furthermore, we can write $PP^* = \sum_{i = 1}^{m}e_{\omega_i} e_{\omega_i}^*$, and thus write $X_i:= (D/m) G^{1/2}e_{\omega_i} e_{\omega_i}^*G^{1/2}$, and note that
    \[\|A_G\|_{\operatorname{op}} = \frac{D}{m}\|P^*GP\|_{\operatorname{op}} = \frac{D}{m}\|G^{1/2}PP^*G^{1/2}\|_{\operatorname{op}} = \left\|\sum_{i = 1}^mX_i\right\|_{\operatorname{op}}.\]
    Note that $\E[X_i] = G/m$. Also note that since $X_i$ is rank $1$, we have that $\|X_i\|_{\operatorname{op}} = \Tr(X_i) = (D/m)e_{\omega_i}^*Ge_{\omega_i} = \Tr(G)/m$. Now write $Y_i:= X_i - G/m$, and note that $\E[Y_i] = 0$. Also note that $\|Y_i\|_{\operatorname{op}}\leq\|X_i\|_{\operatorname{op}} + \|G\|_{\operatorname{op}}/m\leq\Tr(G)/m + \Tr(G)/m = 2\cdot\Tr(G)/m$ with probability $1$. Furthermore, 
    \[\E[Y_i^2] = \E[X_i^2] - \frac{G^2}{m^2} = \frac{\Tr(G)G}{m^2} - \frac{G^2}{m^2}\preceq \frac{\Tr(G)G}{m^2}\implies\sum_{i = 1}^m\E[Y_i^2]\preceq\frac{\Tr(G)G}{m}.\] 
    Now write 
    \[\sigma^2:= \left\|\frac{\Tr(G)G}{m}\right\|_{\operatorname{op}} = \frac{\mu\Tr(G)}{m}, d:= \frac{2\Tr\left(\frac{\Tr(G)G}{m}\right)}{\left\|\frac{\Tr(G)G}{m}\right\|_{\operatorname{op}}} = \frac{2\Tr(G)}{\mu},\]
    and note that by \cref{fact:matrix-Bernstein} we have with probability $\geq 1 - O(d^{-100})$ that 
    \[\|A_G\|_{\operatorname{op}}\leq \|G\|_{\operatorname{op}} + \left\|\sum_{i = 1}^mY_i\right\|_{\operatorname{op}}\] 
    \[\leq \mu + O\left(\frac{\Tr(G)}{m}\cdot\log\left(\frac{2\Tr(G)}{\mu}\right) + \sqrt{\frac{\Tr(G)\cdot\mu}{m}}\cdot\sqrt{\log\left(\frac{2\Tr(G)}{\mu}\right)}\right).\]
    Write $\xi:= \frac{\Tr(G)}{m}\cdot\log\left(\frac{2\Tr(G)}{\mu}\right)$. Then we have 
     \[\|A_G\|_{\operatorname{op}}\leq \mu + O(1)\cdot(\xi + \sqrt{\xi\mu})\leq O(1)\cdot\max\{\mu, \xi\},\]
    as desired.
\end{proof} 

\paragraph{Isserlis-Wick Theorem} For any integer $\ell\geq 1$, let $\cP_2(2\ell)$ be the set of perfect matchings on $[2\ell]$. More precisely, every $\pi\in\cP_2(2\ell)$ is a partition of $[2\ell]$ of the form $\Big\{\{i_1, i_2\}, \{i_3, i_4\}, \ldots, \{i_{2\ell - 1}, i_{2\ell}\}\Big\}$. For any $\pi = \Big\{\{i_1, i_2\}, \{i_3, i_4\}, \ldots, \{i_{2\ell - 1}, i_{2\ell}\}\Big\} = \Big\{e_1, e_2, \ldots, e_\ell\Big\}\in\cP_2(2\ell)$, where $e_t:= \{i_{2t - 1} < i_{2t}\}\in\binom{[2\ell]}{2}$ for all $t\in[\ell]$, we define the \emph{crossing graph} $G_\pi$ of $\pi$ to be a graph on the vertex set $[\ell]$, where $t_1, t_2\in[\ell]$ form an edge in $G_\pi$ if $e_{t_1}, e_{t_2}$ \emph{cross} each other. For two edges $e_1 = \{a < b\}, e_2 = \{c < d\}$, where $a, b, c, d\in[2\ell]$ are distinct elements, we say $e_1, e_2$ cross each other if either $a < c < b < d$ or $c < a < d < b$.

We can now recall the following classic result about moments of Gaussian random variables:
\begin{fact}[Isserlis-Wick Theorem \cite{Isserlis18,Wick50}]\label{fact:isserlis-wick}
    Let $(X_1, \ldots, X_{2\ell})$ be a mean $0$ Gaussian vector. Then 
    \[\E[X_1\cdots X_{2\ell}] = \sum_{\pi\in\cP_2(2\ell)}\prod_{e = \{i, j\}\in\pi}\E[X_iX_j].\]
\end{fact}

\paragraph{Probabilist's Hermite Polynomials} 

\begin{definition}[Normalized Probabilist's Hermite Polynomials]\label{def:hermite-polynomial}
    For every $n\geq 0$ define the (normalized) probabilist's Hermite polynomials $h_n(X)\in\R[X]$ as $h_{-1}(X):= 0, h_0(X):= 1, h_1(X):= X$, where for $n\geq 1$ we have the recurrence $Xh_{n}(X):= \sqrt{n}\cdot h_{n - 1}(X) + \sqrt{n + 1}\cdot h_{n + 1}(X)$. 
\end{definition}
Note that $\deg(h_n) = n$ for all $n\geq 0$. Furthermore, $\spn\{h_t(X):t\leq n\} = \spn\{X^t:t\leq n\}$. This is also referred to as the Hermite polynomials having the \emph{same degree filtration} as the degree monomials.

Furthermore, 
\[\E_{g\sim\cN(0, 1)}\left[h_n(g)h_m(g)\right] = \1(n = m).\]
In general, we can also define multivariate Hermite polynomials:
\begin{definition}[Normalized Multivariate Probabilist's Hermite Polynomials]\label{def:hermite-polynomial-multivar}
    For a tuple $\bm{r}\in\Z_{\geq 0}^\Lambda, X\in\C^\Lambda$, we can define $h_{\bm{r}}(X):= \prod_{\lambda\in\Lambda}h_{\bm{r}_\lambda}(X_\lambda)$. 
\end{definition}

Then we have 
\begin{equation}\label{eq:tuple-hermite-ortho}
    \E_{\bm{g}\sim\cN(0, \Id_\Lambda)}\left[h_{\bm{r}}(\bm{g})h_{\bm{s}}(\bm{g})\right] = \prod_{\lambda\in\Lambda}\E_{\bm{g}_\lambda\sim\cN(0, 1)}\left[h_{\bm{r}_\lambda}(\bm{g}_\lambda)h_{\bm{s}_\lambda}(\bm{g}_\lambda)\right] = \1(\bm{r} = \bm{s}).
\end{equation}
See \cite[Section~11.2]{ODonnell21} for an exposition of univariate and multivariate Hermite polynomials.

\paragraph{Matrix-Valued Polynomials} Let $\cK(\R^m;\C^{N\times N})$ denote all matrix-valued polynomials on $\R^m$. More precisely, given $\bm g = (g_1, \ldots, g_m)\in\R^m$, a matrix-valued polynomial $p(\bm g)\in\cK(\R^m;\C^{N\times N})$ denotes a polynomial of the form 
\[\sum_{\bm r\in\Z_{\geq 0}^m}\left(\prod_{i = 1}^m g_i^{\bm r_i}\right)\cdot M_{\bm r},\]
where $M_{\bm r}\in\C^{N\times N}$, and $M_{\bm r}\neq 0$ only for finitely many $\bm r\in\Z_{\geq 0}^m$. For any $\bm r\in\Z_{\geq 0}^m$, write $|\bm r| = \sum_{i = 1}^m\bm r_i$. For any $\psi\in\cK(\R^m;\C^{N\times N})$, we define the \emph{degree} of $\psi$ to be $\max_{M_{\bm r}\neq 0}|\bm r|$, where $\psi(\bm g) = \sum_{\bm r\in\Z_{\geq 0}^m}\left(\prod_{i = 1}^m g_i^{\bm r_i}\right)\cdot M_{\bm r}$. 

We can naturally equip $\cK(\R^m;\C^{N\times N})$ with the Gaussian measure, where we define 
\[\langle F(\bm g), G(\bm g)\rangle:= \E_{\bm g\sim\cN(0,\Id_m)}\tr(F(\bm g)^*G(\bm g)).\]

\paragraph{Gaussian Concentration and Gaussian Chaos}
For a $\R$-valued random variable $X$ and for any $p\geq 1$, define $\|X\|_{p}:= \E[|X|^p]^{1/p}$. Recall that $\|\cdot\|_p$ norms satisfy \emph{H\"{o}lder's inequality}, i.e. if $X_1, X_2$ are two $\R$-valued random variables, then $\|X_1X_2\|_1\leq\|X_1\|_p\cdot\|X_2\|_{p'}$, for any $p\geq 1$, where $p^{-1} + (p')^{-1} = 1$.

More generally, let $(V, \|\cdot\|)$ be a Hilbert space, and let $X$ be a $V$-valued random variable. Then define $\|X\|_p:= \E[\|X\|^p]^{1/p}$.
\begin{fact}[Concentration of Lipschitz Functions \cite{Vershynin18}]\label{fact:gaussian-lipschitz-conc}
Let $f:\R^m\to\R$ be a $\sigma$-Lipschitz function, i.e. $|f(x) - f(y)|\leq \sigma\cdot\|x - y\|_2$ for all $x, y\in\R^m$. Consider the random variable $X:= f(\bm g) - \E_{\bm g}f(\bm g)$ for $\bm g\sim\cN(0, \Id_m)$. Then we have 
\[\Pr\left(|X|\geq t\right)\leq 2\cdot\exp\left(-\Omega(1)\cdot\frac{t^2}{\sigma^2}\right).\]
In particular, for every $p\geq 2$ we have $\|X\|_p\leq O(\sigma\sqrt{p})$. 
\end{fact}

Recall that we treat $\C^{N\times N}$ as a Hilbert space by equipping it with the inner product $\langle X, Y\rangle:= \tr(X^*Y)$. 

We need the following statement about Gaussian hypercontractivity. See \cite{Janson97,ODonnell21} for a detailed treatment, and \cite{Nelson66,Nelson73,Gross75} for a historical perspective.
\begin{proposition}[Gaussian Hypercontractivity]\label{fact:gaussian-hypercontractivity}
    Let $Y\in\cK(\R^m;\C^{N\times N})$ be a matrix-valued polynomial of degree $\leq L$. Define the random variable $X:= Y(\bm g)$, where $\bm g\sim\cN(0, \Id_m)$. Then $\|X\|_r\leq (r - 1)^{L/2}\cdot\|X\|_2$ for any $r\geq 2$.
\end{proposition}
\begin{proof}
For $\rho\in[0,1]$, let $U_\rho$ denote the Gaussian noise operator (see \cite[Definition~11.12]{ODonnell21}), defined for functions $f:\R^m\to\R$ as
\[(U_\rho f)(x):=\E_{\bm z\sim\cN(0,\Id_m)}f\left(\rho x+\sqrt{1-\rho^2}\,\bm z\right).\]
We apply $U_\rho$ to
matrix-valued functions entrywise.

We claim that 
\begin{equation}\label{eq:scalar-gaussian-2-r}
    \|U_\rho f\|_r\leq \|f\|_2
\end{equation}
for $\rho:= (r-1)^{-1/2}$. Indeed, write $r':=r/(r-1)$. By duality, we have
\[\|U_\rho f\|_r \overset{\text{Duality}}{=} \sup_{\|u\|_{r'}\leq 1} \big|\langle u,U_\rho f\rangle\big|\overset{\text{\cite[Section~11.1, Gaussian Hypercontractivity Theorem]{ODonnell21}}}{\leq} \sup_{\|u\|_{r'}\leq 1}\big|\langle u,U_\rho f\rangle\big|\] 
\[\leq\sup_{\|u\|_{r'}\leq 1}\|u\|_{1+a}\|f\|_{1+b}\]
for any $a,b\geq 0$ and $0\leq\rho\leq\sqrt{ab}\leq 1$. Choose 
\[
    a:=\frac{1}{r-1},
    \qquad
    b:=1,
    \qquad
    \rho:=\sqrt{ab}=\frac{1}{\sqrt{r-1}}.
\]
Then we have $1+a=r'$ and $1+b=2$. Hence
\[
    \big|\langle u,U_\rho f\rangle\big|
    \leq
    \|u\|_{r'}\|f\|_2\leq\|f\|_2,
\]
where the last inequality follows since while applying duality we took a supremum only over $u$ for which $\|u\|_{r'}\leq 1$. Now, by triangle inequality,
\begin{align*}
    \|(U_\rho F)(x)\|_{\mathrm{HS}}
    &=
    \left\|
        \E_{\bm z}
        F\bigl(\rho x+\sqrt{1-\rho^2}\,\bm z\bigr)
    \right\|_{\mathrm{HS}} \leq
    \E_{\bm z}
    \left\|
        F\bigl(\rho x+\sqrt{1-\rho^2}\,\bm z\bigr)
    \right\|_{\mathrm{HS}}=
    \bigl(U_\rho\|F(\cdot)\|_{\mathrm{HS}}\bigr)(x).
\end{align*}
Applying \eqref{eq:scalar-gaussian-2-r} to the nonnegative scalar-valued
function $x\mapsto\|F(x)\|_{\mathrm{HS}}$ gives
\begin{equation}\label{eq:matrix-gaussian-2-r}
    \|U_\rho F(\bm g)\|_r
    \leq
    \|F(\bm g)\|_2.
\end{equation}

Expand $Y$ in the multivariate Hermite basis as in the definition of $\cK(\R^m;\C^{N\times N})$ to get
\[
    Y(\bm g)
    =
    \sum_{\substack{\bm s\in\Z_{\geq 0}^m\\|\bm s|\leq L}}
    h_{\bm s}(\bm g)M_{\bm s},
    \qquad
    M_{\bm s}\in\C^{N\times N}.
\]
Such an expansion exists because the Hermite polynomials have the same
degree filtration as the degree monomials. Moreover by \cite[Proposition~11.33]{ODonnell21} we have
\[
    U_\rho h_{\bm s}
    =
    \rho^{|\bm s|}h_{\bm s}.
\]
Thus define
\[
    F(\bm g)
    :=
    \sum_{\substack{\bm s\in\Z_{\geq 0}^m\\|\bm s|\leq L}}
    \rho^{-|\bm s|}h_{\bm s}(\bm g)M_{\bm s}.
\]
It follows that $U_\rho F=Y$. Therefore, $
    \|Y(\bm g)\|_r
    \overset{\eqref{eq:matrix-gaussian-2-r}}{\leq}
    \|F(\bm g)\|_2.$

By the orthonormality of the multivariate Hermite polynomials
\cite[Proposition~11.33]{ODonnell21},
\begin{align*}
    \|F(\bm g)\|_{\cK, 2}^2
    &=
    \E_{\bm g}
    \tr\bigl(F(\bm g)^*F(\bm g)\bigr)=
    \sum_{|\bm s|\leq L}
    \rho^{-2|\bm s|}
    \tr(M_{\bm s}^*M_{\bm s})\leq
    \rho^{-2L}
    \sum_{|\bm s|\leq L}
    \tr(M_{\bm s}^*M_{\bm s})=
    \rho^{-2L}\|Y(\bm g)\|_2^2.
\end{align*}
Since $\rho^{-1}=\sqrt{r-1}$, we have $
    \|Y(\bm g)\|_r
    \leq
    \rho^{-L}\|Y(\bm g)\|_2
    =
    (r-1)^{L/2}\|Y(\bm g)\|_2,$
as desired.
\end{proof}

\section{The Twisted Bosonic Model}\label{sec:twisted-fock-space}

We define a bosonic model with a ``twisted'' phase factor. Such operators and their mathematical properties have also been studied in an abstract setting by \cite{DoresicMM94,MeljanacMP94,MeljanacM96}. 

Let us first review the standard bosonic Fock space. Given set $\Lambda$, we define for every $\bm{r} \in \Z_{\geq 0}^\Lambda$ the function $f_{\bm{r}}(\bm{r}') := \1(\bm{r} = \bm{r}')$. Write $\cF_t:= \spn\{f_{\bm{r}}:|\bm{r}| = t\}$, and let $\cF := \spn\{f_{\bm{r}}:\bm{r}\in\Z_{\geq 0}^\Lambda \}$.\footnote{Recall that when we take the (linear-algebraic) span of infinitely many vectors, we only look at linear combinations with finitely many non-zero coefficients, i.e. for every $f\in\cF$ there exists a finite set $T$ such that $f = \sum_{\bm{r}\in T}\alpha_{\bm{r}}f_{\bm{r}}$. Note that in physics Fock spaces are usually defined as $\ell_2$ spaces in which $f$ can have infinitely many non-zero coefficients. We choose our current definition so as to avoid functional-analytic subtleties.} We declare $\{f_{\bm r}\}_{\bm r\in\Z_{\geq 0}^\Lambda}$ to be an orthonormal basis of $\cF$, and all subsequent adjoints are taken w.r.t. this basis.
\begin{definition}[Bosonic Operators]\label{def:bosonop}
For each $i\in \Lambda$, the annihilation and creation operators are defined by
\[
\sfb_i f_{\bm{r}}
:=
\sqrt{\bm{r}_i}\,f_{\bm{r}-e_i},
\qquad
\sfb_i^* f_{\bm{r}}
:=
\sqrt{\bm{r}_i+1}\,f_{\bm{r}+e_i},
\]
where, $\bm{r}+e_i$ (resp. $\bm{r}-e_i$) increments (resp. decrements, when $\bm{r}_i>0$) the
$i^{\mathrm{th}}$ entry of $\bm{r}$ by $1$. The number operator for mode $i$ is defined as $\Num_i := \sfb_i^*\sfb_i$, and the total number operator is defined as $\Num := \sum_{i\in\Lambda} \Num_i$. Furthermore, we use the shorthand $(\pm 1)^{\Num_i}:\cF\to\cF$ to refer to the linear operator which acts as $(\pm 1)^{\Num_i}f_{\bm r}:= (\pm 1)^{\bm r_i}f_{\bm r}$. Note that $1^{\Num_i} = \Id$.
\end{definition}

These satisfy the following standard commutation relations. 

\begin{lemma}[Bosonic commutation relations]\label{lem:bosonopcomm}
For all \(i\neq j\in\Lambda\),
\[
[\sfb_i,\sfb_j]=[\sfb_i^*,\sfb_j^*]=0,
\qquad
[\sfb_i,\sfb_j^*]=\delta_{ij}\Id.
\]
Moreover,
\[
[\sfb_i,\Num_j]=\delta_{ij}\sfb_i,
\qquad
[\sfb_i^*,\Num_j]=-\delta_{ij}\sfb_i^*,
\]
and
\[
\{\sfb_i,(-1)^{\Num_i}\}
=
\{\sfb_i^*,(-1)^{\Num_i}\}
=
0.
\]
\end{lemma}

\begin{proof}
It suffices to check the identities on \(f_{\bm r}\). First,
\[
\sfb_i\sfb_i^*f_{\bm r}
=
(\bm r_i+1)f_{\bm r},
\qquad
\sfb_i^*\sfb_i f_{\bm r}
=
\bm r_i f_{\bm r},
\]
so \([\sfb_i,\sfb_i^*]=\Id\). If \(i\neq j\), then
\[
\sfb_i\sfb_jf_{\bm r}
=
\sqrt{\bm r_i\bm r_j}\,
f_{\bm r-e_i-e_j}
=
\sfb_j\sfb_i f_{\bm r},
\]
and
\[
\sfb_i^*\sfb_j f_{\bm r}
=
\sqrt{(\bm r_i+1)\bm r_j}\,
f_{\bm r+e_i-e_j}
=
\sfb_j\sfb_i^*f_{\bm r}.
\]
The creation operators commute by the same calculation.

For the number operators,
\[
[\sfb_i,\Num_j]f_{\bm r}
=
\delta_{ij}\sqrt{\bm r_i}\,f_{\bm r-e_i}
=
\delta_{ij}\sfb_i f_{\bm r},
\]
and similarly
\[
[\sfb_i^*,\Num_j]f_{\bm r}
=
-\delta_{ij}\sfb_i^*f_{\bm r}.
\]

Finally, since
\[
(-1)^{\Num_i}f_{\bm r}
=
(-1)^{\bm r_i}f_{\bm r},
\]
changing \(\bm r_i\) by one reverses its parity. Hence
\[
(-1)^{\Num_i}\sfb_i
=
-\sfb_i(-1)^{\Num_i},
\qquad
(-1)^{\Num_i}\sfb_i^*
=
-\sfb_i^*(-1)^{\Num_i},
\]
which proves the anti-commutation relations.
\end{proof}

\begin{table}[t]
\centering
\renewcommand{\arraystretch}{1.45}
\begin{tabular}{|c|p{0.38\textwidth}|p{0.43\textwidth}|}
\hline
\textbf{Object}
&
\textbf{Definition}
&
\textbf{Interpretation}
\\
\hline

$f_{\bm r}$
&
$\displaystyle
f_{\bm r}(\bm r')
:=
\1(\bm r=\bm r')
$
&
Fock basis state associated with the occupation vector
$\bm r=(\bm r_i)_{i\in\cH}$. 
\\
\hline

$\sfb_i$
&
$\displaystyle
\sfb_i f_{\bm r}
:=
\sqrt{\bm r_i}\,
f_{\bm r-e_i}
$
&
The usual bosonic annihilation operator on mode $i$.
\\
\hline

$\Num_i$
&
$\displaystyle
\Num_i
:=
\sfb_i^{*}\sfb_i,
\qquad
\Num_i f_{\bm r}
=
\bm r_i f_{\bm r}
$
&
The local number operator, which counts the number of bosons occupying
mode $i$.
\\
\hline

$(\pm1)^{\Num_i}$
&
$\displaystyle
(\pm1)^{\Num_i}f_{\bm r}
:=
(\pm1)^{\bm r_i}f_{\bm r}
$
&
The occupation-parity operator on mode $i$. 
\\
\hline

$\sfK_i$
&
$\displaystyle
\sfK_i
:=
\prod_{j<i}
\eps_{ij}^{\Num_j}
$
&
The twist operator associated with mode $i$.
\\
\hline

$\sfa_i$
&
$\displaystyle
\sfa_i
:=
\sfK_i\sfb_i
$
&
The twisted bosonic annihilation operator.
\\
\hline

\end{tabular}

\caption{Basic objects in the twisted boson algebra. For hyperedge $S_i\in\cH$, and some operator $\mathsf{O}$, we use the shorthand $\mathsf{O}_i$ for $\mathsf{O}_{S^{(i)}}$. }
\label{tab:twisted-boson-algebra}
\end{table}

Now, let us define the twisted bosonic operators on the set of hyperedges $\cH$. Let $\cH = (S^{(1)}, \ldots, S^{(m)})$ be some arbitrary ordering of $\cH$ (where $m = |\cH|$). Write $\eps_{S^{(i)}, S^{(j)}}:= \eps_{ij}$. 
\begin{definition}[Twisted Bosonic Operators]\label{def:twisted-bosonic}
    For each $S^{(i)}\in \cH$, we define the \emph{twist operator} $\sfK_i:\cF\to\cF$ as 
    \[\sfK_i := \prod_{j < i} \eps_{ij}^{\Num_j}.\]
    The twisted creation and annihilation operators are then given as, 
    \[\sfa_i := \sfK_i \sfb_i  \qquad\sfa_i^* := \sfK_i \sfb_i^*   \]
\end{definition}

Individually, the twisted operators satisfy bosonic commutation relations, $[\sfa_i, \sfa_i^*] = \Id$. The twisted construction is summarized in \cref{tab:twisted-boson-algebra}.

\begin{lemma}\label{lem:twisted-commutations}
    The twisted bosonic operators $\{\sfa_S\}_{S\in\cH}$ on $\cF$ satisfy for all $S, T\in\cH$, 
    \begin{enumerate}[(1)]
        \item\label{item:twisted-commutation} $\sfa_S\sfa_T = \eps_{S, T}\sfa_T\sfa_S$,
        \item\label{item:twisted-bosonic-relation} $\sfa_S\sfa_T^* - \eps_{S, T}\sfa_T^*\sfa_S = \1(S = T)\cdot\Id$, 
        \item\label{item:annihilation} $\sfa_S(\cF_t)\seq\cF_{t - 1}$ for all $t\geq  1$. Furthermore, $\sfa_S(\cF_0) = 0$. Similarly, $\sfa^*_S(\cF_t)\seq\cF_{t + 1}$ for all $t\geq  0$.
        \item\label{item:number-operator} $\sfa^*_S \sfa_S f_{\bm{r}} = \bm{r}_S f_{\bm{r}}$ for all $\bm{r}\in\Z_{\geq 0}^\cH$. Consequently, $\sum_{S\in\cH}\sfa_S^* \sfa_S = \Num$. 
    \end{enumerate}
\end{lemma}
\begin{proof}
First, we note that $\sfa_i^2 = \sfb_i^2$ and $\eps_{ii} =+1$. Secondly, note that for $1 \le i < j \le m$ we have,
\[
\sfK_j \sfb_i = \eps_{ij}\sfb_i \sfK_j
\]
as $\{\sfb_i, (-1)^{\Num_i}\} = 0$ and $[\sfb_i, (-1)^{\Num_t}]=0$ for $t\neq i$. Thus, for any $1 \le i < j \le m$ we have, 
\[
\sfa_i \sfa_j = \sfb_i \sfK_i \sfb_j \sfK_j = \sfb_i\sfb_j  \sfK_i \sfK_j
\]
and 
\[
\sfa_j \sfa_i = \sfb_j\sfK_j \sfb_i \sfK_i = \sfb_j \left(\eps_{ij}\sfb_i \sfK_j\right) \sfK_i = \eps_{ij}\sfb_i\sfb_j \sfK_i \sfK_j
\]
where we used that $[\sfb_i, \sfb_j] = 0$. This shows \cref{item:twisted-commutation}. 

For \cref{item:twisted-bosonic-relation}, first note that $\sfa_i^*\sfa_i = \sfb_i^*\sfb_i$ and $\sfa_i \sfa_i^* = \sfb_i \sfb_i^*$. For $1 \le i < j \le m$, the calculation proceeds exactly as \cref{item:twisted-commutation} since $[\sfb_i, \sfb_j^*] = 0$.

\cref{item:annihilation,item:number-operator} follow simply from the definition of the twisted bosons.
\end{proof}

For any Hermitian $R\in\C^{\cH\times\cH}$, define 
\[\cL(R):= \sum_{S, T\in\cH}R_{S, T}\sfa_S^*\sfa_T.\]
Also, for any two Hermitian linear operators $L, L':\cF\to\cF$ we write $L\preceq L'$ if for any $t\in\Z_{\geq 0}$ and any $\psi\in\bigoplus_{\ell\leq t}\cF_\ell$ we have $\langle\psi, L\psi\rangle\leq\langle\psi, L'\psi\rangle$.

We now elucidate certain properties of the operator $\cL$:
\begin{lemma}[Properties of $\cL$]\label{lem:cL-properties}
    We have:
    \begin{enumerate}[(1)]
        \item\label{item:order-preserving} If $R\preceq R'$, then $\cL(R)\preceq\cL(R')$.
        \item\label{item:identity-number} $\cL(\Id) = \Num$.
    \end{enumerate}
\end{lemma}
\begin{proof}
    Since $\cL(R') - \cL(R) = \cL(R' - R)$, it suffices to show that for any $K\succeq 0$, we have $\cL(K)\succeq 0$. Thus write $K = Q^*Q$, and note that 
    \[\cL(K) = \sum_{S\in\cH}\left(\sum_{T\in\cH}Q_{S, T}a_T\right)^*\cdot\left(\sum_{T\in\cH}Q_{S, T}a_T\right)\succeq 0,\]
    as desired.

    Note that $\cL(\Id) = \sum_{S\in\cH}a_S^*a_S \overset{\text{\cref{item:number-operator} of \cref{lem:twisted-commutations}}}{=} \Num$.
\end{proof}

We declare $f_{\bm 0}$ to be the \emph{vacuum state} (since it has no particles), and we define the \emph{vacuum} functional 
\[\varphi(T):= \langle f_{\bm{0}}, Tf_{\bm{0}}\rangle\]
computing the expectation value of the operator $T$ on the vacuum state $\bm{0}:= (0, 0, \ldots, 0)\in\Z_{\geq 0}^\cH$. With this, we are ready to present the twisted bosonic mapping. The key idea is to encode the signs resulting from the anti-commutation relations of the Majorana monomials into the twist phases of the bosons. 

Now, we define the collective operators
\[\sfa := \frac{1}{\sqrt{|\cH|}}\sum_{S\in\cH}\sfa _S, \;\; \sfx := \sfa  + \sfa ^*.\]
Also write $\sfx_S:= \sfa_S + \sfa_S^*$. Note that $\sfx$ is Hermitian. With this, we express the average moments of the SYK model in terms of vacuum expectation values of powers of $\sfx$ on the twisted bosonic space.

In \cref{lem:twisted-boson-mapping}, the Gaussians $g_{S_i}$ should just be thought of as a tool to introduce Isserlis-Wick-style summations.

\begin{lemma}[Twisted Bosonic Mapping]\label{lem:twisted-boson-mapping}
For any $S_1, \dots, S_{2\ell} \in \cH$, we have
    \[
    \varphi(\sfx_{S_1}\cdots \sfx_{S_{2\ell}}) = \E[g_{S_1}g_{S_2}\dots g_{S_{2\ell}}]\cdot \tr[\Gamma_{S_1}\Gamma_{S_2}\dots \Gamma_{S_{2\ell}}]. 
    \]
\end{lemma}
\begin{proof}
    Using \cref{fact:isserlis-wick}, we have 
    \[\E[g_{S_1}g_{S_2}\dots g_{S_{2\ell}}] = \sum_{\pi\in \cP_2(2\ell)}\prod_{\{i,j\} \in \pi} \E[g_{S_i}g_{S_j}] = \sum_{\pi\in \cP_2(2\ell)}\prod_{\{i,j\} \in \pi} \1[S_i = S_j]\]
    where we used the fact that the disorder variables are zero-mean independent Gaussians, thus $\E[g_{S_i}g_{S_j}] =\1[S_i = S_j]\ $.

    For any perfect matching $\pi\in\cP_2(2\ell)$ and a collection of hyperedges $S_1,\dots,S_{2\ell}\in\cH$ define,
     \[
     \1_\pi(S_1,\cdots, S_{2\ell}):=\prod_{\{i, j\}\in\pi}\1(S_i = S_j)
     \]
     On the event that $\1_\pi(S_1,\ldots, S_{2\ell}) = 1$, for any edge $e\in\pi$ denote $S_e$ to be the unique hyperedge. We define the signed crossing phase for any $1_\pi(S_1,\cdots, S_{2\ell}) = 1$ event to be
     \[
     \phi_\pi(S_1,\ldots, S_{2\ell}) := \prod_{\substack{\{e,f\}\subset \pi \\ e \text{ crosses} f}} \eps_{e,f}
     \]
     where the product is over unordered pairs of crossing edges, and we define $\eps_{e,f} := \eps_{S_e,S_f}$.

    The proof will proceed in two parts. We will show, 
    \begin{align*}
 \Gamma_{S_1}\cdots\Gamma_{S_{2\ell}} &= \phi_\pi(S_1,\ldots, S_{2\ell})  \cdot \Id \text{ whenever } \1_\pi(S_1,\ldots, S_{2\ell})=1 \text{ and, } \tag{A} \\ 
        \varphi(\sfx_{S_1}\cdots \sfx_{S_{2\ell}}) &= \sum_{\pi\in\cP_2(2\ell)}\1_\pi(S_1,\ldots, S_{2\ell})\cdot \phi_\pi(S_1,\ldots, S_{2\ell}). \tag{B}\\ 
    \end{align*}
      Taken together, (A) and (B) imply that
    \[
    \varphi(\sfx_{S_1}\cdots \sfx_{S_{2\ell}}) = \sum_{\pi \in \cP_2(2\ell)}\1_\pi(S_1,\ldots, S_{2\ell})\tr[\Gamma_{S_1}\Gamma_{S_2}\cdots \Gamma_{S_{2\ell}}]
    \]
    proving the result.

     Let us first show (A). We show this by induction on $\ell$. For $\ell=1$, this is true since $\Gamma_S^2=\Id$. Assume this is true for $\ell-1$, and consider the case $\ell$. Fix $\pi\in \cP_2(2\ell)$ and let $i$ be the partner of $1$ in $\pi$. Then,
   
     \begin{align*}
       \Gamma_{S_1}\Gamma_{S_2}\cdots \Gamma_{S_i} \cdots \Gamma_{S_{2\ell}} &= \left(\prod_{j=2}^{i-1}\eps_{1j}\right)\Gamma_{S_2}\cdots \underbrace{\Gamma_{S_1}\Gamma_{S_i}}_{=\Id} \cdots \Gamma_{S_{2\ell}} \\ 
       &=\left(\prod_{j=2}^{i-1}\eps_{1j}\right) \cdot \phi_{\pi'}(S_2,\ldots, S_{i-1}, S_{i+1}, \ldots, S_{2\ell}) \cdot \Id 
     \end{align*}
     where $\pi' = \pi\setminus \{\{1,i\}\}$ and we used the inductive hypothesis on $\Gamma_{S_2}\dots \Gamma_{S_{i-1}}\Gamma_{S_{i+1}}\dots \Gamma_{2\ell}$, since $\1_\pi(S_1,\dots,S_{2\ell}) =1 \implies \1_{\pi'}(S_2,\dots,S_{i-1}, S_{i+1},\dots, S_{2\ell})=1$. Now consider a $\{p<q\}\in \pi'$ with $p<i$. If $1 < p < q < i$, then the factors satisfy $\eps_{1p}\eps_{1q}=1$. Otherwise, $p < i < q$, and the edges $e=\{1<i\}$ and $e'=\{p < q\}$ are crossing and contributes precisely $\eps_{ee'}$. Therefore, 
     \[
     \left(\prod_{j=2}^{i-1}\eps_{1j}\right) \cdot \phi_{\pi'}(S_2,\ldots, S_{i-1}, S_{i+1}, \ldots, S_{2\ell})  = \phi_{\pi}(S_1,S_2,\ldots, S_{2\ell}).
     \]
     This proves (A).

     To show (B), first write $\sfa_i := \sfa_{S_i}$ and $\sfx_i:= \sfx_{S_i} = \sfa_i + \sfa^*_i$, and consider
     \[
     \varphi(\sfx_{1}\cdots \sfx_{{2\ell}}) =   \langle f_{\bm{0}}, \left((\sfa_{1} + \sfa_{1}^*) \dots (\sfa_{{2\ell}} + \sfa_{{2\ell}}^*)\right)f_{\bm{0}}\rangle
     \]
     
     By \cref{lem:twisted-commutations}, we have that 
  \begin{equation}\label{eq:ax-comm}
      \sfa_i \sfx_j = \sfa_i(\sfa_j + \sfa_j^*) = (\sfa_j + \sfa_j^*)\sfa_i\eps_{ij} + \1[S_i=S_j] \Id = \sfx_j \sfa_i \eps_{ij} + \1[S_i=S_j] \Id
  \end{equation}

     Now, for any $f\in \cF$, we have that $\langle f_{\bm{0}}, (\sfa+\sfa^*) f\rangle =\langle f_{\bm{0}}, \sfa f\rangle$ since $\sfa f_{\bm{0}} = 0$. 
     We again proceed by induction. For $\ell=1$, note that since $\sfa_i\vac =0$ and $\sfa_i\sfa_i^*\vac=\vac$,
     \[
     \langle \vac, \sfx_i \sfx_j\vac\rangle = \langle \vac, \sfa_i\sfa_j^*\vac\rangle = \1[S_i=S_j] 
     \]
     Now, assume that (B) holds for $\ell-1$. We have, by \cref{eq:ax-comm},
     \begin{align*}
      \varphi(\sfx_{1}\cdots \sfx_{{2\ell}}) &= \langle \vac, \sfa_1 \sfx_2 \dots \sfx_{2\ell} \vac\rangle \\
        &= \eps_{12} \langle \vac,  \sfx_2 \sfa_1 \dots \sfx_{2\ell} \vac\rangle + \1[S_1=S_2] \varphi(\hat{\sfx}_2\sfx_3 \dots \sfx_{2\ell}) 
     \end{align*}
     where we denote omission of an operator $\sfx$ by $\hat{\sfx}$. Repeating this process until $\sfa_1$ is moved to the rightmost position, and subsequently invoking the inductive hypothesis

     \begin{align*}
         \varphi(\sfx_{1}\cdots \sfx_{{2\ell}}) &= \sum_{i=2}^{2\ell} \1[S_1 = S_i] \left(\prod_{j<i}\eps_{1j}\right)\varphi(\sfx_{2}\cdots \hat{\sfx}_{i}, \sfx_{i+1}\cdots \sfx_{{2\ell}}) \\ 
         &= \sum_{i=2}^{2\ell} \1[S_1 = S_i] \left(\prod_{j<i}\eps_{1j}\right) \sum_{\pi'\in \cP_2(2\ell-2)}\1_{\pi'}(S_2,\ldots,\hat{S}_i,\ldots S_{2\ell})\cdot \phi_{\pi'}(S_2,\ldots,\hat{S}_i,\ldots, S_{2\ell}) \\ 
         &=\sum_{\pi \in \cP_2(2\ell)}\1_\pi(S_1,S_2,\dots,S_{2\ell})\phi_\pi(S_1,S_2,\dots,S_{2\ell})
     \end{align*}     
     where the crossing phase factor on the partition $\pi$ is determined from that of $\pi'$ by the same bookkeeping as for part (A). 
\end{proof}
With \cref{lem:twisted-boson-mapping} we now show that the disorder-averaged moments of the SYK Hamiltonian map to the vacuum state functional of the operator $\sfx$ on the twisted bosons.
\begin{proposition}[SYK Moments]\label{prop:varphi-trace-re}
    $\varphi(\sfx^{2\ell}) = \E\tr(H^{2\ell})$, where $H = \frac{1}{\sqrt{|\cH|}}\sum_{S\in\cH}g_S\Gamma_S$ is the SYK Hamiltonian from \cref{eq:syk-def}.
\end{proposition}

\begin{proof}
       Note that 
    \[\E\tr(H^{2\ell}) = \frac{1}{|\cH|^\ell}\sum_{S_1, \ldots, S_{2\ell}\in\cH}\E\left[\prod_{i = 1}^{2\ell}g_{S_i}\right]\cdot\tr\left(\Gamma_{S_1}\cdots\Gamma_{S_{2\ell}}\right)\] 
    By \cref{lem:twisted-boson-mapping} we have,
    \[\varphi(\sfx_{S_1}\cdots \sfx_{S_{2\ell}}) = \E\left[\prod_{i = 1}^{2\ell}g_{S_i}\right]\cdot\tr\left(\Gamma_{S_1}\cdots\Gamma_{S_{2\ell}}\right).\]
    Combining this with the fact that,
    \[\varphi(\sfx^{2\ell}) = \frac{1}{|\cH|^\ell}\sum_{S_1, \ldots, S_{2\ell}\in\cH}\varphi(\sfx_{S_1}\cdots \sfx_{S_{2\ell}}),\]
    where $\sfx_S = \sfa_S + \sfa_S^*$, we conclude the proof.
\end{proof}
\subsection{Spectral Bounds for $\varphi(\sfx^{2\ell})$}

\begin{table}[t]
\centering
\renewcommand{\arraystretch}{1.4}
\begin{tabular}{|c|p{0.34\textwidth}|p{0.47\textwidth}|}
\hline
\textbf{Parameter}
&
\textbf{Definition}
&
\textbf{Interpretation}
\\
\hline

$\cE$
&
$\displaystyle
\cE_{S,T}
:=
|\cH|^{-1}\cdot\eps_{S,T}
=
|\cH|^{-1}\cdot(-1)^{|S\cap T|}\in\C^{\cH\times\cH}
$
&
Sign matrix encoding the Majorana-monomial commutation
structure. 
\\
\hline

$u$
&
$\displaystyle
u
:=
|\cH|^{-1/2}\cdot\1\in\C^\cH
$
&
Normalized uniform vector on the hyperedge set $\cH$. 
\\
\hline

$\Pi_{u^\perp}$
&
$\displaystyle
\Pi_{u^\perp}
:=
\Id - uu^*
$
&
Orthogonal projection operator onto the $u^\perp$ space
\\
\hline

$\delta$
&
$\displaystyle
\delta
:=
\left\|
\Pi_{u^\perp}\cE u
\right\|
$
&
Quantifies the extent to which $u$ is not an eigenvector of $\cE$. In particular, $\delta = 0$ iff $u$ is an eigenvector of $\cE$.
\\
\hline

$\rho^{+}$
&
$\displaystyle
\rho^{+}
:=
\lambda_{\max}
\left(
\Pi_{u^\perp}\cE\Pi_{u^\perp}
\right)
$
&
The largest positive eigenvalue in the $u^\perp$ space
\\
\hline

$\rho^{-}$
&
$\displaystyle
\rho^{-}
:=
\max\left\{
0,
-\lambda_{\min}
\left(
\Pi_{u^\perp}\cE\Pi_{u^\perp}
\right)
\right\}
$
&
The (magnitude of the) largest negative eigenvalue in the $u^\perp$ space
\\
\hline
\end{tabular}

\caption{``Spectral leakage'' parameters of the Sign Matrix.}
\label{tab:sign-matrix-parameters}
\end{table}

We now provide a spectral bound for $\varphi(\sfx^{2\ell})$ (see \cref{lem:varphi-spectral-dom} for a formal statement). First, we do away with some necessary definitions. We write, 
\begin{align}\label{eq:spectral-notation-def}
\cE&:=\frac1{|\cH|}\,(\eps_{S,T})_{S,T\in\cH}\in\C^{\cH\times\cH},\\
u&:= \1/\sqrt{|\cH|} = |\cH|^{-1/2}\cdot\begin{bmatrix}1 & 1 & \cdots & 1\end{bmatrix}^\top,\\
\Pi_{u^\perp}&:= \Id - uu^*,\\
q&:=\langle u,\cE u\rangle,\\
\delta&:=\|\Pi_{u^\perp}\cE u\|,\\
\rho^+&:=\lambda_{\max}\!\left(\Pi_{u^\perp}\cE\Pi_{u^\perp}\right),\\
\rho^-&:=\max\left\{0, -\lambda_{\min}\!\left(\Pi_{u^\perp}\cE\Pi_{u^\perp}\right)\right\}.
\end{align}
We shall call $\cE$ the \emph{sign matrix}. 

The vector $u$ can be interpreted as the ``uniform channel'', with $\rho^{\pm}$ measuring the leakage into the transversal channels. That is, $\delta=0$ and $\rho^{\pm} = 0$ imply that $\cE = q uu^*$. The parameter $\delta$ on the other hand measures how far $u$ is from being an eigenvector of $\cE$; for the complete hypergraph, we have $\delta=0$. These parameters are summarized in \cref{tab:sign-matrix-parameters}.
\begin{proposition}\label{prop:cE-dom}
We have
    \[(q - \delta)uu^* - (\rho^- + \delta)\Pi_{u^\perp}\preceq\cE\preceq(q + \delta)uu^* + (\rho^+ + \delta)\Pi_{u^\perp}.\]
\end{proposition}
\begin{proof}
    Consider an arbitrary $\psi\in\C^\cH$ and write $\psi = s u + y$, where $s\in\C, \langle y, u\rangle = 0$. Then 
    \[\langle\psi, \cE\psi\rangle = |s|^2q + 2\Ree(\bar{s}\cdot u^*\cE y) + \langle y, \cE y\rangle = |s|^2q + 2\Ree(\bar{s}\cdot u^*\cE\Pi_{u^\perp}y) + \langle y, \cE y\rangle\] 
    \[\leq |s|^2q + 2\Ree(\bar{s}\cdot u^*\cE\Pi_{u^\perp}y) + \rho^+\|y\|^2\leq |s|^2q + 2|s|\cdot\|\Pi_{u^\perp}\cE u\|\cdot\|y\| + \rho^+\|y\|^2\] 
    \[= |s|^2q + 2\delta|s|\cdot\|y\| + \rho^+\|y\|^2\leq |s|^2q + \delta|s|^2 + \delta\|y\|^2 + \rho^+\|y\|^2 = (q + \delta)|s|^2 + (\rho^+ + \delta)\|y\|^2\] 
    \[= (q + \delta)\langle\psi, uu^*\psi\rangle + (\rho^+ + \delta)\langle\psi, \Pi_{u^\perp}\psi\rangle,\]
    as desired.

    Similarly, we also have
    \[\langle\psi, \cE\psi\rangle = |s|^2q + 2\Ree(\bar{s}\cdot u^*\cE\Pi_{u^\perp}y) + \langle y, \cE y\rangle\geq |s|^2q - 2|s|\cdot\|\Pi_{u^\perp}\cE u\|\cdot\|y\| - \rho^-\|y\|^2\] 
    \[\geq |s|^2q - \delta|s|^2 - \delta\|y\|^2 - \rho^-\|y\|^2\implies\cE\succeq(q - \delta)uu^* - (\rho^- + \delta)\Pi_{u^\perp}.\qedhere\]
\end{proof}
We now bound the vacuum functional $\varphi(\sfx^{2\ell})$ in terms of the parameters $q$, $\rho^+$, and $\delta$.
\begin{lemma}[Spectral Bounds for $\varphi(\sfx^{2\ell})$]\label{lem:varphi-spectral-dom}
If $(q - \rho^+)_+ < 1$ we have    
\[\varphi(\sfx^{2\ell})\leq\left(\frac{4(1 + (\rho^+ + \delta)\ell)}{1 - (q - \rho^+)_+}\right)^\ell.\]
\end{lemma}
\begin{proof}
    Note that 
    \[\varphi(\sfx^{2\ell}) = \langle f_{\bm{0}}, \sfx^{2\ell}f_{\bm{0}}\rangle = \|\sfx^\ell f_{\bm{0}}\|_2^2.\]
    Let $\Pi_\ell$ be the orthogonal projection $\cF\to\bigoplus_{0\leq \ell'\leq\ell}\cF_{\ell'}$. Note that by \cref{item:annihilation} of \cref{lem:twisted-commutations} we have that $\sfx^\ell f_{\bm{0}} = (\Pi_\ell \sfx\Pi_\ell)^\ell f_{\bm{0}}$. Consequently, 
    \[\varphi(\sfx^{2\ell}) = \|(\Pi_\ell \sfx\Pi_\ell)^\ell f_{\bm{0}}\|_2^2\leq \|\Pi_\ell \sfx\Pi_\ell\|^{2\ell}_{\operatorname{op}}.\]
    Now, recall that $\sfx = \sfa + \sfa^*$, and $\sfa^*(\cF_t)\seq\cF_{t + 1}$ for all $t\geq 0$ (by \cref{item:annihilation} of \cref{lem:twisted-commutations}), to obtain that
    \[\|\Pi_\ell \sfx\Pi_\ell\|_{\operatorname{op}}\leq 2\max_{0\leq t < \ell}b_t,\]
    where $b_t:= \|\sfa^*\|_{\operatorname{op}, \cF_t\to\cF_{t + 1}}$ for $t\geq 0$ (with $b_{-1}:= 0$). Thus it suffices to show
    \[b_t\leq\sqrt{\frac{1 + (\rho^+ + \delta) t}{1 - (q - \rho^+)_+}}\]
    for all $t\geq  0$.

    Towards that end, note that 
    \begin{align}\label{eq:aa-dom}
\sfa\sfa^*
&= |\cH|^{-1}\sum_{S,T\in\cH} \sfa_S\sfa_T^* \notag\overset{\text{\cref{item:twisted-bosonic-relation} of
\cref{lem:twisted-commutations}}}{=}
|\cH|^{-1}\sum_{S,T\in\cH}
\bigl(\1(S=T)\Id+\eps_{S,T}\sfa_T^*\sfa_S\bigr)\notag\\
&= \Id
+|\cH|^{-1}\sum_{S,T\in\cH}\eps_{S,T}\sfa_T^*\sfa_S\notag= \Id+\cL(\cE)\notag\\
&\overset{\text{\cref{prop:cE-dom,lem:cL-properties}}}{\preceq}
\Id+(q+\delta)\cL(uu^*)
+(\rho^++\delta)\bigl(\Num-\cL(uu^*)\bigr)\notag\\
&= \Id+(q-\rho^+)\cL(uu^*)
+(\rho^++\delta)\Num\notag = \Id+(q-\rho^+)\sfa^*\sfa
+(\rho^++\delta)\Num.
\end{align}
    Now, note that for any $t\geq  0$, we have $(\sfa^*\sfa)(\cF_t)\seq\cF_t$, $(\sfa\sfa^*)(\cF_t)\seq\cF_t$, and $\Num\vert_{\cF_t} = t\cdot\Id_{\cF_t}$. Consequently, 
    \[\sfa\sfa^*\preceq \Id+(q-\rho^+)\sfa^*\sfa
+(\rho^++\delta)\Num\implies\|\sfa\sfa^*\|_{\operatorname{op}, \cF_t}\leq 1 + (q - \rho^+)_+\cdot\|\sfa^*\sfa\|_{\operatorname{op}, \cF_t} + (\rho^+ + \delta)t,\]
    where $\|\cdot\|_{\operatorname{op}, \cF_t}$ refers to the operator norm restricted to $\cF_t$. However, we also have $\|\sfa\sfa^*\|_{\operatorname{op}, \cF_t} = b_t^2, \|\sfa^*\sfa\|_{\operatorname{op}, \cF_t} = b_{t - 1}^2$, i.e. for all $t\geq  0$ we have 
    \[b_t^2\leq (q - \rho^+)_+\cdot b_{t - 1}^2 + 1 + (\rho^+ + \delta)t\] 
    \[\implies b_t^2\leq\sum_{s = 0}^t(q - \rho^+)_+^s\cdot(1 + (\rho^+ + \delta)(t - s))\leq\frac{1 + (\rho^+ + \delta) t}{1 - (q - \rho^+)_+}\implies b_t\leq\sqrt{\frac{1 + (\rho^+ + \delta) t}{1 - (q - \rho^+)_+}},\]
    as desired.
\end{proof}

\subsection{Towards a lower bound for the operator norm of $H^{\cH}_{\operatorname{SYK}}$: Spectral edge of the relevant Krylov space}
In this subsection we will construct a test function which achieves energy $\approx 2/\sqrt\nu\approx \sqrt{2n}/k$ against $\sfx$. See \cref{lem:x-large-eig} for a formal statement.

Note that by \cref{prop:cE-dom} we know that $\cE\succeq(q - \delta)uu^* - (\rho^- + \delta)(\Pi_{u^\perp})$. Recall from the proof of \cref{lem:varphi-spectral-dom} that $\sfa\sfa^* = \Id + \cL(\cE), \cL(uu^*) = \sfa^*\sfa$, and thus we have 
\begin{align}\label{eq:aalowerbound}
    \sfa\sfa^*\succeq\Id + (q - \delta)\sfa^*\sfa - (\rho^- + \delta)(\Num - \sfa^*\sfa) = \Id + (q + \rho^-)\sfa^*\sfa - (\rho^- + \delta)\cdot\Num.
\end{align}
Now define,
\begin{align}\label{eq:radial-vec-def}
\psi_t&:=(\sfa^*)^tf_{\bm{0}},\\
v_t&:= \psi_t/\|\psi_t\|_2,\\
\beta_t&:=\|\psi_{t + 1}\|_2/\|\psi_t\|_2.
\end{align}
Note that the state $f_{t\bm e_i}$\footnote{$t\bm e_i\in\Z_{\geq 0}^\cH$ refers to the tuple with $t$ particles on the $i^{\mathrm{th}}$ hyperedge.} of $\psi_t$ will have coefficient $\sqrt{t!/|\cH|^{t}} \neq 0$, so $\psi_t\neq 0$ for all $t \ge 0$.  Just by definition we have $\sfa^*v_t = \beta_tv_{t + 1},\beta_0^2 = 1$. As we will show soon, $\beta_t$ act as effective ``hopping amplitudes'' on the relevant Krylov space spanned by the $v_t$. Also note that $\psi_t, v_t\in\cF_t$. The states $\psi_t$ represent collective excitations of the twisted bosons.

\begin{lemma}[Lower bound on hopping amplitude]\label{lem:betam}
    Assume $0 \le q < 1$. Then 
    \[\beta_t^2\geq\frac{1 - q^{t + 1} - (\rho^- + \delta)t}{1 - q}.\]
\end{lemma}
\begin{proof}
    Taking the expectation value of $\sfa\sfa^*$ on $v_t$, by \eqref{eq:aalowerbound} we have 
    \[\langle v_t, \sfa\sfa^*v_t\rangle\geq 1 + (q + \rho^-)\langle v_t, \sfa^*\sfa v_t\rangle - (\rho^- + \delta)\cdot\langle v_t, \Num v_t\rangle.\]
    Now, $\langle v_t, \sfa\sfa^*v_t\rangle = \|\sfa^*v_t\|^2 = \beta_t^2$, since $\|v_{t + 1}\|_2 = 1$. Also, $\langle v_t, \sfa^*\sfa v_t\rangle = \|\sfa v_t\|^2$. Since $v_t\in\cF_t$ we also have $\langle v_t, \Num v_t\rangle = t$. Consequently, 
    \[\beta_t^2\geq 1 + (q + \rho^-)\|\sfa v_t\|^2 - (\rho^- + \delta)t\geq  1 + q\| \sfa v_t\|^2 - (\rho^- + \delta)t.\]
    Now note that $\langle \sfa v_t, v_{t - 1}\rangle = \langle v_t, \sfa ^*v_{t - 1}\rangle = \beta_{t - 1}\langle v_t, v_t\rangle = \beta_{t - 1}$. Consequently, $\sfa v_t = \beta_{t - 1}v_{t - 1} + z$, where $\langle z, v_{t - 1}\rangle = 0$. In particular, $\|\sfa v_t\|^2 = \beta_{t - 1}^2 + \|z\|^2\geq\beta_{t - 1}^2$, and consequently, 
    \[\beta_t^2\geq 1 + q\beta_{t - 1}^2 - (\rho^- + \delta)t\implies\beta_t^2\geq\sum_{s = 0}^t q^s - (\rho^- + \delta)\sum_{r = 1}^t rq^{t - r}\] 
    \[\geq\frac{1 - q^{t + 1}}{1 - q} - (\rho^- + \delta)t\sum_{r = 0}^{t - 1}q^r\geq\frac{1 - q^{t + 1}}{1 - q} - \frac{(\rho^- + \delta)t}{1 - q},\]
    as desired.
\end{proof}

Define the parameters
\begin{equation}\label{eq:main-params}
    \nu := 1-q,\qquad B:= \sqrt{\nu/(\rho^- + \delta)},\qquad M:= \lfloor 1/\nu\rfloor.
\end{equation}

 \begin{definition}[Krylov Space]\label{defn:relevant-subspace}
 For parameters $t_0, M$ define the subspace (on the twisted Fock space)
\begin{equation}
    \cR_{t_0, M}:= \spn\{v_t: t_0\leq t < t_0 + M\}
\end{equation}
and $\Pi_{\cR_{t_0, M}}$ to be the orthogonal projector onto $\cR_{t_0, M}$. If $t_0, M$ are understood from context we will not mention it explicitly.
 \end{definition}
 The collective $\sfx$ operator projected onto $\cR$ resembles a one-dimensional ``hopping model'' with nearly-constant hopping. The spectral edge of this hopping model scales as $2/\sqrt{\nu}$.
\begin{lemma}[Spectral edge of hopping model]\label{lem:x-large-eig}
     Suppose $B\geq\kappa\geq\Omega_n(1)$ for some parameter $\kappa$. Write $t_0:= \lceil \kappa/\nu\rceil, L:= t_0 + M$. Then 
    \[\lambda_{\max}\left(\Pi_\cR \sfx \Pi_\cR\right)\geq\frac{2}{\sqrt{\nu}}\cdot\left(1 - O(1)\cdot\max\left\{\kappa^{-1}, \nu^2\right\}\right).\]
\end{lemma}
\begin{proof}
    Note that for $t_0\leq t < t_0 + M$ we have $q^{t + 1}\leq e^{-\nu(t + 1)}\leq e^{-\kappa}, (\rho^- + \delta)t\leq (\rho^- + \delta)(\kappa + 2)/\nu$. Since $\kappa\geq\Omega(1)$ we have $(\rho^- + \delta)t\leq O(1)\cdot(\rho^- + \delta)\kappa/\nu\leq O(1)\cdot(\nu/\kappa^2)(\kappa/\nu)\leq O(1/\kappa)$, and thus $q^{t + 1} + (\rho^- + \delta)t\leq e^{-\kappa} + O(1/\kappa)\leq O(1/\kappa)$.
    
    Now, by the definition of $\sfx$ we have that 
\[\Pi_\cR \sfx\Pi_\cR = 
\begin{pmatrix}
0 & \beta_{t_0} & 0 & \cdots & 0 \\
\beta_{t_0} & 0 & \beta_{t_0+1} & \ddots & \vdots \\
0 & \beta_{t_0+1} & 0 & \ddots & 0 \\
\vdots & \ddots & \ddots & \ddots & \beta_{t_0+M-2} \\
0 & \cdots & 0 & \beta_{t_0+M-2} & 0
\end{pmatrix},
\]
where this matrix is expressed in the orthonormal basis $\{v_t\}_{t_0\leq t < t_0 + M}$ defining $\cR$.

Now write 
\[\psi:= \sum_{t = 0}^{M - 1}\sin\left(\frac{(t + 1)\pi}{M + 1}\right)\cdot v_{t_0 + t}\in\cR.\]
Since every entry of $\psi$ is non-negative, by \cref{lem:betam} we have 
\[\lambda_{\max}\left(\Pi_\cR \sfx\Pi_\cR\right)\geq\frac{\langle\psi,\Pi_\cR \sfx\Pi_\cR\psi\rangle}{\langle\psi, \psi\rangle}\geq \frac{1 - O(1/\kappa)}{\sqrt{\nu}}\cdot 2\cos\left(\frac{\pi}{M + 1}\right)\] 
\[\geq\frac{2}{\sqrt{\nu}}\cdot\left(1 - O(1)\cdot\max\left\{\kappa^{-1}, \nu^2\right\}\right),\]
as desired.\footnote{Note that $2\cos(\pi/(M + 1))$ is the largest eigenvalue of the adjacency matrix of the path graph on $M$ vertices, which is what we get once we replace all the $\beta_\ast$s with their lower bound. Note that we can do this comparison without worry since $\psi$ has non-negative coefficients for all $v_t$.}
\end{proof}
\begin{remark}
    \emph{A priori}, increasing $L$ improves the bound on the spectral edge. However, for subsequent applications of Gaussian Hypercontractivity (\cref{fact:gaussian-hypercontractivity}) to obtain the spectral lower bound in \cref{subsec:lowerbound}, $L$ will play the role of degree of a polynomial, which we require to be $n^{O(1)}$~(see~\cref{{thm:syk-spectral-norm-lower}}).
\end{remark}

\subsection{A lower bound for $\|H^{\cH}_{\operatorname{SYK}}\|_{\op}$: Explicit Mapping from Twisted Boson Space to Fermionic Space} \label{subsec:lowerbound}

We now transfer a lower bound on $\lambda_{\max}(\Pi_\cR\sfx\Pi_\cR)$ to a lower bound on $\E_{\bm g}\|H^{\cH}_{\operatorname{SYK}}\|_{\op}$. See \cref{lem:lower-bound-syk-general} for a formal statement.

Recall \cref{def:hermite-polynomial}, and recall that for any $\bm r\in\Z_{\geq 0}^\cH, \bm g\in\R^\cH$ we define $h_{\bm r}(\bm g):= \prod_{S\in\cH}h_{\bm r_S}(\bm g_S)$. Also fix some ordering $\cH = (S_1, \ldots, S_m)$ (where $m = |\cH|$) and define $\Gamma^{\bm r}:= \Gamma_{S_1}^{\bm r_1}\cdots\Gamma_{S_m}^{\bm r_m}$.

Now define the linear map $U:\cF\to\cK(\R^\cH, \C^{N\times N})$ given as 
\begin{equation}\label{eq:isometry-def}
    \left(Uf_{\bm r}\right)(\bm g):= h_{\bm r}(\bm g)\cdot\Gamma^{\bm r}.
\end{equation}
We first show that this map is an isometric embedding of the twisted bosons into fermionic space. This isometry, for usual bosons, is an instance of the \emph{Wiener-It\^o chaos decomposition} \cite{Wiener38,Ito51}. Incorporating the twists into this decomposition is fairly straightforward though, as shown below:
\begin{proposition}[Isometric Embedding]\label{prop:isometry}
    For any $\psi, \psi'\in\cF$, $\langle\psi, \psi'\rangle_\cF = \langle U\psi, U\psi'\rangle_{\cK}$.
\end{proposition}
\begin{proof}
    Observe that 
\[\langle Uf_{\bm r}, Uf_{\bm s}\rangle_{\cK} = \E_{\bm g}\langle h_{\bm r}(\bm g)\Gamma^{\bm r}, h_{\bm s}(\bm g)\Gamma^{\bm s}\rangle\overset{\text{\cref{eq:tuple-hermite-ortho}}}{=}\1(\bm r = \bm s)\cdot\langle\Gamma^{\bm r}, \Gamma^{\bm s}\rangle = \1(\bm r = \bm s).\]
Here we use the fact that $\Gamma_S^2 = \Id$ for all non-empty $S\seq[n]$. Since $U$ is a linear map which maps an orthonormal basis of $\cF$ to an orthonormal set in $\cK$, we are done.
\end{proof}
Also, we define the Hamiltonian $H(\bm g)\in\cK(\R^\cH;\C^{N\times N})$ with coefficients given by $\bm g$:
\[H(\bm g):= \frac{1}{\sqrt{m}}\sum_{S\in\cH}\bm g_S\Gamma_S.\]
Finally define the left multiplication operator $\cM :\cK(\R^\cH;\C^{N\times N})\to\cK(\R^\cH;\C^{N\times N})$ given as
\[(\cM F)(\bm g):= H(\bm g)\cdot F(\bm g).\]
Note that $\cM $ is linear. We now show that the embedding $U$ as defined in \cref{eq:isometry-def} satisfies an intertwining property with respect to the Hamiltonian $H(\bm g)$. This intertwining property essentially follows from the fact that Hermite polynomials and creation-annihilation operators in bosonic spaces satisfy the same recurrence; Indeed, Hermite polynomials satisfy the recurrence $g_ih_{\bm r}(\bm g) = \sqrt{\bm r_i + 1}h_{\bm r + e_i}(\bm g) + \sqrt{\bm r_i}h_{\bm r - e_i}(\bm g)$, while bosonic creation-annihilation operators satisfy the recurrence $(\sfb_i + \sfb^*_i)f_{\bm r} = \sqrt{\bm r_i + 1}f_{\bm r + e_i} + \sqrt{\bm r_i}f_{\bm r - e_i}$. This correspondence is known as the Wiener-It\^o-Segal transform \cite{Wiener38,Ito51,Segal56}. Note that while the Wiener-It\^o-Segal transform is stated only for usual bosons, incorporating the twists in our twisted bosonic system into the transform is fairly straightforward, as shown below. For a modern treatment of both the Wiener-It\^o chaos decomposition stated earlier, and the Wiener-It\^o-Segal transform, we refer the reader to the textbook of Janson \cite{Janson97}. For a physics perspective, see \cite[Chapter~1]{Simons74}.
\begin{lemma}[Intertwining]\label{lem:intertwining}
    We have $\cM U = U\sfx$ as linear maps $\cF\to\cK(\R^\cH, \C^{N\times N})$.
\end{lemma}
\begin{proof}
    It suffices to show that $\cM Uf_{\bm r} = U\sfx f_{\bm r}$ for all $\bm r\in\Z_{\geq 0}^\cH$. Towards that end note that 
    \[(\cM Uf_{\bm r})(\bm g) = \frac{1}{\sqrt{m}}\left(\sum_{i = 1}^m\bm g_{S_i}\Gamma_{S_i}\right)\cdot h_{\bm r}(\bm g)\Gamma^{\bm r} = \frac{1}{\sqrt{m}}\sum_{i = 1}^m\bm g_{S_i}h_{\bm r}(\bm g)\Gamma_{S_i}\Gamma^{\bm r}.\]
    Using the commutation relations $\Gamma_S\Gamma_T = \eps_{S, T}\Gamma_T\Gamma_S$ repeatedly we obtain $\Gamma_{S_i}\Gamma^{\bm r} = \sigma_i(\bm r)\Gamma^{\bm r + e_i}$, where $\sigma_i(\bm r) \equiv \prod_{j<i} \eps_{ij}^{\bm r_j}$ is the twist phase. Furthermore, by \cref{def:hermite-polynomial}, we have $g_ih_{\bm r}(\bm g) = \sqrt{\bm r_i + 1}h_{\bm r + e_i}(\bm g) + \sqrt{\bm r_i}h_{\bm r - e_i}(\bm g)$. Thus 
    \[(\cM Uf_{\bm r})(\bm g) = \frac{1}{\sqrt{m}}\sum_{i = 1}^m\bm \left(\sqrt{\bm r_i + 1}h_{\bm r + e_i}(\bm g) + \sqrt{\bm r_i}h_{\bm r - e_i}(\bm g)\right)\cdot\sf\sigma_i(\bm r)\Gamma^{\bm r + e_i}.\]
    Now note that since $\Gamma_S^2 = \Id$, we have $\Gamma^{\bm r + e_i} = \Gamma^{\bm r - e_i}$ when $\bm r_i\geq 1$. Thus 
    \[(\cM Uf_{\bm r})(\bm g) = \frac{1}{\sqrt{m}}\sum_{i = 1}^m \sigma_i(\bm r) \cdot \left(\sqrt{\bm r_i + 1}\cdot(Uf_{\bm r + e_i})(\bm g) + \sqrt{\bm r_i}\cdot(Uf_{\bm r - e_i})(\bm g)\right)\]
    \[ = \frac{1}{\sqrt{m}}\sum_{i = 1}^m\left(U\left(\sqrt{\bm r_i + 1}\sigma_i(\bm r)\cdot f_{\bm r + e_i})\right)(\bm g) + U\left(\sqrt{\bm r_i}\sigma_i(\bm r)\cdot f_{\bm r - e_i})\right)(\bm g)\right)\]
    \[\implies \cM Uf_{\bm r} = U\cdot \frac{1}{\sqrt{m}}\sum_{i = 1}^m\left(\sqrt{\bm r_i + 1}\sigma_i(\bm r)\cdot f_{\bm r + e_i} + \sqrt{\bm r_i}\sigma_i(\bm r)\cdot f_{\bm r - e_i}\right).\]
    Thus it suffices to show that $\sqrt{\bm r_i + 1}\sigma_i(\bm r)\cdot f_{\bm r + e_i} + \sqrt{\bm r_i}\sigma_i(\bm r)\cdot f_{\bm r - e_i} = (\sfa_{S_i} + \sfa^*_{S_i}) f_{\bm r}$. But this follows from the definitions.
\end{proof}

We now use the multiplication operator to yield a lower bound for $\|H(\bm g)\|_{\operatorname{op}}$ multiplied by a ``tilt'' $w(\bm g)$. Subsequently in an upcoming Lemma we use Gaussian hypercontractivity to decouple the tilt $w(\bm g)$ from $\|H(\bm g)\|_{\operatorname{op}}$ to yield a lower bound just for $\E_{\bm g\sim\cN(0, \Id_m)}\|H(\bm g)\|_{\operatorname{op}}$. The technique of decoupling a product using hypercontractivity is a standard method in the analysis of boolean functions, see \cite{ODonnell21} for further perspectives.
\begin{lemma}[Operator Norm Witness] \label{lem:transfer-virtual-to-real}
    Let $\cR$ be any space of the kind defined in \cref{defn:relevant-subspace}. Then there exists $F\in\cK(\R^\cH, \C^{N\times N})$ of degree $\leq L$ such that 
    \[\E_{\bm g\sim\cN(0, \Id_m)}\left[\|H(\bm g)\|_{\operatorname{op}}\cdot w(\bm g)\right]\geq \lambda_{\max}(\Pi_\cR \sfx\Pi_\cR),\]
    where:
    \begin{enumerate}[(1)]
        \item $w(\bm g) := \|F(\bm g)\|_{\operatorname{HS}}^2$,
        \item More explicitly, we can choose $F = U\psi$, where $\psi\in\cR$ is a unit vector for which $\langle\psi, \sfx\psi\rangle_{\cF} = \lambda_{\max}(\Pi_\cR \sfx\Pi_\cR)$, 
        \item $\E_{\bm g}w(\bm g) = 1$.
    \end{enumerate}
\end{lemma}
\begin{proof}
    By \cref{prop:isometry} we have $\langle\psi, \sfx\psi\rangle_{\cF} = \langle U\psi, U\sfx\psi\rangle_{\cK}$. By \cref{lem:intertwining} we have $\langle U\psi, U\sfx\psi\rangle_{\cK} = \langle U\psi, \cM U\psi\rangle_{\cK}$. Write $F:= U\psi$. Then 
    \[\langle U\psi, \cM U\psi\rangle_{\cK} = \langle F, \cM F\rangle_{\cK} = \E_{\bm g\sim\cN(0, \Id_m)}\tr(F(g)^*(\cM F)(g)) = \E_{\bm g\sim\cN(0, \Id_m)}\tr(F(\bm g)^*H(\bm g)F(\bm g)).\]
    Since $H(\bm g)\preceq\|H(\bm g)\|_{\operatorname{op}}\cdot\Id$, we have $F(\bm g)^*H(\bm g)F(\bm g)\preceq\|H(\bm g)\|_{\operatorname{op}}\cdot F(\bm g)^*F(\bm g)$, and thus $\tr(F(\bm g)^*H(\bm g)F(\bm g))\leq \|H(\bm g)\|_{\operatorname{op}}\cdot \tr(F(\bm g)^*F(\bm g)) = \|H(\bm g)\|_{\operatorname{op}}\cdot w(\bm g)$. Finally, note that since $\psi\in\bigoplus_{t\leq L}\cF_t$, $F\in\cK$ has degree $\leq L$ by the definition of $U$. Furthermore, $\E_{\bm g}w(\bm g) = \langle F, F\rangle_\cK = \langle\psi,\psi\rangle_\cF = 1$; thus we are done.
\end{proof}

We now show that the Hamiltonian $H(\bm g)$ is Lipschitz in $\bm g$. This has already been shown in~\cite{anschuetzStronglyInteractingFermions2025, feng2020spectrum}, where further connections to concepts such as the commutator index and the Lov\'asz theta function of the commutator graph have been drawn. For our purposes the following simple spectral argument will suffice.
\begin{lemma}[Lipschitzness]\label{lem:lipschitz-operator}
    For every $\bm g, \bm g'\in\R^\cH$ we have 
    \[\|H(\bm g) - H(\bm g')\|_{\operatorname{op}}\leq\sqrt{\frac{\rho^+ + \delta}{1 - q}}\cdot\|\bm g - \bm g'\|_2,\]
    where $q, \rho^+, \delta$ are as in \cref{eq:spectral-notation-def}, assuming $q < 1$.
\end{lemma}
\begin{proof}
    Consider an arbitrary $\bm g\in\R^\cH, \bm g\neq 0$. Since $H(\bm g) - H(\bm g') = H(\bm g - \bm g')$, it suffices to show that $\|H(\bm g)\|_{\operatorname{op}}\leq\sqrt{\frac{\rho^+ + \delta}{1 - q}}\cdot\|\bm g\|_2$. Let $\psi\in\C^{N}$ be an arbitrary unit vector, and define the matrix $X_\psi:= X\in\R^{\cH\times\cH}$ defined as $X_{S, T}:= \bm g_S\bm g_T\Ree\langle\Gamma_S\psi, \Gamma_T\psi\rangle$. Note that for any $v\in\R^\cH$ we have $\langle v, Xv\rangle = \|\sum_S\bm g_Sv_S\Gamma_S\psi\|^2\geq 0$ and thus $X\succeq 0$. Furthermore, $\Tr(X) = \sum_S\bm g_S^2\|\Gamma_S\psi\|^2 = \|\bm g\|^2$, where $\|\Gamma_S\psi\|^2 = \|\psi\|^2 = 1$, since $\Gamma_S$ is unitary.\footnote{Recall that we fixed an unitary representation for $\{\gamma_i\}_{1\leq i\leq n}$ into $\C^{N\times N}$.} Finally, note that 
    \[\|H(\bm g)\psi\|^2 = \frac{1}{m}\left\langle\sum_{S\in\cH}\bm g_S\Gamma_S\psi, \sum_{S\in\cH}\bm g_S\Gamma_S\psi\right\rangle = \frac{1}{m}\sum_{S, T\in\cH}\bm g_S\bm g_T\langle\Gamma_S\psi, \Gamma_T\psi\rangle\] 
    \[\implies\|H(\bm g)\psi\|^2 = \Ree\|H(\bm g)\psi\|^2 = \frac{1}{m}\sum_{S, T\in\cH}\bm g_S\bm g_T\Ree\langle\Gamma_S\psi, \Gamma_T\psi\rangle = \frac{1}{m}\sum_{S, T\in\cH}X_{S, T} = \langle u, Xu\rangle,\]
    where recall that $u = u^\cH:= \1/\sqrt{m}\in\R^\cH$. 

    Now, note that if $\eps_{S, T} = -1$, i.e. $\Gamma_S\Gamma_T = -\Gamma_T\Gamma_S$, then $2\Ree\langle\Gamma_S\psi, \Gamma_T\psi\rangle = \langle\Gamma_S\psi, \Gamma_T\psi\rangle + \langle\Gamma_T\psi, \Gamma_S\psi\rangle = \langle\psi, \Gamma_S\Gamma_T\psi\rangle + \langle\psi, \Gamma_T\Gamma_S\psi\rangle = 0$, i.e. $\eps_{S, T} = -1\implies X_{S, T} = 0$. Consequently, 
    \[\Tr(\cE X) = \frac{1}{m}\sum_{S, T}\eps_{S, T}X_{S, T} = \frac{1}{m}\sum_{S, T}X_{S, T} = \|H(\bm g)\psi\|^2.\]
    On the other hand, by \cref{prop:cE-dom} we have $\cE\preceq (q + \delta)uu^* + (\rho^+ + \delta)(\Pi_{u^\perp})$. Since $X\succeq 0$, we have\footnote{Note that if $X\succeq 0$ and $A\succeq B$, then $\Tr(BX)\leq\Tr(AX)$, since the product of two PSD matrices has non-negative trace.} 
    \[\Tr(\cE X)\leq\Tr\left(\left(q + \delta)uu^* + (\rho^+ + \delta)(\Pi_{u^\perp})\right)X\right).\]
    But $\Tr(uu^*X) = \langle u, Xu\rangle = \frac{1}{m}\sum_{S, T\in\cH}X_{S, T} = \Tr(\cE X)$. Consequently, 
    \[\Tr(\cE X)\leq(q + \delta)\Tr(\cE X) + (\rho^+ + \delta)(\Tr(X) - \Tr(\cE X)) = (q - \rho^+)\Tr(\cE X) + (\rho^+ + \delta)\Tr(X)\]
    \[\implies\|H(\bm g)\psi\|^2 = \Tr(\cE X)\leq \frac{\rho^+ + \delta}{1 - q + \rho^+}\cdot\|\bm g\|^2\implies \|H(\bm g)\psi\|\leq\sqrt{\frac{\rho^+ + \delta}{1 - q}}\cdot\|\bm g\|,\]
    as desired.
\end{proof}

We now combine these results to obtain a lower bound on the operator norm of the SYK Hamiltonian in terms of the spectral edge of the twisted boson operator $\sfx = \sfa + \sfa^*$ on the relevant Krylov subspace $\cR$.
\begin{lemma}[Lower Bounds on Operator Norm of the SYK Hamiltonian]\label{lem:lower-bound-syk-general}
    Let $n\geq k\geq 2$ be even integers, and let $\alpha:= k^2/n < 1/16$. Consider the SYK Hamiltonian $H^\cH_{\operatorname{SYK}}$ (from \cref{eq:syk-def}) for the hypergraph $\cH$. Let $\cR$ be any space of the kind defined in \cref{defn:relevant-subspace}. Then
    \[\E\|H^\cH_{\operatorname{SYK}}\|_{\operatorname{op}}\geq \lambda_{\max}(\Pi_\cR\sfx\Pi_\cR) - O(1)\cdot \sqrt{\frac{L(\rho^+ + \delta)}{1 - q}}.\]
\end{lemma}
\begin{proof}
    Note that $H^\cH_{\operatorname{SYK}} = H(\bm g)$. Apply \cref{lem:transfer-virtual-to-real} to obtain that there exists $F\in\cK(\R^m;\C^{N\times N})$ such that 
    \[\E_{\bm g\sim\cN(0, \Id_m)}\left[\|H(\bm g)\|_{\operatorname{op}}\cdot w(\bm g)\right]\geq\lambda_{\max}(\Pi_\cR\sfx\Pi_\cR).\]
    where $w(\bm g) := \|F(\bm g)\|^2_{\operatorname{HS}}$. Since $\E_{\bm g}w(\bm g) = 1$, we have 
    \[\E_{\bm g}\left[\|H(\bm g)\|_{\operatorname{op}}\cdot w(\bm g)\right] = \E_{\bm g}\|H(\bm g)\|_{\operatorname{op}} + \E_{\bm g}\left[\left(\|H(\bm g)\|_{\operatorname{op}} - \E_{\bm g}\|H(\bm g)\|_{\operatorname{op}}\right)\cdot w(\bm g)\right]\]
    \[\overset{\text{H\"{o}lder}}{\leq}\E_{\bm g}\|H(\bm g)\|_{\operatorname{op}} + \left\|\|H(\bm g)\|_{\operatorname{op}} - \E_{\bm g}\|H(\bm g)\|_{\operatorname{op}}\right\|_p\cdot \|w(\bm g)\|_{p'},\]
    where we set $p:= L + 1$, where $L$ is as in \cref{lem:x-large-eig}. Now, note that
    \[\|w(\bm g)\|_{p'} = \|F(\bm g)\|_{\cK, 2p'}^2\overset{\text{\cref{fact:gaussian-hypercontractivity}}}{\leq}(2p' - 1)^L\|F(\bm g)\|_{\cK, 2}^2 = (1 + 2/L)^L\cdot\E_{\bm g}w(\bm g)\leq O(1).\]
    By \cref{fact:gaussian-lipschitz-conc,lem:lipschitz-operator} we also have 
    \[ \left\|\|H(\bm g)\|_{\operatorname{op}} - \E_{\bm g}\|H(\bm g)\|_{\operatorname{op}}\right\|_p\leq O(1)\cdot \sqrt{\frac{\rho^+ + \delta}{1 - q}}\cdot\sqrt{L},\]
    as desired. 
\end{proof}

\section{Operator Norm of the SYK Model}
Recall the sign matrix $\cE\in\C^{\cH\times\cH}$ defined as $\cE_{S, T}:= |\cH|^{-1}\cdot\eps_{S, T} = |\cH|^{-1}\cdot(-1)^{|S\cap T|}$. Also recall the other notation from \cref{eq:spectral-notation-def}:
\begin{align*}
u&:= \1/\sqrt{|\cH|},\\
q&:=\langle u,\cE u\rangle,\\
\delta&:=\|(\Pi_{u^\perp})\cE u\|,\\
\rho^+&:=\lambda_{\max}\!\left((\Pi_{u^\perp})\cE(\Pi_{u^\perp})\right),\\
\rho^-&:=\max\left\{0, -\lambda_{\min}\!\left((\Pi_{u^\perp})\cE(\Pi_{u^\perp})\right)\right\}.
\end{align*}

We first bound these quantities in the ``dense'' case, i.e. the case when $\cH = \binom{[n]}{k}$:
\subsection[The Dense Case]{The Dense Case: $\cH = \binom{[n]}{k}$}
Throughout this section assume $\cH = \binom{[n]}{k}$. 

Recall the setting of \cref{fact:johnson-eigenspaces}, i.e. we have matrices $M_t\in\C^{\binom{[n]}{k}\times\binom{[n]}{k}}$ given as $M_t(S, T):= \binom{|S\cap T|}{t}$, and we have a detailed description of their spectrum. Now note that 
\[(-1)^{|S\cap T|} = (1 - 2)^{|S\cap T|} = \sum_{t = 0}^{|S\cap T|}(-2)^t\binom{|S\cap T|}{t} = \sum_{t = 0}^{k}(-2)^t\binom{|S\cap T|}{t},\]
and consequently we have 
\[\cE = \frac{1}{\binom{n}{k}}\sum_{t = 0}^k(-2)^tM_t,\]
i.e. the sign matrix $\cE$ is a matrix in the Johnson scheme. By \cref{fact:johnson-eigenspaces} we obtain that the eigenvalues of $\cE$ are given by
\begin{equation}\label{eq:complete-sign-eig}
    \lambda_j:= \binom{n}{k}^{-1}\cdot\sum_{t = j}^k(-2)^t\binom{k - j}{t - j}\cdot\binom{n - t - j}{k - t}
\end{equation}
for $0\leq j\leq k$, with $\lambda_j$ occurring $d_j:= \binom{n}{j} - \binom{n}{j - 1}$ times in $\operatorname{spec}(\cE)$. Furthermore, $u = \1/\sqrt{\binom{n}{k}}$ is the eigenvector corresponding to $\lambda_0$. Consequently, in the context of \cref{eq:spectral-notation-def} we obtain that 
\begin{equation}\label{eq:complete-param-values}
    q = \lambda_0, \delta = 0, \rho^+ = \max_{j\geq 1}\lambda_j.
\end{equation}
To upper bound these quantities, we need the following auxiliary technical lemma:
\begin{lemma}[Sign Matrix Eigenvalues]\label{lem:complete-matrix-values}
    Suppose $n\geq k\geq 2$. Write $\alpha:= k^2/n$, and assume $\alpha < 1/16$. For $\lambda_j$ as in \cref{eq:complete-sign-eig}, we have
    \begin{equation}\label{eq:lambdaj-estimate}
        \left|\frac{\lambda_j}{(-2)^j\binom{n - 2j}{k - j}\binom{n}{k}^{-1}} - 1\right|\leq 8\alpha.
    \end{equation}
    In particular, $\lambda_j > 0$ if and only if $j$ is even. Moreover, 
    \begin{enumerate}[(1)]
        \item\label{item:lambda1-bound} $(1 - k/n)(1 - 8\alpha)\leq |\lambda_1|/(2k/n)\leq 1 + 8\alpha$.
        \item\label{item:lambdaj-bound} $|\lambda_j|\leq 2(8k/n)^j$ for all $j\geq 0$. Consequently, $\sum_{j = 0}^{k}d_j|\lambda_j|\leq e^{9k}$, where $d_j:= \binom{n}{j} - \binom{n}{j - 1}$.
        \item\label{item:even-lambdaj-upperbound} $\max_{\substack{2\leq j\leq k\\j\text{ even}}}\lambda_j = \Theta(\alpha/n)$.
        \item\label{item:odd-lambdaj-upperbound} $\max_{\substack{1\leq j\leq k\\j\text{ odd}}}|\lambda_j| = |\lambda_1| \leq (2 + O(\alpha))\cdot\sqrt{\alpha/n}$.
        \item $2\alpha - 2\alpha^2\leq 1 - \lambda_0\leq 2\alpha$. 
    \end{enumerate}
\end{lemma}
\begin{proof}
    We know that 
    \[\lambda_j = \binom{n}{k}^{-1}\cdot\sum_{t = j}^k(-2)^t\binom{k - j}{t - j}\cdot\binom{n - t - j}{k - t} = \binom{n}{k}^{-1}(-2)^j\cdot\sum_{s = 0}^{k - j}\underbrace{(-2)^s\binom{k - j}{s}\cdot\binom{n - 2j - s}{k - j - s}}_{:= f(s)}.\]
    For any $s \geq 0$, we have
    \[\left|\frac{f(s)}{f(0)}\right| = \left|\frac{(-2)^s\binom{k - j}{s}\cdot\binom{n - 2j - s}{k - j - s}}{\binom{n - 2j}{k - j}}\right|\leq \frac{1}{s!}\left(\frac{2k^2}{n - 2k}\right)^s = \frac{1}{s!}\left(\frac{2\alpha}{1 - 2k/n}\right)^s\leq\frac{(4\alpha)^s}{s!}.\]
    Consequently, 
    \[\left|\frac{\lambda_j}{(-2)^j\binom{n - 2j}{k - j}\binom{n}{k}^{-1}} - 1\right|\leq\sum_{s = 1}^{k - j}\frac{(4\alpha)^s}{s!}\leq e^{4\alpha} - 1\leq 8\alpha,\]
    as desired. \cref{item:lambda1-bound} then follows by direct computation.

    For \cref{item:lambdaj-bound}, we have that 
    \[\binom{n - 2j}{k - j}\binom{n}{k}^{-1} = \frac{(k)_j(n - k)_j}{(n)_{2j}}\leq\frac{k^jn^j}{(n - 2k)^{2j}}\leq\left(\frac{4k}{n}\right)^j\]
    where the last inequality holds since $n > 16k^2\geq 4k$. Consequently, 
    \[|\lambda_j|\leq (1 + 8\alpha)\cdot 2^j\binom{n - 2j}{k - j}\binom{n}{k}^{-1}\leq 2\cdot \left(\frac{8k}{n}\right)^j,\]
    as desired. Consequently, since $d_j\leq\binom{n}{j}\leq n^j/j!$, we have 
    \[\sum_{j = 0}^{k}d_j|\lambda_j|\leq 2\sum_{j = 0}^{k}\frac{n^j}{j!}\cdot\left(\frac{8k}{n}\right)^j\leq 2\sum_{j = 0}^{\infty}\frac{(8k)^j}{j!} = 2e^{8k}\leq e^{9k}\]
    since $k\geq 2$. \cref{item:even-lambdaj-upperbound} is a direct consequence of \cref{eq:lambdaj-estimate,item:lambdaj-bound}, and \cref{item:odd-lambdaj-upperbound} follows from \cref{item:lambda1-bound} and the fact that $\alpha < 1/16$.

    Finally, let us consider $\lambda_0$. Note that 
    \[\lambda_0 = \binom{n}{k}^{-1}\cdot\sum_{s = 0}^{k}(-2)^s\binom{k}{s}\cdot\binom{n - s}{k - s} \overset{(\ast)}{=} \sum_{s = 0}^{k}(-2)^s\E\binom{X}{s} = \E\sum_{s = 0}^{k}(-2)^s\binom{X}{s} = \E(-1)^X,\]
    where $X = |S\cap T|$ is the random variable obtained by considering the intersection size of the sets $S, T$ independently and uniformly chosen from $\binom{[n]}{k}$. To see why $(\ast)$ is true, note that $\binom{X}{s} = \sum_{U\in\binom{S}{s}}\1(U\seq T)$, and thus $\E\left[\binom{X}{s}\mid S\right] = \sum_{U\in\binom{S}{s}}\Pr_T(U\seq T) = \binom{k}{s}\cdot\frac{\binom{n - s}{k - s}}{\binom{n}{k}}$. Since $\E\left[\binom{X}{s}\mid S\right]$ does not depend on $S$, we have $\E\binom{X}{s} = \binom{k}{s}\cdot\frac{\binom{n - s}{k - s}}{\binom{n}{k}}$, as desired.

    Consequently, $1 - \lambda_0 = \E_{S, T}(1 - (-1)^X) = 2\cdot\E_{S, T}\1(|S\cap T|\text{ is odd})$. Finally, note that for any integer $t\geq 0$, we have $t - t(t - 1)\leq \1(t\text{ is odd})\leq t$. Note that $\E_{S, T}|S\cap T| = \alpha, \E_{S, T}\binom{|S\cap T|}{2} = \binom{k}{2}^2/\binom{n}{2}\leq\alpha^2/2$ for $\alpha\leq 1$. Thus 
    \[2(\alpha - \alpha^2)\leq 1 - \lambda_0\leq 2\alpha,\]
    as desired.
\end{proof}

Putting \cref{eq:complete-param-values,lem:complete-matrix-values} yields an upper bound on the spectral norm of the SYK Hamiltonian for the dense case $\cH = \binom{[n]}{k}$:
\begin{theorem}[Operator Norm of the Dense SYK: Upper Bound]\label{thm:syk-spectral-norm}
    Let $n\geq k\geq 2$ be even integers, and let $\alpha:= k^2/n < 1/16$. Consider the SYK Hamiltonian (from \cref{eq:syk-def}) for the dense case $\cH = \binom{[n]}{k}$, i.e.
    \[H:= \frac{1}{\sqrt{\binom{n}{k}}}\sum_{S\in\binom{[n]}{k}}g_S\Gamma_S\]
    where $\{g_S\}_{S\in\binom{[n]}{k}}$ are i.i.d $\cN(0, 1)$ Gaussians. Then 
    \[\E\|H\|_{\operatorname{op}}\leq \sqrt{\frac{2}{\alpha}} + O(1).\]
    In particular, if $\alpha\leq o_n(1)$ then $\E\|H\|_{\operatorname{op}}\leq (1 + o(1))\cdot\sqrt{2/\alpha}$.
\end{theorem}
\begin{proof}
    Let $\ell\geq 1$ be an integer. Then 
    \[\E\left[\|H\|_{\operatorname{op}}\right]^{2\ell}\overset{\text{Jensen}}{\leq}\E\left[\|H\|_{\operatorname{op}}^{2\ell}\right]\leq\E\left[\Tr(H^{2\ell})\right] = N\cdot\E\left[\tr(H^{2\ell})\right]\overset{\text{\cref{prop:varphi-trace-re}}}{=} N\cdot\varphi(\sfx^{2\ell})\] 
    \[\overset{\text{\cref{lem:varphi-spectral-dom}}}{\leq} N\cdot\left(\frac{4(1 + (\rho^+ + \delta)\ell)}{1 - (q - \rho^+)_+}\right)^\ell\]
    where $N = 2^{n/2}$ is the dimension of $H$, and $\varphi, x, q, \rho^+, \delta$ are defined as in \cref{sec:twisted-fock-space}. Consequently, 
    \[\E\|H\|_{\operatorname{op}}\leq 2\cdot 2^{n/(4\ell)}\cdot\sqrt{\frac{1 + (\rho^+ + \delta)\ell}{1 - (q - \rho^+)_+}}\]
    for all integers $\ell\geq 1$. 
    
    Now, by \cref{eq:complete-param-values} we know that $q = \lambda_0, \delta = 0, \rho^+ = \max_{j\geq 1}\lambda_j$. Note that $q - \rho^+\leq q\leq 1 - 2\alpha + 2\alpha^2 < 1$, and thus by \cref{lem:complete-matrix-values} we obtain that $1 - q\geq 2(\alpha - \alpha^2), \rho^+\leq O(\alpha/n)$. Consequently, $1 - (q - \rho^+)_+ = 1 - q + \rho^+\geq 1 - q$ for large enough $n$,\footnote{Note that we claim the result for all $n > 16k^2$. For $n$ smaller than what we need for our inequalities to hold, we can simply inflate the implicit constant in the $O(1)$ in the theorem statement to make it work.} and thus
    \[\E\|H\|_{\operatorname{op}}\leq 2^{n/(4\ell)}\cdot\sqrt{\frac{2 + O(\alpha\ell/n)}{\alpha - \alpha^2}}.\]
    Choose $\ell := \lceil n/\sqrt{\alpha}\rceil$ to obtain that 
    \[\E\|H\|_{\operatorname{op}}\leq \sqrt{\frac{2}{\alpha}}\cdot \left(1 + O(\sqrt{\alpha})\right) \leq \sqrt{\frac{2}{\alpha}} + O(1),\]
    as desired.
\end{proof}

\begin{theorem}[Operator Norm of the Dense SYK: Lower Bound]\label{thm:syk-spectral-norm-lower}
    Let $n\geq k\geq 2$ be even integers, and let $\alpha:= k^2/n < 1/16$. Consider the SYK Hamiltonian (from \cref{eq:syk-def}) for the dense case $\cH = \binom{[n]}{k}$. Then 
    \[\E\|H\|_{\operatorname{op}}\geq \sqrt{\frac{2}{\alpha}}\cdot\left(1 - O(1)\cdot\max\left\{k^{-1/2}, \alpha^2\right\}\right) = \frac{\sqrt{2n}}{k}\cdot\left(1 - O(1)\cdot\max\left\{k^{-1/2}, \frac{k^4}{n^2}\right\}\right).\]
\end{theorem}
\begin{proof}
    Throughout the proof we'll borrow notation from \cref{lem:x-large-eig}. By \cref{lem:lower-bound-syk-general} we know that if $B\geq\Omega(1)$, then
    \[\E\|H\|_{\operatorname{op}}\geq \lambda_{\max}(\Pi_\cR\sfx\Pi_\cR) - O(1)\cdot \sqrt{\frac{L(\rho^+ + \delta)}{1 - q}}.\]
    Now, we'll use \cref{lem:complete-matrix-values} to compute the parameters in \cref{lem:x-large-eig}. Note that $\rho^- = |\lambda_1|\leq O(\sqrt{\alpha/n})$, and thus $B = \sqrt{\nu/\rho^-} \geq\Omega\left(\sqrt{\alpha/\sqrt{\alpha/n}}\right)\geq\Omega(\sqrt{k})\geq\Omega(1)$, and thus our invocation of \cref{lem:x-large-eig} is justified. Choose $\kappa = \Theta(\sqrt{k})\leq B$. Then $L\leq O(\kappa/\nu)\leq O(n/k^{3/2}), 1 - q\geq\Omega(k^2/n), \rho^+ + \delta\leq O(k^2/n^2)$, and thus 
    \[\frac{L(\rho^+ + \delta)}{1 - q}\leq O(1)\cdot\frac{n\cdot k^2\cdot n}{k^{3/2}\cdot n^2\cdot k^2}\leq O(k^{-3/2}).\]
    Also 
    \[\lambda_{\max}(\Pi_\cR\sfx\Pi_\cR)\geq \sqrt{\frac{2}{\alpha}}\cdot\left(1 - O(1)\cdot\max\left\{k^{-1/2}, \alpha^2\right\}\right),\]
    and we're done.
\end{proof}

Finally, using \cref{fact:gaussian-lipschitz-conc}, we can even obtain strong subgaussian concentration for $\|H_{\syk}\|_{\op}$:

\begin{theorem}[Subgaussian Tails for the Operator Norm]\label{thm:syk-subgaussian}
    Let $n\geq k\geq 2$ be even integers, and let $\alpha:= k^2/n < 1/16$. Consider the SYK Hamiltonian (from \cref{eq:syk-def}) for the dense case $\cH = \binom{[n]}{k}$. Then for any $t\geq 0$ we have
    \[\Pr\left(|\|H\|_{\operatorname{op}} - \E_{\bm g}\|H\|_{\operatorname{op}}|\geq t\right)\leq 2\cdot\exp\left(-\Omega(nt^2)\right).\]
\end{theorem}
\begin{proof}
    Towards applying \cref{fact:gaussian-lipschitz-conc}, consider the function $\bm g\mapsto \|H(\bm g)\|_{\op} = \|H\|_{\op}$. By \cref{lem:lipschitz-operator} we know that the Lipschitz constant of this map is $\leq\sqrt{(\rho^+ + \delta)/(1 - q)}$. From the proof of \cref{thm:syk-spectral-norm-lower} we know that $(\rho^+ + \delta)/(1 - q)\leq O(1/n)$. We are now done by \cref{fact:gaussian-lipschitz-conc}.
\end{proof}

\subsection[The Sparse Case]{The Sparse Case: $\cH$ is a random hypergraph}
Suppose $\cH$ is a random $k$-uniform hypergraph with $m$ hyperedges, i.e. every hyperedge of $\cH$ sampled independently, uniformly at random, from $\binom{[n]}{k}$. Note that $\cH$ is a multiset in general.

Write $\cE^{\comp}:= \cE^{\binom{[n]}{k}} = \binom{n}{k}^{-1}\cdot(\eps_{S, T})_{S, T\in\binom{[n]}{k}}$. Then note that 
\[\cE^\cH = \frac{\binom{n}{k}}{m}P^*\cE^{\comp}P,\]
where $P:\C^\cH\to\C^{\binom{[n]}{k}}$ is the usual map where for every $S\in\cH$ we have $Pe_{S}:= e_{S}\in\C^\Omega$, where $\{e_S\}_{S\in\cH}$ (resp. $\{e_S\}_{S\in\binom{[n]}{k}}$) is the standard basis for $\C^{\cH}$ (resp. $\C^{\binom{[n]}{k}}$).

Furthermore, since $\cE^{\comp}$ is Hermitian, we can write 
\[\cE^{\comp} = \frac{\lambda_0\1\1^*}{\binom{n}{k}} + \sum_{j\geq 1}\lambda_j\Pi_j\]
where $\lambda_j$ are as in \cref{eq:complete-sign-eig}, and $\Pi_j$ refers to the orthogonal projection $\C^{\binom{[n]}{k}}\to V_j$, where $V_j$ is as in \cref{fact:johnson-eigenspaces}. Write 
\[G^+:= \sum_{j\geq 1:\lambda_j\geq 0}\lambda_j\Pi_j,\;\;G^{-}:= -\sum_{j\geq 1:\lambda_j < 0}\lambda_j\Pi_j,\]
and thus $\cE^{\comp} = \frac{\lambda_0\1\1^*}{\binom{n}{k}} + G^+ - G^-$. 

Note that $G^+, G^-\succeq 0$. Furthermore by the $\mathfrak{S}_n$-invariance of $V_j$ (see \cref{fact:johnson-eigenspaces}), we have that for any $S, T\in\binom{[n]}{k}$ we have $\langle e_S, G^+ e_S\rangle = \langle e_T, G^+ e_T\rangle$ for all $S, T\in\binom{[n]}{k}$. Since $\sum_{S\in\binom{[n]}{k}}\langle e_S, G^+ e_S\rangle = \Tr(G^+)$ we obtain that $\langle e_S, G^+ e_S\rangle = \tr(G^+)$ for all $S\in\binom{[n]}{k}$. Similarly, we also have $\langle e_S, G^- e_S\rangle = \tr(G^-)$ for all $S\in\binom{[n]}{k}$.

Now write $A^{\pm}:= \frac{\binom{n}{k}}{m}P^*G^{\pm}P$. Note that $A^\pm\succeq 0$.

\begin{theorem}[Operator Norm of the Sparse SYK: Upper Bound]\label{thm:syk-spectral-norm-random}
    Let $n\geq k\geq 2$ be even integers, and suppose $\alpha:= k^2/n < 1/16$.  Let $\cH$ be a random $k$-uniform hypergraph with $m\geq 2^{Ck}n\log n$ edges, where $C > 0$ is some sufficiently large absolute constant. Let 
    \[H^\cH_{\operatorname{SYK}}:= \frac{1}{\sqrt{|\cH|}}\sum_{S\in\cH}g_S\Gamma_S\]
    be the SYK Hamiltonian corresponding to $\cH$ (as in \cref{eq:syk-def}). Then with probability $\geq 1 - O(1/\log(n))$ over the randomness of $\cH$ we have 
    \[\E_{\bm g}\|H^\cH_{\operatorname{SYK}}\|_{\operatorname{op}}\leq \sqrt{\frac{2}{\alpha}}\cdot\left(1 + O\left(\sqrt{\alpha + e^{-\Omega(k)}}\right)\right) = \frac{\sqrt{2n}}{k}\cdot\left(1 + O\left(\sqrt{\frac{k^2}{n} + e^{-\Omega(k)}}\right)\right).\]
    In particular if $k\geq\omega(\log(n))$ then $\E\|H^\cH_{\operatorname{SYK}}\|_{\operatorname{op}}\leq \sqrt{2n}/k + O(1)$.
\end{theorem}
\begin{proof}
Define $q, \delta, \rho^+$ as in \cref{eq:spectral-notation-def}. Note that 
    \[\cE^\cH = \lambda_0uu^* + A^+ - A^-,\]
    where $u:= \1/\sqrt{m}$. Let $\Pi = \Pi_{u^\perp}$, and note that 
    \[\rho^+ = \lambda_{\max}\left(\Pi\cE^\cH\Pi\right) = \lambda_{\max}\left(\Pi(A^+ - A^-)\Pi\right).\]
    Since $\Pi(A^+ - A^-)\Pi\preceq\Pi A^+\Pi$, and since $\|\Pi A^+\Pi\|_{\operatorname{op}}\leq\|A^+\|_{\operatorname{op}}$, we have $\rho^+\leq\lambda_{\max}(A^+) = \|A^+\|_{\operatorname{op}}$. To bound $\|A^+\|_{\operatorname{op}}$ we aim to apply \cref{lem:random-proj-compression}. Towards that, note that $\Tr(G^+) = \sum_{j\geq 1:\lambda_j\geq 0}\lambda_j\dim(V_j)\leq\sum_{j\geq 1}|\lambda_j|\dim(V_j)$. Since $\alpha < 1/16$ we obtain by \cref{lem:complete-matrix-values} that $\Tr(G^+)\leq e^{9k}$. Also by \cref{lem:complete-matrix-values} we have that $\mu := \|G^+\|_{\operatorname{op}} = \max_{j\geq 1}\lambda_j = \max_{j\geq 2, j\text{ even}}\lambda_j = \Theta(\alpha/n) = \Theta(k^2/n^2)$. Thus write $d:= \Tr(G)/\mu\leq O(e^{9k}/(k^2/n^2))\leq O(e^{9k}n^2)\leq e^{O(k)}n^2$. Finally, since $\langle e_S, G^+ e_S\rangle = \tr(G^+)$ for all $S\in\binom{[n]}{k}$, \cref{lem:random-proj-compression} applies and we obtain that with probability $\geq 1 - O(n^{-200})$, \footnote{Note that we also have $d = \Tr(G^+)/\mu\geq \Omega(k^2/(k^2/n^2))\geq\Omega(n^2)$.}
    \[\|A^+\|_{\operatorname{op}}\leq O(1)\cdot\max\left\{O\left(\frac{\alpha}{n}\right), \frac{e^{9k}}{m}\cdot\log\left(e^{O(k)}n^2\right)\right\}\leq O(1)\cdot\max\left\{\frac{\alpha}{n}, \frac{e^{9k}}{m}\cdot(k + \log n)\right\}\]
    \[\leq O(1)\cdot\max\left\{\frac{\alpha}{n}, \frac{e^{O(k)}\cdot\log(n)}{m}\right\}\implies\rho^+\leq O(1)\cdot\max\left\{\frac{\alpha}{n}, \frac{e^{O(k)}\cdot\log(n)}{m}\right\}.\]
    
    Now write $\tilde{A}:= A^+ - A^-$. Note that $\tilde{A}(S, T) = m^{-1}((-1)^{|S\cap T|} - \lambda_0)$. Now write $\tilde{a}_{ij}:= (-1)^{|S_i\cap S_j|} - \lambda_0$, where recall that $S_1, \ldots, S_m$ are the hyperedges of $\cH$, i.e they are independent samples from $\binom{[n]}{k}$. Note that by the proof of \cref{lem:complete-matrix-values} we know that for $i\neq j$, $\E_\cH \tilde{a}_{ij} = 0, \E_\cH \tilde{a}_{ij}^2 = 1 - \lambda_0^2$. Also note that $\tilde{a}_{ii} = 1 - \lambda_0$ with probability $1$. Now suppose we have $i, j, i', j'$ with $i\neq j, i'\neq j'$. Then if $\#\{i, j, i', j'\} = 4$ then $\tilde{a}_{ij}, \tilde{a}_{i'j'}$ are independent mean $0$ random variables, and thus $\E_{\cH}[\tilde{a}_{ij}\tilde{a}_{i'j'}] = 0$. Suppose $\#\{i, j, i', j'\} = 3$. WLOG $i = i'$. Then conditional on $S_i$, $h_{ij}, h_{ij'}$ are still independent, and thus we still have $\E_{\cH}[\tilde{a}_{ij}\tilde{a}_{i'j'}] = 0$. Consequently, $\E_{\cH}[\tilde{a}_{ij}\tilde{a}_{i'j'}] \neq 0$ only if $\{i, j\} = \{i', j'\}$ in which case it equals $1 - \lambda_0^2$.
    
    Note that $\delta = \|(\Pi_{u^\perp})\tilde{A}u\|_2 = \|\tilde{A}u - uu^*\tilde{A}u\|_2$, and thus $\delta^2 = \|\tilde{A}u\|_2^2 - |u^*\tilde{A}u|^2$. Similarly, $q = u^*\cE^\cH u = \lambda_0 + m^{-2}\sum_{1\leq i, j\leq m}\tilde{a}_{ij} = \lambda_0 + m^{-1}\cdot(1 - \lambda_0) + 2m^{-2}\sum_{1\leq i < j\leq m}\tilde{a}_{ij}$. 

    Thus $\E_\cH q = \lambda_0 + m^{-1}\cdot(1 - \lambda_0)$. Furthermore, 
    \[\E_\cH\left[\left(q - \lambda_0 - \frac{1 - \lambda_0}{m}\right)^2\right] = \frac{4}{m^4}\sum_{\substack{1\leq i < j\leq m\\1\leq i' < j'\leq m}}\E_{\cH}[\tilde{a}_{ij}\tilde{a}_{i'j'}] = \frac{4}{m^4}\sum_{\substack{1\leq i < j\leq m}}\E_{\cH}[\tilde{a}_{ij}^2] = \frac{4}{m^4}\binom{m}{2}(1 - \lambda_0^2)\]
    \[\implies\E_\cH\left[\left(q - \lambda_0\right)^2\right] = \frac{(1 - \lambda_0)^2}{m^2} + \frac{2(m - 1)}{m^3}(1 - \lambda_0^2)\leq\frac{4(1 - \lambda_0)}{m^2}\overset{\text{\cref{lem:complete-matrix-values}}}{\leq}\frac{8\alpha}{m^2}.\]
    Thus if we write $\zeta:= 1/(100\log n)$, then by Markov's inequality, with probability $\geq 1 - \zeta$, we have $|q - \lambda_0| \leq O(m^{-1}\sqrt{\alpha/\zeta})$.

   Now, note that $\E_\cH\delta^2\leq\E_\cH\|\tilde{A}u\|_2^2$. Thus we estimate $\E_\cH\|\tilde{A}u\|_2^2$. Towards that, note that
   \[\E_\cH\|\tilde{A}u\|_2^2 = \sum_{i = 1}^m\E_\cH(\tilde{A}u)_i^2 = \frac{1}{m^3}\sum_{i = 1}^m\E_\cH\left[\left(\sum_{j = 1}^m \tilde{a}_{ij}\right)^2\right].\]
   Fix any $i$ and any choice for $S_i$. Write $\eta:= 1 - \lambda_0$. Then 
   \[\E_\cH\left[\left(\sum_{j = 1}^m \tilde{a}_{ij}\right)^2\mid S_i\right] = \E_\cH\left[\left(\eta + \sum_{j\neq i} \tilde{a}_{ij}\right)^2\mid S_i\right] = \eta^2 + 2\eta\cdot\E_\cH\left[\sum_{j\neq i} \tilde{a}_{ij}\mid S_i\right] + \E_\cH\left[\left(\sum_{j\neq i} \tilde{a}_{ij}\right)^2\mid S_i\right]\]
   \[ = \eta^2 + \E_\cH\left[\left(\sum_{j\neq i} \tilde{a}_{ij}\right)^2\mid S_i\right] = \eta^2 + \E_\cH\left[\sum_{j\neq i} \tilde{a}_{ij}^2\mid S_i\right] + 2\cdot\E_\cH\left[\sum_{\substack{j, j'\in[m]\setminus\{i\}\\j < j'}} \tilde{a}_{ij}\tilde{a}_{ij'}\mid S_i\right]\]
   \[ = \eta^2 + (m - 1)(1 - \lambda_0^2),\]
   where we use the fact that $\E[\tilde{a}_{ij}\mid S_i] = \E[\tilde{a}_{ij}\tilde{a}_{ij'}\mid S_i] = 0$ for $j, j'\neq i, j\neq j'$. Consequently 
   \[\E_\cH\|\tilde{A}u\|_2^2 = \frac{(1 - \lambda_0)^2 + (m - 1)(1 - \lambda_0^2)}{m^2}\leq \frac{2(1 - \lambda_0)}{m}\leq\frac{4\alpha}{m}\implies\E_\cH\delta^2\leq\frac{4\alpha}{m}.\]
   As above, by Markov's inequality, with probability $\geq 1 - \zeta$ we have $\delta\leq O(\sqrt{\alpha/(m\zeta)})$. 

   Consequently, with probability $\geq 1 - O(1/\log n)$, we have that 
   \[\rho^+ + \delta\leq O\left(\frac{\alpha}{n}\right) + \frac{e^{O(k)}\cdot\log(n)}{m} + O\left(\sqrt{\frac{\alpha\log n}{m}}\right)\leq O\left(\frac{\alpha}{n}\right) + \frac{e^{-\Omega(k)}}{n} + O\left(\frac{k}{2^{Ck/2}n}\right)\]
   \[\leq O(1)\cdot\left(\underbrace{\frac{\alpha}{n} + \frac{e^{-\Omega(k)}}{n}}_{:= \xi/n}\right),\]
   where we use the fact that $m\geq 2^{Ck}n\log(n)$ for some sufficiently large absolute constant $C > 0$. With probability $\geq 1 - O(1/\log n)$ we also have that 
   \[q - \rho^+\geq \lambda_0 - O\left(\frac{\sqrt{\alpha\log(n)}}{m}\right) - O\left(\frac{\alpha}{n}\right) - \frac{e^{-\Omega(k)}}{n}\overset{\text{\cref{lem:complete-matrix-values}}}{\geq}\frac{7}{8} - O\left(\frac{1}{n}\right) > 0\]
   for sufficiently large $n$. Consequently, $(q - \rho^+)_+ = q - \rho^+$. Furthermore, 
   \[1 - (q - \rho^+)_+ = 1 - (q - \rho^+)\geq 1 - q\geq 1 - \lambda_0 - O\left(\frac{\sqrt{\alpha\log(n)}}{m}\right)\geq 2(\alpha - \alpha^2) - O\left(\frac{e^{-\Omega(k)}}{\sqrt{n^3\log(n)}}\right) > 0\]
   for sufficiently large $n$, since $\alpha\geq 4/n$.

   Consequently, as in the proof of \cref{thm:syk-spectral-norm}, if $H^\cH_{\operatorname{SYK}} = \frac{1}{\sqrt{m}}\sum_{S\in\cH}g_S\Gamma_S$ is the SYK Hamiltonian corresponding to the hypergraph $\cH$, then we have (by \cref{lem:varphi-spectral-dom})
   \[\E_{\{g_S\}}\|H^\cH_{\operatorname{SYK}}\|_{\operatorname{op}}\leq 2\cdot 2^{n/(4\ell)}\cdot\sqrt{\frac{1 + (\rho^+ + \delta)\ell}{1 - (q - \rho^+)_+}}\leq 2\cdot 2^{n/(4\ell)}\cdot\sqrt{\frac{1 + (\rho^+ + \delta)\ell}{1 - q}}. \]
   Thus with probability $\geq 1 - O(1/\log n)$ over the randomness of $\cH$ we have that 
   \[\E_{\{g_S\}}\|H^\cH_{\operatorname{SYK}}\|_{\operatorname{op}}\leq 2\cdot 2^{n/(4\ell)}\cdot\sqrt{\frac{1 + \xi\ell/n}{1 - q}}\leq 2\cdot 2^{n/(4\ell)}\cdot\sqrt{\frac{1 + \xi\ell/n}{2(\alpha - \alpha^2) - O\left(\frac{e^{-\Omega(k)}}{\sqrt{n^3\log(n)}}\right)}}\]
   \[\leq 2^{n/(4\ell)}\cdot\sqrt{\frac{2 + O(\xi\ell/n)}{(\alpha - \alpha^2) - \alpha\cdot O(e^{-\Omega(k)}/\sqrt{n\log(n)})}}\]
   Now set $\ell:= \lceil n/\sqrt{\xi}\rceil$ to obtain that 
   \[\E\|H^\cH_{\operatorname{SYK}}\|_{\operatorname{op}}\leq\sqrt{\frac{2}{\alpha - \alpha^2 - o_n(\alpha)}}\cdot(1 + O(\sqrt{\xi}))\leq \sqrt{\frac{2}{\alpha}}\cdot(1 + O(\sqrt{\xi}))\cdot\left(1 + O\left(\alpha + \frac{e^{-\Omega(k)}}{\sqrt{n\log(n)}}\right)\right).\]
   Note that $\xi\leq O(\alpha + e^{-\Omega(k)})$, and thus 
   \[\E\|H^\cH_{\operatorname{SYK}}\|_{\operatorname{op}}\leq \sqrt{\frac{2}{\alpha}}\cdot(1 + O(\sqrt{\alpha + e^{-\Omega(k)}}))\cdot\left(1 + O\left(\alpha + \frac{e^{-\Omega(k)}}{\sqrt{n\log(n)}}\right)\right)\] 
   \[\leq\sqrt{\frac{2}{\alpha}}\cdot\left(1 + O\left(\sqrt{\alpha + e^{-\Omega(k)}}\right)\right),\]
   as desired.
\end{proof}
Finally, we prove the corresponding lower bound for the sparse SYK.
\begin{theorem}[Operator Norm of the Sparse SYK: Lower Bound]\label{thm:syk-spectral-norm-lower-random}
    Let $n\geq k\geq 2$ be even integers, and let $\alpha:= k^2/n < 1/16$. Consider the SYK Hamiltonian (from \cref{eq:syk-def}) for a random hypergraph $\cH$ with $m\geq 2^{Ck}n\log(n)$ hyperedges for some large enough absolute constant $C > 0$. Then 
    \[\E\|H\|_{\operatorname{op}}\geq
\sqrt{\frac2\alpha}
\left(
1-O(1)\max\left\{
k^{-1/4},
\alpha^2
\right\}
\right).\]
with probability $\geq 1-O(1/\log n)$ over the randomness of $\cH$.
\end{theorem}
\begin{proof}
Throughout the proof we'll borrow notation from \cref{lem:x-large-eig}. By \cref{lem:lower-bound-syk-general} we know that if $B\geq\Omega(1)$, then
\begin{equation}
\label{eq:random-lower-transfer}
\E_{\bm g}\left\|H^\cH_{\operatorname{SYK}}\right\|_{\operatorname{op}}\geq\lambda_{\max}\!\left(\Pi_\cR\sfx\Pi_\cR\right)-O(1)\sqrt{\frac{L(\rho^++\delta)}{1-q}}.
\end{equation}

Note that for any unit vector $v\in\R^\cH$ such that $\langle v, u\rangle = 0$, we have
\begin{equation*}
-\langle v,\mathcal E^\cH v\rangle = \langle v,A^-v\rangle-\langle v,A^+v\rangle \leq \langle v,A^-v\rangle \leq \|A^-\|_{\operatorname{op}}\implies\rho^-\leq \|A^-\|_{\operatorname{op}}.
\end{equation*}

We now apply \cref{lem:random-proj-compression} to $G^-$. By
\cref{lem:complete-matrix-values},
\[\|G^-\|_{\operatorname{op}} = \max_{\substack{j\geq1\\\lambda_j<0}}|\lambda_j| = |\lambda_1| \leq O\left(\sqrt{\frac{\alpha}{n}}\right) = O\left(\frac{k}{n}\right),\;\Tr(G^-) = \sum_{\substack{j\geq1\\\lambda_j<0}}|\lambda_j|\dim(V_j)\leq e^{9k}.\]
Moreover, by the $\mathfrak{S}_n$-invariance of the Johnson eigenspaces,
the diagonal entries of $G^-$ are all equal. Hence
\cref{lem:random-proj-compression} gives, with probability at least
$1-O(n^{-100})$,
\[\|A^-\|_{\operatorname{op}}\leq O(1)\max\left\{\frac{k}{n}, \frac{e^{9k}}{m}\log\left(e^{O(k)}n^2\right)\right\}\leq O(1)\max\left\{\frac{k}{n}, \frac{e^{O(k)}(k+\log n)}{m}\right\}\leq O\left(\frac{k}{n}\right).\]
where the last inequality follows from
$m\geq2^{Ck}n\log n$ after choosing $C$ sufficiently large. Consequently we have
\[\rho^-\leq O\left(\frac{k}{n}\right) = O\left(\sqrt{\frac{\alpha}{n}}\right).\]

As proved in \cref{thm:syk-spectral-norm-random}, with probability at
least $1-O(1/\log n)$,
\[\delta\leq O\left(\sqrt{\frac{\alpha\log n}{m}}\right)\leq O\left(\frac{e^{-\Omega(k)}}{n}\right),\qquad \rho^++\delta\leq O(1)\cdot\frac{\alpha+e^{-\Omega(k)}}{n}\leq O\left(\frac{k}{n}\right),\]
\begin{align*}
(2 + o_n(1))\alpha\geq 2\alpha + O\!\left(
\frac{e^{-\Omega(k)}}{\sqrt{n^3\log n}}
\right)\geq\nu = 1-q
&\geq
2(\alpha-\alpha^2)
-
O\!\left(
\frac{e^{-\Omega(k)}}{\sqrt{n^3\log n}}
\right)
\geq \Omega(\alpha).
\end{align*}
Consequently we have $\nu = \Theta(\alpha)$. Thus
\[
B
:=
\sqrt{\frac{\nu}{\rho^- + \delta}}
 \geq
\Omega\!\left(
\sqrt{\frac{\alpha}{k/n}}
\right)
\geq
\Omega(\sqrt{k})
\geq\Omega(1).
\]
Hence the hypothesis needed by \cref{lem:x-large-eig} is satisfied. Now choose $\kappa = \Theta(\sqrt{k})\leq B$, and note that this choice yields 
\[L\leq O\left(\frac{\kappa}{\nu}\right)\leq O\!\left(\frac{n}{k^{3/2}}\right).\]

Thus
\begin{align*}
\frac{L(\rho^++\delta)}{1-q}
&\leq
O\!\left(\frac{n}{k^{3/2}}\right)
\frac{(\alpha+e^{-\Omega(k)})/n}{\alpha}
 =
O\!\left(
\frac{\alpha+e^{-\Omega(k)}}{\alpha k^{3/2}}
\right)
\leq
O\!\left(\frac1{\alpha\sqrt{k}}\right),
\end{align*}
where in the last inequality we used
$\alpha+e^{-\Omega(k)}\leq O(k)$ for $k\geq2$. Therefore
\begin{equation}
\label{eq:random-transfer-error-root}
\sqrt{\frac{L(\rho^++\delta)}{1-q}}\leq O\!\left(\frac1{\sqrt{\alpha}\,k^{1/4}}\right) = \sqrt{\frac2\alpha}\,O\!\left(k^{-1/4}\right).
\end{equation}

Finally \cref{lem:x-large-eig} applied with the preceding bounds on $\nu$, $\rho^-+\delta$, $\kappa$, and $L$, gives
\begin{equation}
\label{eq:random-x-large-eigenvalue}
\lambda_{\max}\!\left(\Pi_\cR\sfx\Pi_\cR\right)\geq\sqrt{\frac2\alpha}\left(1-O(1)\max\left\{k^{-1/2},\alpha^2\right\}\right).
\end{equation}
Combining
\eqref{eq:random-lower-transfer},
\eqref{eq:random-transfer-error-root}, and
\eqref{eq:random-x-large-eigenvalue}, and absorbing the transfer error into
the existing $O(k^{-1/4})$ relative error, yields
\[\E_{\bm g}\left\|H^\cH_{\operatorname{SYK}}\right\|_{\operatorname{op}}\geq\sqrt{\frac2\alpha}\left(1-O(1)\max\left\{k^{-1/4},\alpha^2\right\}\right).\]
All the required inequalities hold simultaneously with probability
at least $1-O(1/\log n)$, and thus we are done.
\end{proof}

Finally, as in \cref{thm:syk-subgaussian}, we obtain strong subgaussian concentration for $\|H^\cH_{\syk}\|_{\op}$:

\begin{theorem}[Subgaussian Tails for the Operator Norm for Sparse SYK]\label{thm:syk-subgaussian-sparse}
    Let $n\geq k\geq 2$ be even integers, and let $\alpha:= k^2/n < 1/16$. Consider the SYK Hamiltonian (from \cref{eq:syk-def}) for a random hypergraph $\cH$ with $m\geq 2^{Ck}n\log(n)$ hyperedges for some large enough absolute constant $C > 0$. Then for any $t\geq 0$ we have
    \[\Pr\left(|\|H^\cH_{\syk}\|_{\operatorname{op}} - \E_{\bm g}\|H^\cH_{\syk}\|_{\operatorname{op}}|\geq t\right)\leq 2\cdot\exp\left(-\Omega(kt^2)\right)\]
    with probability $\geq 1 - O(1/\log(n))$ over the randomness of $\cH$.
\end{theorem}
\begin{proof}
    Towards applying \cref{fact:gaussian-lipschitz-conc}, consider the function $\bm g\mapsto \|H(\bm g)\|_{\op} = \|H\|_{\op}$. By \cref{lem:lipschitz-operator} and the proof of \cref{thm:syk-spectral-norm-lower} we know that the Lipschitz constant of this map is $\leq\sqrt{(\rho^+ + \delta)/(1 - q)}\leq O(1/\sqrt{k})$, and we're now done by \cref{fact:gaussian-lipschitz-conc}.
\end{proof}

\section{Improved bounds for Spectral Norms of Quantum Spin Glasses}
Recall \cref{thm:main-pauli-improved-spectral-norm}:
\begin{theorem}[Improved Bounds for Spectral Norm of Quantum Spin Glasses]\label{thm:pauli-improved-spectral-norm}
    For any $n\geq k\geq 1$ with $k\leq\sqrt{n}/2$ we have
    \[\Omega\left(\frac{\sqrt{n}}{k}\right)\leq\E_{\bm g}\|H_{k-\QSG}\|_{\op}\leq O\left(\sqrt{\frac{n}{k}}\right).\]
\end{theorem}

Our proof of \cref{thm:pauli-improved-spectral-norm} will be very similar to the proof of \cref{thm:main-thm-complete}. To that extent, define the sign matrix 
\[\cE:= \frac{1}{m}(\eps_{P, Q})_{P, Q\in\cP_{n, k}}\]
where $\eps_{P, Q} = \eps_{Q, P}\in\{\pm 1\}$ are such that $PQ = \eps_{P,Q}QP$ for all $P, Q\in\cP_{n, k}$. As with \cref{lem:complete-matrix-values}, we record spectral facts about $\cE$ necessary to show \cref{thm:pauli-improved-spectral-norm}. As usual, define $u:= \1/\sqrt{m}, \rho^+:= \lambda_{\max}(\Pi_{u^\perp}\cE\Pi_{u^\perp}), \delta:= \|\Pi_{u^\perp}\cE u\|_2$.

\begin{lemma}[Spectrum of the Pauli sign matrix]\label{lem:pauli-sign-spec}
Let $1\leq k\leq n$ be integers with $k\le \sqrt n/2$. Then $\mathcal E u=q_{n,k}u$, where
\[
    \left|q_{n,k}-\left(1-\frac{4k^2}{3n}\right)\right|
    \le
    \frac{8k^2}{9n}.
\]
In particular, $\delta = 0$. Moreover, $q_{n,k}$ is the simple largest eigenvalue, and
\[
    \rho^+=\frac{2k}{3n}\,q_{n-1,k-1}\leq\frac{2k}{3n}.
\]
\end{lemma}

Once we have \cref{lem:pauli-sign-spec} the proof of \cref{thm:pauli-improved-spectral-norm} follows exactly as the proof of \cref{thm:syk-spectral-norm}:

\begin{proof}[Proof of \cref{thm:pauli-improved-spectral-norm}]
    Observe that the proofs of \cref{prop:varphi-trace-re,lem:varphi-spectral-dom} use nothing specific about the sign matrix of the SYK Hamiltonian beyond the fact that it is Hermitian and all its entries are $\pm 1$-valued. In particular, they apply to our current sign matrix $\cE$ arising from $H_{k-\QSG}$ as well, and thus as in the proof of \cref{thm:syk-spectral-norm} we have, for every $\ell\geq 1$, 
    \[\E_{\bm g}\|H_{k-\QSG}\|_{\op}\leq 2\cdot 2^{n/(2\ell)}\cdot\sqrt{\frac{1 + \rho^+\cdot\ell}{1 - q_{n, k}}},\]
    where $\rho^+, q_{n, k}$ are as in \cref{lem:pauli-sign-spec}. Furthermore by \cref{lem:pauli-sign-spec} we have 
    \[\E_{\bm g}\|H_{k-\QSG}\|_{\op}\leq 2\cdot 2^{n/(2\ell)}\cdot\sqrt{\frac{1 + O(\ell k/n)}{4k^2/(9n)}}.\]
    Putting $\ell = n$ yields the desired result.
\end{proof}

Thus it only remains to show \cref{lem:pauli-sign-spec}. $\cE$ belongs to the Bose-Mesner algebra of the non-binary Johnson scheme~\cite{TarnanenAG85} on the alphabet $\{\Id,X,Y,Z\}$, and \cite{TarnanenAG85} provide a detailed treatment of the eigenvalues of such matrices. We provide a self-contained proof of \cref{lem:pauli-sign-spec} here though, the only external input being \cref{fact:johnson-eigenspaces}. 
\begin{proof}[Proof of \cref{lem:pauli-sign-spec}]
We calculate the eigenvalues of $\cE$ by explicitly describing the eigenspaces of $\cE$.
\paragraph{The constant eigenvector.}
Identify a Pauli string $P\in\cP_{n, k}$ with a pair $(S,a)$, where $S\in\binom{[n]}k$ is its support and $a\in\{X,Y,Z\}^S$ records its nonidentity labels. For $P = (S, a), Q = (T, b)\in\cP_{n, k}$, any any $i\in S\cap T$, $i$ contributes a sign of $+1$ to $\eps_{P, Q}$ if $P_i = Q_i$, and contributes $-1$ otherwise. Consequently, the expected contribution over a uniform second label is $(1-1 - 1)/3 = -1/3$. Thus for any $P\in\cP_{n, k}$ we have 
\[
    \frac1m\sum_{Q\in\mathcal P_{n,k}}\varepsilon(P,Q)
    =\E_{Q\sim\cP_{n, k}} (-1/3)^{|P\cap Q|}=:q_{n,k}.
\]
Note that $q_{n, k}$ is independent of $P$ and thus $\mathcal E u=q_{n,k}u$ and $\delta=0$. 

Now let $J$ be the random variable denoting $|S\cap T|$ for uniformly random $S, T\in\binom{[n]}{k}$. One can verify that for any integer $J\geq 0$ we have
\[
    \left|(-1/3)^J-1+\frac43J\right|
    \le \frac89J(J-1).
\]
Since $\mathbb E[J(J-1)]
    =\frac{k^2(k-1)^2}{n(n-1)}$ we have that for $n\ge2$,
\begin{equation}\label{eq:pauli-row-sum-expansion}
    \left|q_{n, k}-1+\frac{4k^2}{3n}\right|
    \le
    \frac{8k^2(k-1)^2}{9n(n-1)},
\end{equation}
as desired.

\paragraph{Higher Eigenspaces}
Choose a real orthonormal basis $1,\varphi_1,\varphi_2$ of
functions $\{X,Y,Z\}\to\C$, with respect to the uniform measure,
such that $\varphi_1,\varphi_2$ have mean zero.

For $R\subseteq[n]$, $|R|=r\le k$, and
$b\in\{1,2\}^R$, consider the subspace $V_{R, b}\seq\C^{\cP_{n, k}}$ consisting of all functions $F:\cP_{n, k}\to\C$ defined as
\[
    F(S,a)
    =\1(R\seq S)\cdot
      \left(\prod_{i\in R}\varphi_{b_i}(a_i)\right)\cdot
      f(S\setminus R),
\]
where $f:\binom{[n]\setminus R}{k-r}\to\C$ is an arbitrary function. Under the usual inner product on $\C^{\cP_{n, k}}$, $V_{R, b}$ and $V_{R', b'}$ are orthogonal subspaces if $(R, b)\neq (R', b')$, by definition. Furthermore 
\[\sum_{(R, b)}\dim_{\C}(V_{R, b}) =     \sum_{r=0}^k
       \binom nr2^r\binom{n-r}{k-r}
    =\binom nk\sum_{r=0}^k\binom kr2^r
    =3^k\binom nk=m,
\]
and thus $\C^{\cP_{n, k}} = \bigoplus_{(R, b)}V_{R, b}$. 

Now, for $x\in\{X,Y,Z\}$ and $i\in\{1,2\}$, the mean-zero
property of $\varphi_1, \varphi_2$ gives
\[
    \frac13\sum_{y\in\{X,Y,Z\}}
    \bigl(2\cdot \1(x=y)-1\bigr)\varphi_i(y)
    =\frac23\varphi_i(x).
\]
We next show that each $V_{R,b}$ is invariant under $\cE$, i.e. $\cE(V_{R, b})\seq V_{R, b}$. To that extent fix any $F\in V_{R,b}$ and write
\[
    F(T,c)=\1(R\subseteq T)
    \left(\prod_{i\in R}\varphi_{b_i}(c_i)\right)
    f(T\setminus R).
\]
We have to show that $\cE F\in V_{R, b}$. Now, for a fixed string $(S,a)$, the definition of $\cE$ gives
\[
    (\cE F)(S,a)
    =
    \frac1{\binom nk}\sum_{T\in\binom{[n]}k}
    \E_c\left[\eps_{(S,a),(T,c)}F(T,c)\right],
\]
where, conditional on $T$, the labels $(c_i)_{i\in T}$
are independent and uniform on $\{X,Y,Z\}$.
Only $T\supseteq R$ contribute to this sum.

If $R\not\subseteq S$, choose $i\in R\setminus S$.
The string $(S,a)$ is identity at site $i$, so the
commutation sign is independent of $c_i$. Since
$F(T,c)$ contains the mean-zero factor
$\varphi_{b_i}(c_i)$, averaging over $c_i$ gives zero.
Thus $(\cE F)(S,a)=0$ whenever $R\not\subseteq S$.

Now suppose $R\subseteq S$, and fix $T\supseteq R$.
The commutation sign factors over the common occupied sites:
\[
    \eps_{(S,a),(T,c)}
    =
    \prod_{i\in S\cap T}\bigl(2\cdot\1(a_i=c_i)-1\bigr).
\]
For $i\in R$, averaging the corresponding sign together
with the label factor in $F$ gives
\[
    \E_{c_i}\left[
       \bigl(2\cdot\1(a_i=c_i)-1\bigr)\varphi_{b_i}(c_i)
    \right]
    =\frac23\varphi_{b_i}(a_i).
\]
For $i\in(S\cap T)\setminus R$, there is no label factor
from $F$, and
\[
    \E_{c_i}\bigl[2\cdot\1(a_i=c_i)-1\bigr]=-\frac13.
\]
All other sites contribute $1$. Independence of the labels
therefore yields
\[
    \E_c\left[\eps_{(S,a),(T,c)}F(T,c)\right]
    =
    \left(\frac23\right)^r
    \left(\prod_{i\in R}\varphi_{b_i}(a_i)\right)
    \left(-\frac13\right)^{|(S\setminus R)\cap(T\setminus R)|}
    f(T\setminus R).
\]
Consequently,
\[
    (\cE F)(S,a)
    =
    \1(R\subseteq S)
    \left(\prod_{i\in R}\varphi_{b_i}(a_i)\right)
    g(S\setminus R)\in V_{R, b},
\]
where
\[
    g(A)
    :=
    \frac{(2/3)^r}{\binom nk}
    \sum_{B\in\binom{[n]\setminus R}{k-r}}
    \left(-\frac13\right)^{|A\cap B|}f(B).
\]
Thus $V_{R,b}$ is $\cE$-invariant. Moreover, if we identify $V_{R,b}$ with $\C^{\binom{[n]\setminus R}{k-r}}$,
the induced matrix is
\[
    \cE\vert_{V_{R, b}} = \frac{(2/3)^r}{\binom nk}\,T_r\in\C^{\binom{[n]\setminus R}{k-r}\times\binom{[n]\setminus R}{k-r}},
    \qquad
    T_r(A,B):=(-1/3)^{|A\cap B|}.
\]
Writing $M_s(A,B):=\binom{|A\cap B|}{s}$, we have $ T_r=\sum_{s=0}^{k-r}\left(-\frac43\right)^sM_s$. Since $2(k-r)\le n-r$, we may apply
\cref{fact:johnson-eigenspaces} on the $(k-r)$-subsets
of $[n]\setminus R$. It follows that the eigenvalues of
$\cE$, as a multiset, are
\begin{equation}\label{eq:pauli-exact-spectrum}
    \lambda_{r,j}
    =
    \frac{(2/3)^r}{\binom nk}
    \sum_{s=j}^{k-r}
       \left(-\frac43\right)^s
       \binom{k-r-j}{s-j}
       \binom{n-r-s-j}{k-r-s},
    \qquad
    \substack{0\le r\le k,\\0\le j\le k-r},
\end{equation}
with corresponding multiplicities
\[
    d_{r,j}
    =
    \binom nr2^r
    \left[
       \binom{n-r}{j}-\binom{n-r}{j-1}
    \right].
\]
The mode $(r,j)=(0,0)$ is the one-dimensional constant
space and has eigenvalue $q_{n,k}$.

\paragraph{Estimating the eigenvalues.}
We use the same definition of $q_{n,k}$ when $k=0$, so
$q_{n,0}=1$. As in the proof of
\cref{lem:complete-matrix-values}, the intersection variable
$J$ satisfies
\[
    \E\binom Jh
    =
    \binom kh\frac{\binom{n-h}{k-h}}{\binom nk}.
\]
Expanding $(-1/3)^J=(1-4/3)^J$ therefore gives
\[
    \binom nk q_{n,k}
    =
    \sum_{h=0}^k
       \left(-\frac43\right)^h
       \binom kh\binom{n-h}{k-h}.
\]
Substituting $s=j+h$ into
\eqref{eq:pauli-exact-spectrum} and applying this identity
with parameters $n-r-2j$ and $k-r-j$, we obtain
\begin{equation}\label{eq:pauli-factored-spectrum}
    \lambda_{r,j}
    =
    \left(\frac23\right)^r
    \left(-\frac43\right)^j
    \frac{\binom{n-r-2j}{k-r-j}}{\binom nk}
    q_{n-r-2j,k-r-j}.
\end{equation}

For every integer $J\ge0$,
\[
    0\le1-(-1/3)^J\le\frac43J.
\]
Taking expectations yields the general bound
\begin{equation}\label{eq:pauli-row-sum-bounds}
    1-\frac{4k^2}{3n}\le q_{n,k}\le1.
\end{equation}
Write $t=r+j$. Since $n\ge4k^2$ and $0\le t\le k$,
\[
    \frac{(k-r-j)^2}{n-r-2j}
    \le\frac{(k-t)^2}{n-2t}
    \le\frac{k^2}{n}.
\]
The last inequality follows using $n\ge2k\ge 2t$. Thus \eqref{eq:pauli-row-sum-bounds} gives $\frac23\le q_{n-r-2j,k-r-j}\le1$.  In particular, $\lambda_{r,j}$ has sign $(-1)^j$.

The binomial ratio in \eqref{eq:pauli-factored-spectrum}
is the probability that a uniform $k$-subset of $[n]$
contains a prescribed set of $r+j$ sites and avoids a
disjoint prescribed set of $j$ sites. Therefore
\[
    \frac{\binom{n-r-2j}{k-r-j}}{\binom nk}
    \le\frac{(k)_{r+j}}{(n)_{r+j}}
    \le\left(\frac kn\right)^{r+j},
\]
and hence
\begin{equation}\label{eq:pauli-mode-bound}
    |\lambda_{r,j}|
    \le\left(\frac{4k}{3n}\right)^{r+j}.
\end{equation}
Every nonconstant eigenvalue consequently has absolute
value at most $4k/(3n)\le1/3$, whereas $q_{n,k}\ge2/3$.
Thus $q_{n,k}$ is the simple largest eigenvalue.

The two modes with $r+j=1$ satisfy
\[
    \lambda_{1,0}
    =\frac{2k}{3n}q_{n-1,k-1}
    \ge\frac{4k}{9n}>0,
    \qquad
    \lambda_{0,1}<0.
\]
If $r+j\ge2$, then $k\ge2$, and
\eqref{eq:pauli-mode-bound} gives
\[
    |\lambda_{r,j}|
    \le\frac{16k^2}{9n^2}
    <\frac{4k}{9n}
    \le\lambda_{1,0},
\]
because $n\ge4k^2>4k$. For $k=1$, there are no modes
with $r+j\ge2$. Consequently,
\[
    \rho^+
    =\lambda_{1,0}
    =\frac{2k}{3n}q_{n-1,k-1}
    \le\frac{2k}{3n},
\]
as desired.
\end{proof}

\section{Acknowledgments}

S.M. acknowledges support from a Brown Investigator Award, a program of the Brown Institute for Basic Sciences at the California Institute of Technology. S.M. thanks Dima Abanin for helpful discussions on quantum simulation and related collaborations. We thank Bobak Kiani, David Gamarnik and Eric Anschuetz for helpful comments about the first version of this pre-print. We also thank David Gamarnik for pointing us out to \cite{AnschuetzGK25}, and remarking that our techniques could also be useful in improving spectral norm bounds for quantum spin glasses.

\paragraph{Use of AI Tools:} We observed (see \cref{sec:baby-johnson} for a key observation) that there is an explicit deterministic operator with a Johnson-like structure that captures the expected trace moments of $H_{\syk}$. We analyzed a preliminary version of this deterministic operator with the help of GPT-5.5 Pro and GPT-5.6 Pro. In particular, GPT-5.6 Pro suggested \cref{lem:twisted-boson-mapping}, \cref{prop:isometry}, and \cref{lem:intertwining} to us. We used GPT-5.6 Pro to search for related literature, assist with the verification of our technical proofs, and we used GPT-6-Pro to assist in writing the proof of \cref{thm:pauli-improved-spectral-norm}.

\bibliography{syk-spectral}
\bibliographystyle{alphaurl}
\end{document}